\documentclass[letterpaper,11pt]{article} 
\usepackage{amsmath, amssymb, amsthm}
\usepackage[margin=1in]{geometry}     
\usepackage{hyperref}
\hypersetup{hypertexnames=false}
\usepackage{booktabs}
\usepackage{float}
\usepackage{algorithm}
\usepackage{algpseudocode}
\usepackage{tikz}
\usepackage{graphicx}
\usepackage{indentfirst}
\usepackage{multirow}

\newtheorem{theorem}{Theorem}[section]
\newtheorem{lemma}[theorem]{Lemma}
\newtheorem{proposition}[theorem]{Proposition}
\newtheorem{corollary}[theorem]{Corollary}
\newtheorem{definition}[theorem]{Definition}

\newtheorem{remark}[theorem]{Remark}

\newcommand{\softO}{\widetilde{O}}

\begin{document}
\title{Improved Quantum Algorithms for Black-Box Abelian Group Decomposition}
\author{%
  Junrong Luo\thanks{School of Mathematics and Statistics,
  Wuhan University.
  Email: \href{mailto:junrong_luo@whu.edu.cn}{\texttt{junrong\_luo@whu.edu.cn}}.}
  \and
  Yinan Li\thanks{School of Artificial Intelligence,
  Wuhan University.
  Email: \href{mailto:Yinan.Li@whu.edu.cn}{\texttt{Yinan.Li@whu.edu.cn}}.}
  \and
  Fran\c{c}ois Le Gall\thanks{Graduate School of Mathematics,
  Nagoya University.
  Email: \href{mailto:legall@math.nagoya-u.ac.jp}{\texttt{legall@math.nagoya-u.ac.jp}}.}
}
\date{}
\maketitle
\begin{abstract}
Decomposing finite Abelian black-box groups into cyclic
factors is a basic problem in quantum computation, with
applications to the Abelian hidden subgroup problem.
Cheung and Mosca (\emph{Quantum Inf. Comput.}, 2001)
developed a polynomial-time quantum algorithm for this problem. Recently, Regev (\emph{J.~ACM}, 2025) proposed a factoring algorithm
that uses smaller quantum circuits than Shor's algorithm,
assuming a number-theoretic conjecture.
Ragavan and Vaikuntanathan (\emph{CRYPTO}, 2024)
reduced its quantum space requirements.
Pilatte (\emph{Forum Math. Pi}, 2026) proved a version
of the conjecture and established unconditional correctness
for modified versions of these algorithms.

We give a quantum algorithm for decomposing finite Abelian
black-box groups by adapting the quantum sampling and classical
lattice-reduction method of Regev's factoring algorithm.
The algorithm computes an invariant-factor decomposition and
corresponding cyclic generators with high probability.
For a group $G$ of order at most $2^n$ with unique encodings
and reversible group-operation cost
$T_{\mathrm{op}}=\Omega(\sqrt n)$, our algorithm uses
$O(\sqrt n)$ quantum circuits, each with
$\widetilde O(nT_{\mathrm{op}})$ gates and executed at most $O(n)$ times.
For comparison, we consider the algorithm of Cheung and
Mosca (Quantum Inf. Comput. 1(3), 26--32, 2001), which
decomposes the Sylow $p$-subgroups separately.
We also consider the common-modulus implementation of the
extended Cheung--Mosca algorithm of Bermejo-Vega, Lin, and
Van den Nest.
Our algorithm reduces the sum of the circuit gate counts from
$\widetilde O(n^2T_{\mathrm{op}})$ to
$\widetilde O(n^{3/2}T_{\mathrm{op}})$.
The total quantum time complexity bound reduces from
$\widetilde O(n^3T_{\mathrm{op}})$ to
$\widetilde O(n^{5/2}T_{\mathrm{op}})$, and the quantum space
bound reduces from $O(n^2)$ to $O(n)$ qubits.
The classical computation uses polynomially many bit operations
and group-operation queries. These improvements rely on two technical ingredients. We prove that the integer relation lattices used in our
algorithm admit, with high probability, integral bases
whose vectors have Euclidean norm at most $e^{O(\sqrt n)}$.
We preserve the circuit-size advantage of Regev's algorithm
for group decomposition by incorporating all $O(n)$ sampled
generators while keeping the lattice-reduction dimension
at $O(\sqrt n)$.
\end{abstract}

\section{Introduction}\label{sec:introduction}

\subsection{Background}
\label{sec:background}

\subsubsection*{Quantum algorithms for arithmetic and hidden subgroup problems}

Shor developed polynomial-time quantum algorithms for integer factoring and discrete logarithms~\cite{Shor97}. For factoring an integer \(N\), the quantum subroutine determines the order of an element \(a\in(\mathbb Z/N\mathbb Z)^\times\) by period finding for the function \(x\mapsto a^x\bmod N\). After the quantum Fourier transform, the measurement outcomes  are processed using continued fractions to recover the order. For discrete logarithms, Shor uses modular exponentiation and quantum Fourier transforms to obtain relations from which the discrete logarithm is recovered. Kitaev developed a polynomial-time quantum algorithm for the Abelian stabilizer problem, which includes order finding and discrete logarithms as special cases. And his method is based on eigenvalue measurement~\cite{Kitaev95}. Mosca and Ekert subsequently described the hidden subgroup problem as a more general formulation of order finding, discrete logarithms, and the Abelian stabilizer problem. And they showed that the Abelian hidden subgroup problem can also be described and analysed in terms of eigenvalue estimation~\cite{mosca1999hidden}. Jozsa further reviewed Shor's factoring and discrete-logarithm algorithms and drew out the Abelian hidden subgroup problem as their unifying generalization~\cite{jozsa2001quantum}.

 Hallgren developed quantum algorithms for periodic functions over the reals with irrational periods and applied them to Pell's equation and the principal ideal problem in real quadratic number fields~\cite{hallgren2007pell}. Then he treated constant-degree number fields, where under the logarithmic embedding the unit group forms a real-valued lattice; computing the unit group is reduced to finding a basis of this lattice, using Fourier sampling to obtain information from its dual lattice~\cite{Hallgren2005}. Schmidt and Vollmer independently obtained a quantum algorithm for the irrational period lattice of a function on \(\mathbb Z^n\) and applied it to unit-group computation in fixed-degree number fields~\cite{schmidtVollmer2005unit}. For number fields of arbitrary degree, Eisenträger, Hallgren, Kitaev, and Song removed the constant-degree restriction by extending the hidden subgroup problem to continuous groups and giving a quantum algorithm for the HSP over \(\mathbb R^n\). Their construction uses Gaussian-weighted superpositions of lattice points to represent a real-valued lattice. The algorithm for computing the unit group is polynomial in the field degree and the logarithm of its discriminant~\cite{eisentragerHallgrenKitaevSong2014unit}.
 Biasse and Song then reduced class-group computation and the principal ideal problem in arbitrary-degree number fields to \(S\)-unit-group computation, and reduced the latter to the continuous hidden subgroup problem introduced by Eisenträger, Hallgren, Kitaev, and Song~\cite{eisentragerHallgrenKitaevSong2014unit,biasseSong2025Sunits}. Later work supplied a detailed treatment of the \(S\)-unit algorithm~\cite{biasseSong2025Sunits}. De Boer and Felderhoff gave explicit polynomial bounds for the quantum gate and memory complexity under the Generalized Riemann Hypothesis~\cite{deBoerFelderhoff2026}.

Regev~\cite{Regev25} proposed a multidimensional variant of
Shor's factoring algorithm with asymptotically smaller quantum gate complexity. For an $n$-bit integer, the algorithm uses $\softO(n^{3/2})$
quantum gates per circuit execution, improving the $\softO(n^2)$ quantum gates
in Shor's algorithm~\cite{Shor97}. The main idea behind Regev's quantum algorithm is to design a lattice for implementing the modular exponentiation and recover the period via classical LLL post-processing.
The original implementation of Regev's algorithm requires $O(n^{3/2})$ qubits, and its
correctness and complexity analysis depends on a number-theoretic assumption. Ragavan and Vaikuntanathan improve the space efficiency and noise tolerance of Regev's algorithm. They used reversible exponentiation based on Fibonacci numbers to reduce the quantum space to \(O(n\log n)\) qubits while using \(O(n^{3/2}\log n)\) gates~\cite{Ragavan26}. And they also modified the classical post-processing so that it tolerates a constant fraction of corrupted runs. Ragavan and Vaikuntanathan further replaces Fibonacci numbers by more general sequences defined through linear recurrences, obtaining constant-factor improvements in space and/or circuit size~\cite{Ragavan26}. Pilatte proves a version of Regev's number-theoretic conjecture using tools from analytic number theory, including zero-density estimates~\cite{Pilatte26}. His results establish unconditional correctness for slightly
modified versions of Regev's factoring algorithm and its
extensions to discrete logarithms in multiplicative groups
modulo integers. 
The resulting factoring algorithm uses
$O(n^{3/2}\log^3 n)$ gates and $O(n\log^3 n)$ qubits per
circuit execution, with $O(\sqrt n)$ executions. Eker\aa{} and G\"artner extended method to discrete logarithms and discuss extensions to order finding and complete factorization~\cite{Ekera24}. Further work studies extensions to discrete logarithms on
elliptic curves and hyperelliptic
Jacobians~\cite{BarbulescuBarcauPasol2026,BarbulescuBisson2026}.

Beyond Abelian groups, Ettinger and Høyer studied the HSP for non-abelian groups and gave an algorithm for the dihedral case using a linear number of oracle calls~\cite{EttingerHoyer2000}. Ettinger, Høyer, and Knill then showed that the hidden subgroup of an arbitrary finite group can be identified with polynomially many oracle calls, although their algorithm requires exponential time~\cite{EttingerHoyerKnill1999,Ettinger04}. Kuperberg gave a subexponential-time algorithm for the dihedral HSP with time and query complexity \(2^{O(\sqrt{\log N})}\); the algorithm also applies to the hidden-shift problem over arbitrary finitely generated Abelian groups~\cite{Kuperberg05}. Regev modified Kuperberg's algorithm so that the running time remains subexponential while the space requirement becomes polynomial~\cite{regev2004dihedral}. Childs and van Dam studied a generalized hidden-shift problem whose cases \(M=N\) and \(M=2\) correspond respectively to the Abelian HSP and the dihedral HSP, and gave a polynomial-time quantum algorithm when \(M>N^\varepsilon\) for any fixed \(\varepsilon>0\)~\cite{ChildsVanDam2007}.

\subsubsection*{Quantum algorithms for black-box group computation}

Black-box group computation is an important theoretical model in quantum computing, particularly for studying quantum speedups for algebraic problems. In the black-box group model, group elements are represented by unique strings, while group multiplication and inversion are provided through quantum oracle access.
Even for Abelian groups, this model includes problems such as
order finding, discrete logarithms, and group
decomposition~\cite{Shor97,Cheung01}. 

The Fourier-sampling techniques underlying Shor's algorithms
extend naturally to general finite Abelian
groups~\cite{Shor97,mosca1999hidden}.
The standard implementation of the quantum Fourier transform
over $G$ uses cyclic coordinates from an explicit decomposition
\[
G\cong \mathbb Z_{n_1}\times\cdots\times\mathbb Z_{n_r}.
\]
In the black-box model, this structural information can be
difficult to recover classically.
Cheung and Mosca~\cite{Cheung01} developed a polynomial-time
quantum algorithm for decomposing a finite Abelian black-box
group into a direct product of cyclic subgroups. 
Their algorithm decomposes the Sylow $p$-subgroups separately,
using the Abelian hidden subgroup algorithm and Smith normal form.
It explicitly computes generators of the cyclic factors
and their orders.
Bermejo-Vega, Lin, and Van den Nest~\cite{BermejoVegaLinVanDenNest2014}
presented an extended Cheung--Mosca algorithm.
It uses a common modulus to recover relations among all input
generators and solves linear equations to express these
generators in the resulting cyclic coordinates.
These result extends the applicability of Fourier-based quantum
algorithms to finite Abelian black-box groups. Under
the Generalized Riemann Hypothesis
Cheung and Mosca algorithm can be used
to compute class numbers of imaginary quadratic number fields~\cite{Cheung01}. For the group isomorphism problem, the Abelian case can be solved` efficiently. Le Gall later gave a polynomial-time quantum algorithm in the black-box model for a non-Abelian class consisting of extensions of an Abelian group \(A\) by a cyclic group \(\mathbb Z_m\), under the condition \(\gcd(|A|,m)=1\)~\cite{LeGall10}. Bermejo-Vega, Lin, and Van den Nest subsequently introduced black-box normalizer circuits and showed that finite Abelian group decomposition is complete for the associated complexity class. They also showed that black-box normalizer circuits can be efficiently simulated classically if a subroutine for solving the group-decomposition problem is provided~\cite{BermejoVegaLinVanDenNest2014}.

Watrous gave a polynomial-time quantum algorithm for computing the order of a solvable black-box group with unique encoding~\cite{watrous2001solvable}. Membership testing, testing equality of subgroups, and testing normality of a subgroup reduce to computing orders of solvable groups and therefore also admit polynomial-time quantum algorithms. His algorithm also produces a quantum state approximating the uniform superposition over the elements of a chosen subgroup, which allows existing quantum algorithms for Abelian groups to be applied to Abelian factor groups of solvable groups~\cite{watrous2001solvable}. Ivanyos, Magniez, and Santha subsequently gave quantum implementations of black-box group algorithms of Beals and Babai. For black-box groups with unique encoding, their algorithms solve constructive membership, compute the group order and a presentation, find generators for the center, construct a composition series, and find Sylow subgroups in time polynomial in the input size and \(\nu(G)\), where \(\nu(G)=1\) for solvable groups~\cite{IvanyosMagniezSantha2003}.

For non-Abelian black-box groups, Friedl, Ivanyos, Magniez, Santha, and Sen introduced the Translating Coset problem, which generalizes Hidden Translation and Hidden Subgroup. Their results give efficient quantum algorithms for certain solvable groups and subexponential-time algorithms for arbitrary solvable groups~\cite{FriedlEtAl2014}. Ivanyos, Sanselme, and Santha later gave an efficient quantum algorithm for the Hidden Subgroup Problem in groups of nilpotency class at most two~\cite{IvanyosSanselmeSantha2012}.

For black-box groups, Watrous showed that Group Non-Membership has polynomial-size quantum proofs that can be verified in polynomial time, placing the problem in \(\mathrm{QMA}\)~\cite{Watrous2000}. Aaronson and Kuperberg later gave a \(\mathrm{QCMA}\) protocol under group-theoretic assumptions. Without these assumptions, their protocol uses only polynomially many queries to the group oracle. They conjectured that Group Non-Membership is in \(\mathrm{QCMA}\)~\cite{AaronsonKuperberg2007}. In 2025, Le Gall, Nishimura, and Thakkar proved that Group Order Verification for black-box groups is in \(\mathrm{QCMA}\cap\mathrm{coQCMA}\)~\cite{LeGallNishimuraThakkar2025}. This implies that Group Non-Membership is in \(\mathrm{QCMA}\). They also showed that Group Isomorphism is in \(\mathrm{QCMA}\). Homomorphism, Minimal Normal Subgroup, Proper Subgroup, and Simple Group are in \(\mathrm{QCMA}\cap\mathrm{coQCMA}\), while Intersection, Centralizer, and Maximal Normal Subgroup are in \(\mathrm{coQCMA}\).

A natural question is whether Regev's algorithm can be generalized
to other algebraic problems with improved quantum resource bounds.
For factoring an $n$-bit integer, its circuit uses
$\widetilde O(n^{3/2})$ gates per execution and is executed
$O(\sqrt n)$ times~\cite{Regev25}.
The total quantum time complexity is therefore
$\widetilde O(n^2)$, which matches the asymptotic bound for the
standard modular-exponentiation implementation of Shor's
algorithm~\cite{Shor97}.
This leads to the following question:
\begin{quote}
\itshape
Can Regev's algorithm be generalized to black-box finite Abelian
group decomposition with lower total quantum time complexity
than the Cheung--Mosca algorithm and extended Cheung--Mosca algorithm?
\end{quote}

In this paper, for a finite Abelian black-box group \(G\) with \(|G|\le 2^n\), unique encodings and reversible group-operation cost \(T_{\mathrm{op}}=\Omega(\sqrt n)\), we give a quantum algorithm that computes an invariant-factor decomposition of \(G\) and corresponding generators with high probability using $
\widetilde O(n^{5/2}T_{\mathrm{op}})
$ quantum time complexity and \(O(n)\) qubits. For comparison, we consider the Cheung--Mosca
algorithm~\cite{Cheung01}, which decomposes the Sylow
$p$-subgroups separately and also consider the common-modulus implementation of the
extended Cheung--Mosca
algorithm~\cite{BermejoVegaLinVanDenNest2014}.
Our algorithm reduces the sum of the circuit gate counts from
$\widetilde O(n^2T_{\mathrm{op}})$ to
$\widetilde O(n^{3/2}T_{\mathrm{op}})$.
The total quantum time bound reduces from
$\widetilde O(n^3T_{\mathrm{op}})$ to
$\widetilde O(n^{5/2}T_{\mathrm{op}})$ and the quantum space bound reduces from $O(n^2)$ to $O(n)$
qubits.

\subsection{The Cheung--Mosca Method and Its Extension}
\label{sec:cheung-mosca}

Let \(G\) be a finite Abelian black-box group with \(|G|\leq 2^n\),
unique \(O(n)\)-bit encodings, and an oracle for the group operation.
Write \(r=r(G)\) for its minimum number of generators.
Taking \(t=n+O(\log n)\) independent uniform samples
\(a_1,\ldots,a_t\) gives a generating list with high probability.
The decomposition problem is to find invariant factors
\(2\leq n_1\mid\cdots\mid n_r\) and corresponding cyclic generators.

\paragraph{Cheung--Mosca's algorithm.}
Cheung and Mosca~\cite[Section~4]{Cheung01} reduce group
decomposition to the Abelian hidden subgroup problem over Sylow 
p-subgroups.
They first compute the orders \(o_i=\operatorname{ord}(a_i)\) of the sampled generators and factor these integers using Shor's quantum algorithms~\cite{Shor97}. 
For each prime \(p\mid o_i\), let \(p^e\) be the largest
power of \(p\) dividing \(o_i\), and then
the element \((o_i/p^e)a_i\) has order \(p^e\).
For each nonidentity \(a_i\), B\'ezout's identity guarantees that these \(p\)-power order elements generate the cyclic subgroup \(\langle a_i \rangle\). Since the original elements \(a_i\) generate \(G\), this full collection still generates \(G\).
Collecting its elements of \(p\)-power order constructs
a generating set for the Sylow \(p\)-subgroup \(G_p\).

Set this generating set as \(a_{p,1},\ldots,a_{p,t_p}\),
and set \(Q_p=\max_i\operatorname{ord}(a_{p,i})\).
Every generator order divides \(Q_p\), so the homomorphism
\(f_p:\mathbb Z_{Q_p}^{t_p}\to G_p\) is well defined by
$
  f_p(\mathbf z)=\sum_{i=1}^{t_p}z_i a_{p,i}.
$ It hides the subgroup \(K_p=\ker f_p\).
The Abelian HSP algorithm computes generators of \(K_p\).
Since \(G_p\cong\mathbb Z_{Q_p}^{t_p}/K_p\), these kernel
relations determine an integer presentation of \(G_p\).
Smith normal form computes its cyclic factors and expresses
their generators as integer combinations of the \(a_{p,i}\). The cyclic factors of all the Sylow subgroups form a direct-sum
decomposition of \(G\), since \(G=\bigoplus_p G_p\).
Combining cyclic factors of coprime orders classically produces
the invariant factors and corresponding generators of \(G\).

Let $T_{\mathrm{op}}$ denote the gate cost of a reversible group
operation using $O(n)$ work qubits, and assume
$T_{\mathrm{op}}=\Omega(\sqrt n)$. For the resource bounds below, we use Shor's order-finding procedure to compute the orders of the black-box group elements. Once these orders are available as ordinary integers, their prime factorizations are computed using Pilatte's unconditional Regev-type factoring algorithm~\cite{Pilatte26}. This choice affects the quantum cost of the integer-factorization stage but does not change the group-theoretic steps of the Cheung--Mosca algorithm.
Each order finding circuit uses $\softO(nT_{\mathrm{op}})$ gates and $O(n)$
qubits, with $\softO(1)$ executions.
The order-finding step therefore takes
$\softO(n^2T_{\mathrm{op}})$ quantum time, as we have \(O(n)\) generators.

After the orders \(o_1,\ldots,o_t\) have been computed, their prime factorizations are obtained using Pilatte's unconditional Regev-type factoring algorithm~\cite{Pilatte26}. Before a factoring call, powers of prime factors already found are removed classically. Since every \(o_i\) divides \(|G|\leq2^n\), at most \(O(n)\) successful recursive splits are required over all remaining factors. For an integer of at most \(n\) bits, Pilatte's factoring procedure uses a quantum circuit with
$
O(n^{3/2}\log^3 n)=\softO(n^{3/2})
$
gates and \(O(n\log^3 n)\) qubits, and makes \(O(\sqrt n)\) calls to this circuit for constant success probability. After standard amplification, the integer-factorization stage therefore contributes
$
\softO(n^{3/2})
$
to the sum of the distinct-circuit gate counts and \(\softO(n^2)\) quantum time. The subsequent reduction to Sylow \(p\)-subgroups is unchanged. For a Sylow subgroup \(G_p\), relation sampling uses
$
\softO(t_p\log Q_p\,T_{\mathrm{op}})
$
gates per execution and \(O(n^2)\) qubits, and \(O(r(G_p)+\log n)\) executions suffice to recover the kernel with high probability. Since
$
\sum_p r(G_p)\leq n
$ and
$\sum_p\log_2Q_p\leq n,
$
the Sylow-subgroup computations use \(\softO(n)\) executions and
$
\softO(n^3T_{\mathrm{op}})
$
quantum time.

Consequently, the gate counts of the distinct circuits sum to
$
\softO\!\left(
n^2T_{\mathrm{op}}+n^{5/2}
\right)
=
\softO(n^2T_{\mathrm{op}}),
$
where the last equality uses \(T_{\mathrm{op}}=\Omega(\sqrt n)\). The total quantum time remains
$
\softO(n^3T_{\mathrm{op}}),
$
while the maximum quantum space is \(O(n^2)\) qubits.
\paragraph{The extended Cheung--Mosca algorithm.}
Bermejo-Vega, Lin, and Van den Nest give a formulation that uses
one modulus for the entire group
\cite[Section~5.5, Algorithm~6]{BermejoVegaLinVanDenNest2014}.
They compute \(o_i=\operatorname{ord}(a_i)\) by quantum order finding
and set \(Q=\operatorname{lcm}(o_1,\ldots,o_t)\).
Since \(Qa_i=0_G\) for every \(i\), the homomorphism
\(F_Q:\mathbb Z_Q^t\to G\) is well defined by
$
F_Q(\mathbf z)=\sum_{i=1}^t z_i a_i.
$ The map \(F_Q\) is constant on each coset of
\(K=\ker F_Q\) and takes distinct values on distinct cosets.
Finding generators of \(K\) is therefore an Abelian HSP
over \(\mathbb Z_Q^t\), whose cyclic decomposition is known.
The corresponding integer relation lattice is
\[
\mathcal L_{\mathrm{CM}}
=
\left\{
\mathbf z\in\mathbb Z^t:
\sum_{i=1}^t z_i a_i=0_G
\right\}.
\]
An integer vector belongs to \(\mathcal L_{\mathrm{CM}}\)
precisely when its reduction modulo \(Q\) belongs to \(K\).
Since \(a_1,\ldots,a_t\) generate \(G\),
$
G\cong\mathbb Z^t/\mathcal L_{\mathrm{CM}}.
$
From these relations, classical post-processing based on
Smith normal form computes the invariant factors and
the corresponding cyclic generators in polynomial time.

Under the same resource assumptions, 
each circuit uses $\softO(nT_{\mathrm{op}})$ gates and $O(n)$
qubits, and is executed $\softO(1)$ times.
The order-finding steps take $\softO(n^2T_{\mathrm{op}})$
quantum time. For relation sampling, evaluating $F_Q$ requires
$O(t\log Q)=O(n^2)$ controlled group operations.
Thus the sampling circuit uses \(\softO(n^2T_{\mathrm{op}})\) gates per execution and \(O(n^2)\) qubits, and is executed \(O(n)\) times to recover the kernel with high probability.
The gate counts of all distinct quantum circuits used by the
algorithm sum to $\softO(n^2T_{\mathrm{op}})$.
Including all executions, the total quantum time complexity is \(\softO(n^3T_{\mathrm{op}})\), and the algorithm uses \(O(n^2)\) qubits.

\subsection{Regev-Type Approach to Group Decomposition}
\label{sec:direct-regev-approach}

We consider applying Regev's method~\cite{Regev25} to the integer
relation lattice in the extended Cheung--Mosca calgorithm~\cite{BermejoVegaLinVanDenNest2014}.
Let \(H=\langle g_1,\ldots,g_d\rangle\), where \(g_i\in G\), and define \(\Phi:\mathbb Z^d\to G\) by
$
  \Phi(\mathbf z)=\sum_{i=1}^d z_i g_i.
$
Its kernel \(\mathcal L=\ker\Phi\) is the relation lattice, and \(H\cong\mathbb Z^d/\mathcal L\).
When elements $\{g_i\}$ are the full generating list from
Section~\ref{sec:cheung-mosca}, we have $H=G$ and
$\mathcal L=\mathcal L_{\mathrm{CM}}$.

Regev's quantum procedure prepares an approximation to the state
\[
|\psi_0\rangle\propto
\sum_{\mathbf z\in\{-D/2,\ldots,D/2-1\}^d}
\rho_R(\mathbf z)|\mathbf z\rangle|0\rangle,
\]
where $\rho_R(\mathbf z)=\exp(-\pi\|\mathbf z\|_2^2/R^2)$
is the Gaussian amplitude with width $R>0$, and $D$ is a
power of two.
The procedure then evaluates $\Phi$ in the second register,
mapping $|\mathbf z\rangle|0\rangle$ to
$|\mathbf z\rangle|\Phi(\mathbf z)\rangle$. Measuring that register leaves a superposition supported
on a coset of $\mathcal L$ within the box. Fourier transforms of
the coefficient registers then produce approximations to uniformly
sampled classes in $\mathcal L^*/\mathbb Z^d$~\cite{Regev25}.
The evaluation does not require the orders of the $g_i$.
Ragavan and Vaikuntanathan reduce the space used by Regev's factoring
circuit through reversible exponentiation based on Fibonacci
numbers~\cite{Ragavan26}.

Regev constructs an auxiliary lattice from the approximate
dual samples~\cite[Lemma~4.4]{Regev25}.
Every vector of $\mathcal L$ with norm at most a prescribed
bound $T$ has a short lift to this auxiliary lattice.
His LLL-based procedure selects vectors whose integer span
contains these lifts~\cite[Claim~5.1]{Regev25}.
For sufficiently accurate samples, the selected vectors satisfy
the required norm bound, so their first $d$ coordinates belong to
$\mathcal L$, with constant probability.
We refer to these first $d$ coordinates as the projected vectors.
These projected vectors express every vector of $\mathcal L$
of norm at most $T$ as an integer combination.
If $\mathcal L$ admits an integral basis satisfying this
norm bound, the projected vectors therefore generate
$\mathcal L$.
These relations give an integer presentation of $H$ and
Smith normal form then computes the invariant factors and
corresponding cyclic generators of $H$ in polynomial classical
time~\cite{KannanBachem79}.

Let $T$ bound the Euclidean lengths of an integral basis of
$\mathcal L$.
Following Regev's analysis~\cite{Regev25}, for $d\geq\sqrt n$
a sufficient length for each coefficient register is
$\ell=\Theta(d+n/d+\log T)$.
Applying the same recovery argument with
$d+O(\log n)=O(d)$ independent samples reduces the failure
probability to inverse-polynomial size.
With reversible Fibonacci evaluation~\cite{Ragavan26},
one execution uses $O(n+d\ell)$ qubits and
$\softO(d\ell T_{\mathrm{op}}+d\ell^2+n\ell)$ elementary gates,
where each group operation has cost $T_{\mathrm{op}}$.

To apply this procedure to the full generating list in the
extended Cheung--Mosca calgorithm~\cite{BermejoVegaLinVanDenNest2014}, take
$d=\Theta(n)$.
The general bound $\log T=O(n)$ on an integral basis then permits
$\ell=\Theta(n)$.
The sampling circuit has
$\softO(n^2T_{\mathrm{op}}+n^3)$ gates and is executed $O(n)$ times.
The total quantum time is therefore
$\softO(n^3T_{\mathrm{op}}+n^4)$, using $O(n^2)$ qubits.
The coefficient registers alone contain $\Theta(n^2)$ qubits.
These are upper bounds for the stated implementation.

To preserve the circuit-size advantage of Regev's algorithm,
we take $d=\Theta(\sqrt n)$ and require $\log T=O(\sqrt n)$.
Under the assumption $T_{\mathrm{op}}=\Omega(\sqrt n)$,
each execution uses $\softO(nT_{\mathrm{op}})$ gates and
$O(n)$ qubits.
The circuit is executed $O(\sqrt n)$ times, so the total
quantum time is $\softO(n^{3/2}T_{\mathrm{op}})$.
However, the relation lattice then has only $d$ dimensions,
corresponding to the $d$ selected group elements.
If $r(G)>d$, these elements cannot generate $G$, and the
procedure only decomposes a proper subgroup of $G$.

To apply Regev's algorithm to group decomposition while
preserving this circuit-size advantage, we must prove that
the relevant relation lattices admit integral bases satisfying
the required length bound.
We must also process a generating list of $O(n)$ elements
while keeping each relation lattice $O(\sqrt n)$-dimensional.

\subsection{Techniques and Main Result}
\label{sec:results-and-techniques}

\paragraph{Controlling the lattice dimension.}
We process the sampled list $g_1,\ldots,g_L$ in batches of
$b=O(\sqrt n)$ elements, computing decompositions along
the chain
\[
H_0\subseteq H_1\subseteq\cdots\subseteq
H_t=\langle g_1,\ldots,g_L\rangle,
\]
where $H_0=\langle g_1,\ldots,g_d\rangle$ for
$d=\Theta(\sqrt n)$.
The final subgroup equals $G$ with high probability.
At each extension, the next batch $w_1,\ldots,w_b$ defines
$H_{k+1}=\langle H_k,w_1,\ldots,w_b\rangle$.
We partition the known cyclic factors as
\[
H_k=H^{-}\oplus H^{+}
=
\left(\bigoplus_{i=1}^{r}\langle y_i\rangle\right)
\oplus
\left(\bigoplus_{i=1}^{s}\langle y'_i\rangle\right),
\]
where $\operatorname{ord}(y_i)\leq2^{\sqrt n}$ and
$\operatorname{ord}(y'_i)>2^{\sqrt n}$.
The known cyclic coordinates of $H^{-}$ use $O(n)$ qubits
in total.
Since the product of the factor orders is at most $2^n$,
we have $s\leq\sqrt n$.
We form $d$ auxiliary random elements $h_1,\ldots,h_d$
as integer combinations of $y'_1,\ldots,y'_s$ and the
new batch, recording the coefficients.
Their relation lattice is
\[
\mathcal L
=
\left\{
(\mathbf z,\mathbf u)\in\mathbb Z^d\times\mathbb Z^b:
\sum_{j=1}^{d}z_jh_j
+\sum_{l=1}^{b}u_lw_l=0_G
\right\}.
\]
Its dimension is $d+b=O(\sqrt n)$, independently of the
rank of $H_k$.
The quantum circuit combines Fourier sampling on the
cyclic coordinates of $H^{-}$ with Regev's Gaussian
Fourier sampling~\cite{Regev25} on the integer coordinates.
The joint samples contain cyclic character labels and
approximate phases corresponding to the same character
of $H_{k+1}$.
Classical lattice reduction and character reconstruction
then recover a presentation of $H_{k+1}$, from which
Smith normal form computes its cyclic decomposition.

\paragraph{Short bases for the relation lattices.}
We adapt Pilatte's lattice-point counting
argument~\cite{Pilatte26} to relations among independent
uniform samples from a finite Abelian subgroup and fixed
elements in that subgroup.
Character orthogonality expresses the number of lattice
points in a box as a sum over characters.
We bound the contribution of small-order characters
deterministically.
For the large-order contribution, uniform sampling allows
character orthogonality to supply the required
second-moment estimate.
The factors arising from the fixed elements are bounded
deterministically, so the estimate applies to every fixed
batch.
Comparing the counts at two box sizes establishes the
existence of linearly independent short relations.
A bound in terms of successive minima then implies that
the lattices used here admit integral bases with vector
lengths at most $e^{O(\sqrt n)}$ with high probability.
This bound permits lattice recovery with $O(\sqrt n)$
bits per integer coordinate.

\begin{theorem}[Black-box Abelian group decomposition]
\label{thm:group-decomposition}
Let $G$ be a finite Abelian black-box group with $|G|\leq2^n$, unique encodings and reversible group-operation cost $T_{\mathrm{op}}=\Omega(\sqrt n)$. For all sufficiently large $n$, there is a quantum algorithm with classical post-processing that computes an invariant-factor decomposition of $G$ and corresponding cyclic generators with probability at least $1-O(n^{-3/2})$. The algorithm uses $O(\sqrt n)$ quantum circuits,
each with $\softO(nT_{\mathrm{op}})$ gates.
Their gate counts sum to
$\softO(n^{3/2}T_{\mathrm{op}})$ when each circuit is
counted once.
With each circuit executed $O(n)$ times, the quantum time complexity is
$\softO(n^{5/2}T_{\mathrm{op}})$, and the quantum space
is $O(n)$.
The classical computation uses polynomially many bit
operations and group-operation queries.
\end{theorem}

\section{Preliminaries}
\label{sec:preliminaries}

Throughout the paper, \(G\) denotes a finite Abelian group of order \(N\), written additively unless the application is naturally multiplicative. Unless specified otherwise, \(n\) is an integer size parameter satisfying \(N\leq2^n\). The rank \(r(G)\) is the minimum number of generators of \(G\). All logarithms are base two unless written as \(\ln\), and \(\|\cdot\|_2\) denotes the Euclidean norm. The notation \(\softO\) suppresses polylogarithmic factors in the size parameters. We write \(\mathbb T=\mathbb R/\mathbb Z\). Bold lowercase letters denote vectors; matrices are written in plain uppercase type.

\subsection{Computational Settings and Resource Conventions}
\label{sec:computational-model}

In the black-box model, group elements have unique \(O(n)\)-bit encodings, and the group operation is supplied as a reversible circuit
\[
  U_+:\ |a\rangle|b\rangle\longmapsto |a\rangle|a+b\rangle.
\]
Its controlled form, inverse, and specializations to a classically fixed operand have gate cost at most \(T_{\mathrm{op}}\) and use \(O(n)\) work qubits, returned to zero.

We count elementary one- and two-qubit gates. The gate count of a circuit is the number of elementary
gates in one execution.
We also report the sum of the circuit gate counts,
counting each circuit once.
Quantum time complexity is the total number of elementary
gates applied over all executions, including repetitions.
Quantum space is the maximum number of simultaneously
used qubits.Classical costs are measured separately in bit operations and group-operation queries.

\subsection{Finite Abelian Groups, Characters, and Lattices}
\label{sec:groups-characters-lattices}

An invariant-factor decomposition of \(G\) is an isomorphism
$
  G\cong\prod_{i=1}^r\mathbb Z_{n_i},
$
where the invariant factors satisfy $2\leq n_1\mid n_2\mid\cdots\mid n_r$. The output also includes generators realizing this isomorphism. Trivial factors are omitted, and the trivial group is represented by an empty product. If \(a_1,\ldots,a_t\) generate a subgroup \(H\leq G\), their relation lattice
\[
  \Lambda=\left\{\mathbf z\in\mathbb Z^t:
                  \sum_{i=1}^t z_i a_i=0_G\right\}
\]
gives the presentation \(H\cong\mathbb Z^t/\Lambda\).

Lattice basis matrices are written with basis vectors as columns. If \(A\) is a basis matrix for \(\Lambda\), Smith normal form computes unimodular matrices \(U,V\) such that \(UAV\) is diagonal, with positive entries in divisibility order. The nontrivial diagonal entries are the invariant factors of \(H\), and the corresponding columns of \(U^{-1}\) express cyclic generators in the original generators \(a_i\). Hermite normal form extracts a basis from an integer generating set. Both forms and the required transformations can be computed in polynomial time~\cite{KannanBachem79}.

\begin{theorem}[Finite Abelian Generation~\cite{ref3}]\label{thm:random-generation}
Let \(G\) be a finite Abelian group with minimal number of generators \(r\). For each prime \(p\) dividing \(|G|\), let \(r_p\) denote the \(p\)-rank of \(G\). If \(q=r+j\) elements, with \(j\geq1\), are chosen uniformly and independently from \(G\), the probability \(P_j\) that they generate \(G\) is
\[
  P_j=\prod_{p\mid |G|}\prod_{i=1}^{r_p}
      \left(1-p^{-(r-r_p+j+i)}\right).
\]
Moreover,
\[
  P_j>\prod_{k=j+1}^{\infty}\zeta(k)^{-1}.
\]
\end{theorem}

\begin{corollary}\label{cor:generation-probability}
Let \(G\) be a finite Abelian group with minimal number of generators \(r\). If \(q=r+j\) elements are chosen uniformly and independently from \(G\), where \(j\geq1\), then
$
  P_j>1-2^{-(j-1)}.
$
Consequently, \(1-P_j<2^{-(j-1)}=O(2^{-j})\).
\end{corollary}

\begin{proof}
Theorem~\ref{thm:random-generation} and Euler's product give
\[
  P_j>
  \prod_p\prod_{k=j+1}^{\infty}(1-p^{-k})
  \geq
  1-\sum_p\sum_{k=j+1}^{\infty}p^{-k}
  =1-\sum_p\frac{p^{-j}}{p-1}.
\]
The term for \(p=2\) is \(2^{-j}\). For the odd primes,
\[
  \sum_{p\geq3}\frac{p^{-j}}{p-1}
  =\sum_{p\geq3}p^{-(j-1)}\frac{1}{p(p-1)}
  \leq3^{-(j-1)}\sum_{p\geq3}\frac{1}{p(p-1)}
  <3^{-(j-1)}(1-\ln2)<3^{-j}.
\]
Hence \(P_j>1-2^{-(j-1)}\).
\end{proof}

\begin{definition}[Group Characters and the Dual Group]
\label{def:group-characters}
A character of a finite Abelian group \(G\) is a homomorphism \(\chi:G\to\mathbb C^\times\). The characters form an Abelian group under pointwise multiplication, called the dual group and denoted by \(\widehat G\). Its identity is the trivial character \(\chi_0\), which satisfies \(\chi_0(g)=1\) for every \(g\in G\).
\end{definition}

Every character value has modulus one. For a subgroup \(K\leq G\), its annihilator is
\[
  K^\perp=\{\chi\in\widehat G:\chi(k)=1\text{ for all }k\in K\}.
\]
The annihilator is canonically isomorphic to \(\widehat{G/K}\), and
$
K=\bigcap_{\chi\in K^\perp}\ker\chi.
$
We use character orthogonality in the form
\[
  \mathbf1_K(g)=\frac1{|K^\perp|}\sum_{\chi\in K^\perp}\chi(g).
\]
Every character of a subgroup \(H\leq G\) extends to a character of \(G\); equivalently, the restriction map \(\widehat G\to\widehat H\) is surjective~\cite{Terras99}. Under the identification \(t\mapsto e^{2\pi i t}\), characters may also be written as additive phase homomorphisms with values in \(\mathbb T\).

Characters of \(\mathbb Z^m\) are parametrized by \(\mathbf v\in\mathbb T^m\) through \(\mathbf z\mapsto e^{2\pi i\langle\mathbf v,\mathbf z\rangle}\). Hence a character of \(A\times\mathbb Z^m\), for a finite Abelian group \(A\), may be represented by a pair \((\eta,\mathbf v)\in\widehat A\times\mathbb T^m\).

A full-rank lattice \(\Lambda\subseteq\mathbb R^m\) is generated by \(m\) linearly independent vectors. Its determinant is the covolume of a fundamental parallelepiped. For \(\Lambda\subseteq\mathbb Z^m\), this is the index \([\mathbb Z^m:\Lambda]\). The dual lattice is
\[
  \Lambda^*=\{\mathbf y\in\mathbb R^m:
       \langle\mathbf y,\mathbf z\rangle\in\mathbb Z
       \text{ for all }\mathbf z\in\Lambda\}.
\]
The group \(\Lambda^*/\mathbb Z^m\) is finite of order \(\det\Lambda\) and is the character group of \(\mathbb Z^m/\Lambda\) through the pairing \(e^{2\pi i\langle\mathbf y,\mathbf z\rangle}\). Pairings of a class in \(\mathbb T^m\) with an integer vector are understood modulo one. We use the LLL algorithm~\cite{LLL82} for lattice reduction and geometry-of-numbers estimates~\cite{Cassels59} for short-vector and basis bounds.

\subsection{Relation Lattices and Dual Sampling}
\label{sec:relation-lattices}

For an evaluation sequence \(\mathcal E=(g_1,\ldots,g_m)\), define the homomorphism $\Phi:\mathbb Z^m\longrightarrow G$ by
$
  \Phi(\mathbf z)=\sum_{i=1}^m z_i g_i.
$
Write $H=\operatorname{im}\Phi$ for its image.
Its relation lattice is \(\mathcal L=\ker\Phi\), and the First Isomorphism Theorem gives \(H\cong\mathbb Z^m/\mathcal L\). In particular, \(\mathcal L\) has full rank and \(\det\mathcal L=|H|\). For \(R>0\), write
\[
  \rho_R(\mathbf z)=\exp(-\pi\|\mathbf z\|_2^2/R^2).
\]
Distances between classes in \(\mathbb T^m\) are measured by
$\operatorname{dist}_{\mathbb T^m}(\mathbf x,\mathbf y)
  =\min_{\mathbf a\in\mathbb Z^m}\|\mathbf x-\mathbf y-\mathbf a\|_2.
$
For probability distributions \(\mu,\nu\) on a finite set, their total variation distance is
\[
  \operatorname{TV}(\mu,\nu)=\frac12\sum_x|\mu(x)-\nu(x)|.
\]
Total variation distance cannot increase after applying the same deterministic map to both distributions. We also use the standard product bound

$$
  \operatorname{TV}\!\left(
     \bigotimes_{i=1}^t\mu_i,
     \bigotimes_{i=1}^t\nu_i
  \right)
  \leq
  \sum_{i=1}^t\operatorname{TV}(\mu_i,\nu_i).
$$

\begin{lemma}[Dual Sampling Properties~{\cite[Prop.~A.6, Claim~A.7]{Regev25}}]
\label{lem:dual-sampling}
Let \(\mathcal L\subseteq\mathbb Z^m\) be a full-rank lattice. Let \(R>\sqrt{2m}\), and let \(D\) be a power of two satisfying \(2\sqrt m\,R\leq D<4\sqrt m\,R\). Consider Gaussian Fourier sampling on \(\mathbb Z_D^m\) with amplitudes proportional to \(\rho_R\), and let \(\widetilde{\mathbf v}=\mathbf w/D\) be the scaled Fourier outcome. If the state-preparation and Fourier-transform errors are \(O(2^{-m})\), then there is a joint distribution of \(\widetilde{\mathbf v}\) and a class \(\mathbf v\in\mathcal L^*/\mathbb Z^m\) such that, for an absolute constant \(c_0>0\),
\[
  \operatorname{dist}_{\mathbb T^m}
  (\widetilde{\mathbf v},\mathbf v)
  \leq c_0\sqrt m/R
\]
except with probability \(O(2^{-m})\). The marginal law of \(\mathbf v\) is \(O(2^{-m})\)-close to uniform on \(\mathcal L^*/\mathbb Z^m\) in total variation distance.
\end{lemma}

The parameter choice used in the group-decomposition algorithm is given in Section~\ref{sec:parameters-and-evaluation}. The recovery of \(\mathcal L\) from independent noisy dual samples is treated in Section~\ref{sec:lattice-recovery}.

\section{Black-Box Finite Abelian Group Decomposition}\label{sec:group-decomposition}
\subsection{Overview of the algorithm}\label{sec:subgroup-extensions}

The algorithm constructs a sequence of decomposed subgroups by adjoining batches of sampled elements. Fix a constant $c\geq2$. The candidate list contains $L=n+\lceil c\log_2n\rceil$ elements, of which the first $d=\lceil\sqrt n\rceil+\lceil c\log_2n\rceil$ are used for initialization. The remaining elements are processed in batches of size $b=\lceil\sqrt n\rceil$. At each extension, the order threshold $B=2^{\sqrt n}$ separates the small- and large-order cyclic factors.

\paragraph{Initialization.}
Draw a list $\mathcal S$ of $L$ independent uniform elements of $G$. Let $g_1,\ldots,g_{d}$ be its first $d$ elements and set $H_0=\langle g_1,\ldots,g_{d}\rangle$. Define the homomorphism $\Phi_0:\mathbb Z^{d}\longrightarrow G$ by
$
\Phi_0(\mathbf z)=\sum_{i=1}^{d}z_i g_i.
$
Apply Gaussian Fourier sampling to $\Phi_0$ and collect $d+\lceil c\log_2 n\rceil$ noisy samples from $(\ker\Phi_0)^*/\mathbb Z^{d}$. Using the short-basis bound in Corollary~\ref{cor:uniform-short-basis}, the classical lattice-recovery procedure reconstructs $\ker\Phi_0$ with high probability. Smith normal form then gives an invariant-factor decomposition of $H_0\cong\mathbb Z^{d}/\ker\Phi_0$, together with cyclic generators and their orders. These data initialize the first subgroup extension.

\paragraph{Cyclic-factor partition.}
Process the remaining candidates in batches of size $b$, padding the last batch with $0_G$ if necessary. Suppose that a decomposition of $H_k$ is known and that $w_1,\ldots,w_b$ is the next batch. The goal is to compute a decomposition of $H_{k+1}=\langle H_k,w_1,\ldots,w_b\rangle$. Partition the known cyclic factors as
$
H_k=H^{-}\oplus H^{+},
$
where
$
H^{-}=\bigoplus_{i=1}^{r}\langle y_i\rangle
$
is formed by the factors with $\operatorname{ord}(y_i)=N_i\le B$. The factors with $\operatorname{ord}(y'_i)=M_i>B$ form
$
H^{+}=\bigoplus_{i=1}^{s}\langle y'_i\rangle.
$ The generators and orders in both components are known. Here $r=r(H^{-})$ and $s=r(H^{+})$. All auxiliary notation in this description refers to the fixed extension \(k\); only \(H_k\) and \(H_{k+1}\) carry the iteration index. Since
$
\prod_i M_i\leq |H_k|\leq 2^n
$
and each \(M_i>2^{\sqrt n}\), the number of large-order factors satisfies
$
s\leq\sqrt n.
$
The cyclic coordinates of \(H^{-}\) are used directly, while the following lattice construction handles \(H^{+}\) and the new targets.

\paragraph{Auxiliary elements and local relation lattice.}
Set $m=d+b$. Independently sample coefficients $\alpha_{j,i}$ and $\beta_{j,l}$ uniformly from $[0,2^{2n})\cap\mathbb Z$. For each $1\le j\le d$, first form
$
a_j=\sum_{i=1}^{s}\alpha_{j,i}y'_i,
$
and then add the target contributions to obtain
$
h_j=a_j+\sum_{l=1}^b\beta_{j,l}w_l.
$
Record these coefficients for the character-extension congruences. By Lemma~\ref{lem:random-combinations} and Corollary~\ref{cor:generation-probability}, the $a_j$ generate $H^{+}$ with high probability, since $s\le\sqrt n$ and $d=\lceil\sqrt n\rceil+\lceil c\log_2 n\rceil$. On this event, \(\langle h_1,\ldots,h_d,w_1,\ldots,w_b\rangle =\langle H^{+},w_1,\ldots,w_b\rangle.\) The relations among the auxiliary elements and the targets form the local lattice
\[
\mathcal L
=\left\{(\mathbf z,\mathbf u)\in\mathbb Z^d\times\mathbb Z^b:
       \sum_{j=1}^d z_jh_j+\sum_{l=1}^b u_lw_l=0_G\right\},
\]
of dimension $m=O(\sqrt n)$.

For each fixed current subgroup presentation and target batch, Lemma~\ref{lem:random-combinations} also shows that the joint law of the $h_j$ is exponentially close to independent uniform sampling from $\langle H^{+},w_1,\ldots,w_b\rangle$. Applying Corollary~\ref{cor:uniform-short-basis} to the uniform comparison distribution shows that, with high probability, \(\mathcal L\) has an integral basis whose vectors have Euclidean norm less than \(e^{33m}\). This supplies the short-basis bound needed for lattice recovery.

\paragraph{Joint Fourier sampling.}
To recover relations involving $H^{-}$ as well, combine its known cyclic coordinates with the integer coordinates above through the map $\Psi:H^{-}\times\mathbb Z^m\longrightarrow H_{k+1}$ defined by
$
\Psi(a,\mathbf z,\mathbf u)
=a+\sum_{j=1}^d z_jh_j+\sum_{l=1}^b u_lw_l.
$
Let $P=\ker\Psi$. When the auxiliary elements satisfy the generation condition above, $\Psi$ is surjective and
$(H^{-}\times\mathbb Z^m)/P\cong H_{k+1}$. The subgroup $P$ records all relations among these coordinates, with $\mathcal L$ corresponding to those whose $H^{-}$ coordinate is zero.

Each execution prepares a uniform superposition over the cyclic coordinates of $H^{-}$ and a discrete Gaussian over the $m$ integer coordinates, with width $R=e^{35m}$. After evaluating $\Psi$ and discarding its output, apply exact Fourier transforms to the cyclic coordinates and an approximate Fourier transform to the integer grid. Collect
$
Q=r+m+\lceil c\log_2 n\rceil
$
independent joint samples using the same evaluation elements. Lemma~\ref{lem:joint-dual-sampling}, applied with lattice dimension $m$, couples these outcomes to independent uniform characters of $P^\perp$. Outside the sampling failure event, the cyclic labels agree with those of the ideal characters and the lattice phases have error $O(\sqrt m/R)$.

\paragraph{Lattice recovery and exact phases.}
By Lemma~\ref{lem:uniform-projection}, the ideal lattice components
of the joint samples are independent uniform elements of
$\mathcal L^*/\mathbb Z^m$. Apply the classical lattice-recovery
procedure to the first
$
q=m+\lceil c\log_2 n\rceil
$
measured lattice phases. Using the short-basis bound above,
Lemmas~\ref{lem:lattice-recovery} and~\ref{lem:sublattice-extraction}
give recovery of $\mathcal L$ with high probability. The augmented
lattice used by LLL has dimension $m+q=O(\sqrt n)$.

Use the projected vectors obtained from the LLL-based extraction
procedure and record the integer transformation to an HNF basis of
\(\mathcal L\). On the recovery event, these vectors generate
\(\mathcal L\) and satisfy the norm bound needed for phase recovery.
For each of the $Q$ measured lattice-phase vectors, round its
pairings with these vectors to integers. On the sampling and recovery events, these are the exact pairings with a representative of the ideal dual class near the measured vector. The reconstruction procedure of Section~\ref{sec:exact-character-recovery} transports the pairings through the recorded integer transformations and recovers the exact lattice phases. Together with the cyclic labels, these phases give the $Q$ joint characters used in the subgroup update.

\paragraph{Character extension and subgroup update.}
Compute
$
\Delta:=\det\mathcal L
=
\bigl|\langle h_1,\ldots,h_d,w_1,\ldots,w_b\rangle\bigr|.
$
Since each \(w_l\) belongs to this subgroup, its order divides \(\Delta\). We therefore introduce the finite coordinate group
\[
F=
\left(\prod_{i=1}^{r}\mathbb Z_{N_i}\right)
\times
\left(\prod_{i=1}^{s}\mathbb Z_{M_i}\right)
\times
\mathbb Z_{\Delta}^{\,b},
\]
whose coordinates correspond to the known cyclic generators
\(y_1,\ldots,y_r\), \(y'_1,\ldots,y'_s\), and the new targets
\(w_1,\ldots,w_b\). The exact joint samples give the additive phases of the sampled
characters on \(y_i,h_j,w_l\). Using the recorded coefficients
\(\alpha_{j,i}\) and \(\beta_{j,l}\), solve the character-extension
congruences in Equation~\eqref{eq:character-pullback} to recover the
corresponding phases on \(y'_1,\ldots,y'_s\). On the event that
\(a_1,\ldots,a_d\) generate \(H^{+}\), these phases are uniquely
determined, as shown in Lemma~\ref{lem:annihilator-isomorphism}.

The recovered data give \(Q\) characters of \(F\). With high
probability, they generate the annihilator of the kernel of the
evaluation map from \(F\) to \(H_{k+1}\); hence their common kernel is
exactly this evaluation kernel. Taking its inverse image under the
coordinate-reduction map from an integer lattice to \(F\) gives an
integer presentation of \(H_{k+1}\). Smith normal form then yields the
invariant factors and corresponding cyclic generators used in the
next extension.

\subsection{Quantum Resources and Algorithm Comparison}
\label{sec:decomposition-resources}

Table~\ref{tab:decomposition-resources} gives the resource bounds for the Cheung--Mosca algorithm~\cite{Cheung01}, the extended Cheung--Mosca algorithm of Bermejo-Vega, Lin, and Van den Nest~\cite{BermejoVegaLinVanDenNest2014}, and the algorithm developed here. Circuit size refers to the gate count of one execution, whereas the total gate count includes all circuit executions. For the Cheung–Mosca algorithm, we analyze an implementation that uses Shor’s order-finding procedure and Pilatte’s factoring subroutine. The group-decomposition steps are unchanged.

\begin{table}[H]
  \centering
  \small
  \renewcommand{\arraystretch}{1.2}
  \resizebox{\textwidth}{!}{%
  \begin{tabular}{@{}lccc@{}}
    \toprule
    Resource & Cheung--Mosca & Extended CM & This work \\
    \midrule
    Number of distinct circuits
      & $O(n)$ & $O(n)$ & $O(\sqrt n)$ \\
    Maximum circuit size
      & $\softO(n^2T_{\mathrm{op}})$
      & $\softO(n^2T_{\mathrm{op}})$
      & $\softO(nT_{\mathrm{op}})$ \\
    Sum of distinct-circuit gate counts
      & $\softO(n^2T_{\mathrm{op}})$
      & $\softO(n^2T_{\mathrm{op}})$
      & $\softO(n^{3/2}T_{\mathrm{op}})$ \\
    Maximum executions of one circuit
      & $O(n)$ & $O(n)$ & $O(n)$ \\
    Total circuit executions
      & $\softO(n)$ & $\softO(n)$ & $O(n^{3/2})$ \\
    Total gate count over all executions
      & $\softO(n^3T_{\mathrm{op}})$
      & $\softO(n^3T_{\mathrm{op}})$
      & $\softO(n^{5/2}T_{\mathrm{op}})$ \\
    Qubits
      & $O(n^2)$ & $O(n^2)$ & $O(n)$ \\
    \bottomrule
  \end{tabular}%
  }
  \caption{Quantum resource upper bounds under
    $T_{\mathrm{op}}=\Omega(\sqrt n)$.} 
  \label{tab:decomposition-resources}
\end{table}

\section{Regev-Type Sampling and Lattice Recovery}
\label{sec:sampling-and-recovery}

\subsection{Quantum Sampling and Classical Recovery}
\label{sec:lattice-algorithm}

This subsection develops the Gaussian Fourier-sampling procedure used for
the integer coordinates of an evaluation homomorphism. We first fix the
sampling parameters and give a clean reversible circuit for evaluating the
homomorphism. We then analyze the quantum sampling procedure and recover
the relation lattice from noisy dual samples. Once an integral basis of the
relation lattice is known, Smith normal form yields an invariant-factor
decomposition of the image subgroup.

\subsubsection{Parameters and Reversible Evaluation}
\label{sec:parameters-and-evaluation}

We use the black-box model and the resource conventions of
Section~\ref{sec:computational-model}, with $|G|\leq 2^n$.
Set
$
d=\lceil\sqrt n\rceil+\lceil c\log_2 n\rceil,
b=\lceil\sqrt n\rceil,
$
where $c\geq2$ is the constant fixed in
Section~\ref{sec:subgroup-extensions}.
For initialization the evaluation sequence has length $m=d$, whereas a
subgroup extension uses $m=d+b$. Thus $m=\Theta(\sqrt n)$ in both cases.
For a fixed circuit, write
$
\mathcal E=(e_1,\ldots,e_m)
$
for its evaluation sequence. Define
$
\Phi:\mathbb Z^m\longrightarrow G$
by $\Phi(\mathbf z)=\sum_{i=1}^m z_i e_i .$ Its image is
$
H=\operatorname{im}\Phi=\langle e_1,\ldots,e_m\rangle,
$
and its relation lattice is
\begin{equation}\label{eq:relation-lattice}
\mathcal L=\ker\Phi
=
\left\{
\mathbf z\in\mathbb Z^m:
\sum_{i=1}^m z_i e_i=0_G
\right\}.
\end{equation}
Since $H$ is finite, $\mathcal L$ has finite index in $\mathbb Z^m$ and
hence has full rank. The First Isomorphism Theorem gives
$
\mathbb Z^m/\mathcal L\cong H$ and
$\det\mathcal L=|H|\leq2^n.$

Set the Gaussian width $
R=e^{35m}.
$
The constant in this choice leaves the margin required by the lattice-recovery
inequality under the short-basis bound established later. Define
$
\ell=\left\lceil\log_2(2\sqrt m\,R)\right\rceil
$ and $
D=2^\ell .
$
As in Regev's discretization of Gaussian Fourier sampling
~\cite[Section~3]{Regev25},
$
2\sqrt m\,R\leq D<4\sqrt m\,R.
$
Since $m=\Theta(\sqrt n)$, we have $\ell=O(\sqrt n)$.
During Gaussian preparation, elements of $\mathbb Z_D$ are interpreted
through the centered representatives
$
\{-D/2,\ldots,D/2-1\},
$
while Fourier measurement outcomes are written as
$\{0,\ldots,D-1\}$.
The lattice-recovery procedure uses
$
q=m+\lceil c\log_2 n\rceil
$
independent samples.

Let $c_0\ge1$ be the absolute constant from
Lemma~\ref{lem:dual-sampling}. Define the dyadic upper bound
$
\delta
=
2^{1+\lfloor\log_2(c_0\sqrt m/R)\rfloor}$, and let $
\sigma=\delta^{-1}.
$
Then
$
\frac{c_0\sqrt m}{R}<\delta
\leq
\frac{2c_0\sqrt m}{R}.
$
For all sufficiently large $n$, $\delta<1$, and $\sigma$ is an integral
power of two satisfying
$
\log_2\sigma=O(m)=O(\sqrt n).
$

Group registers contain the unique $O(n)$-bit encodings from
Section~\ref{sec:computational-model}. The reversible group operation is
\[
U_+:\ |a\rangle|b\rangle
\longmapsto
|a\rangle|a+b\rangle.
\]
Its controlled form, inverse, and specialization to a classically fixed
operand have gate cost at most $T_{\mathrm{op}}$ and use $O(n)$ work
qubits, which are returned to zero. Precomputation involving classically
fixed group elements is counted separately in group-operation queries.

For the reversible evaluation circuit, let
$
F_0=0, F_1=1$ and
$F_j=F_{j-1}+F_{j-2},(j\geq2),
$
and let
$
J=\max\{j:F_j\leq D\}.
$
Then $J=\Theta(\log D)=O(\sqrt n)$ and
$
F_J\leq D<F_{J+1}.
$ For a centered coefficient
$z_i\in\{-D/2,\ldots,D/2-1\}$, put
$
t_i=z_i+\frac D2.
$
Thus $0\leq t_i<D<F_{J+1}$. The greedy Fibonacci decomposition used in
\cite[Appendix~G, Lemma~G.1]{Ragavan26} therefore gives bits
$\epsilon_{i,j}\in\{0,1\}$ such that
\begin{equation}\label{eq:fibonacci-recoding}
t_i
=
z_i+\frac D2
=
\sum_{j=1}^{J}\epsilon_{i,j}F_j .
\end{equation}
Starting with $F_J$ and proceeding down to $F_2$, the circuit compares the
current value with $F_j$. If it is at least $F_j$, the circuit sets
$\epsilon_{i,j}=1$ and subtracts $F_j$; otherwise it sets
$\epsilon_{i,j}=0$. Since
$
t_i=\sum_{j=2}^{J}\epsilon_{i,j}F_j,
$
the work register contains zero after the procedure terminates. The
coefficients $\epsilon_{i,j}$ are then used in the group evaluation.
Afterwards, the reversible computation is run in reverse to clear the
registers storing these coefficients. The coefficient registers use $m\ell=O(n)$ qubits, and the registers storing
the coefficients $\epsilon_{i,j}$ use $mJ=O(n)$ qubits. We use two further
group registers $|A_1\rangle$ and $|A_2\rangle$, initialized to
$|0_G\rangle$, as accumulators. A third group register is reused for the
group element $c_j$ defined below. An additional $O(n)$-qubit bit register,
initialized to zero, stores the final encoding of $\Phi(\mathbf z)$.

\begin{lemma}
\label{lem:fibonacci-evaluation}
Under the parameters above, there is a clean reversible circuit implementing
\[
|\mathbf z\rangle|0\rangle|0\rangle
\longmapsto
|\mathbf z\rangle|0\rangle|\Phi(\mathbf z)\rangle .
\]
The circuit uses $O(n)$ reversible group operations,
$
\widetilde O\!\left(nT_{\mathrm{op}}+n^{3/2}\right)
=
\widetilde O\!\left(nT_{\mathrm{op}}\right)
$
elementary gates, and $O(n)$ qubits.
\end{lemma}

\begin{proof}
For $1\leq j\leq J$, define
$
c_j:=\sum_{i=1}^m\epsilon_{i,j}e_i .
$
The following recurrence is the additive analogue of the Fibonacci
multi-exponentiation procedure in
\cite[Section~4.2]{Ragavan26}
Initialize
$
A_1=A_2=0_G.
$
For $j=J,J-1,\ldots,1$, perform
$
A_1\leftarrow A_1+A_2,
$ and
$A_1\leftarrow A_1+c_j,
$
and then exchange the two accumulator registers.
The element $c_j$ is computed in the reusable group buffer immediately
before the second update and is uncomputed immediately afterwards.

We claim that after round $j$,
$
A_1
=
\sum_{r=j+1}^{J}F_{r-j}c_r,
A_2
=
\sum_{r=j}^{J}F_{r+1-j}c_r.
$
For $j=J$, the two accumulators are $0_G$ and $c_J$, so the formulas hold.
For the induction step, use
$F_{k+1}=F_k+F_{k-1}$ in the first accumulator before the exchange.
The stated formulas then follow for round $j$ from those for round $j+1$. After the final round,
\[
\begin{aligned}
A_2
=\sum_{j=1}^{J}F_jc_j
=\sum_{i=1}^m
  \left(\sum_{j=1}^{J}\epsilon_{i,j}F_j\right)e_i
=\sum_{i=1}^m
  \left(z_i+\frac D2\right)e_i
=\Phi(\mathbf z)+a,
\end{aligned}
\]
where
$
a:=\frac D2\sum_{i=1}^m e_i
$
depends only on the fixed evaluation sequence.
The element $a$ may therefore be precomputed classically.
Applying the fixed translation by $-a$ leaves $A_2=\Phi(\mathbf z)$.

All accumulator updates are reversible. After copying the unique encoding
of $\Phi(\mathbf z)$ from $A_2$ to the zeroed output register by bitwise
controlled-NOT gates, apply the inverse translation by $+a$, run the
Fibonacci rounds in reverse, and finally reverse the computation of the
coefficients $\epsilon_{i,j}$. The accumulator, slice, arithmetic, and work registers storing the coefficients $\epsilon_{i,j}$,
then return to their prescribed initial states.

For each of the $J$ rounds, forming and clearing $c_j$ requires $O(m)$
controlled additions by the fixed elements $e_i$. The two accumulator
updates require two further group operations. Including the reverse
computation, the total number of reversible group operations is
$
O((m+1)J)=O(n).
$
The reversible Fibonacci decomposition and its inverse use
$
O(mJ\ell)=O(n^{3/2})
$
elementary gates. Register exchanges contribute at most
$O(nJ)=O(n^{3/2})$ gates, and copying the output uses $O(n)$ gates.
Hence the total gate count is
$
\widetilde O\!\left(
nT_{\mathrm{op}}+n^{3/2}
\right)
=
\widetilde O(nT_{\mathrm{op}}),
$
where the last equality uses
$T_{\mathrm{op}}=\Omega(\sqrt n)$.

The coefficient registers and the registers storing $\epsilon_{i,j}$ occupy
$
O(m\ell+mJ)=O(n)
$
qubits. The arithmetic workspace, group accumulators, slice register,
output register, and group-operation workspace also use $O(n)$ qubits in
total. This proves the claimed space bound.
\end{proof}

\paragraph{Clean evaluation oracle.}
We denote the circuit of Lemma~\ref{lem:fibonacci-evaluation} by
\begin{equation}\label{eq:clean-evaluation}
U_\Phi:\
|\mathbf z\rangle|0\rangle|0\rangle
\longmapsto
|\mathbf z\rangle|0\rangle|\Phi(\mathbf z)\rangle .
\end{equation}
In particular, the coefficient registers are back in their original binary
representation before the oracle output is measured and before the Fourier
transform is applied.

\subsubsection{Gaussian Fourier Sampling and Quantum Resources}
\label{sec:gaussian-sampling}

One execution consists of Gaussian state preparation, evaluation of
$\Phi$, and Fourier measurement. The measured grid point is a noisy
approximation to a class in $\mathcal L^*/\mathbb Z^m$, as made precise
below.

\paragraph{Gaussian state preparation.}
We prepare an approximation to the $m$-dimensional discrete Gaussian
state $|\psi_0\rangle$ over $\mathbb Z_D^m$, using the centered
representatives fixed above:
\begin{equation}\label{eq:gaussian-state}
    |\psi_0\rangle\propto
    \sum_{\mathbf z\in\mathbb Z_D^m}
    \rho_R(\mathbf z)|\mathbf z\rangle,
\end{equation}
where
$\rho_R(\mathbf z)=\exp(-\pi\|\mathbf z\|_2^2/R^2)$.
Since
$
\rho_R(\mathbf z)=\prod_{i=1}^m \rho_R(z_i),
$
the normalized state is a tensor product of $m$ one-dimensional
discrete Gaussian states.

We use the state-preparation construction of
\cite[Section~3]{Regev25}, based on the Grover--Rudolph
procedure~\cite{GroverRudolph02}. For each coordinate, the procedure is
applied to probability weights proportional to $\rho_R(z)^2$, so that
the resulting amplitudes are proportional to $\rho_R(z)$. Choose the
preparation precision so that each normalized one-dimensional factor
has Euclidean error at most
$n^{-(c+1)}/m$. The Euclidean error of a tensor product is at most the
sum of the errors of its factors, and hence the prepared
$m$-dimensional state has error at most $n^{-(c+1)}$.

\paragraph{Evaluation and coset measurement.}
Apply the clean evaluation oracle $U_\Phi$ to the coefficient register.
Up to normalization, its action on the target state is
\begin{equation}
\sum_{\mathbf z\in\mathbb Z_D^m}
 \rho_R(\mathbf z)|\mathbf z\rangle|0\rangle
\xrightarrow{U_\Phi}
\sum_{\mathbf z\in\mathbb Z_D^m}
 \rho_R(\mathbf z)|\mathbf z\rangle|\Phi(\mathbf z)\rangle .
\end{equation}
Measuring the output register in the computational basis gives the
unique encoding of some $h\in H$. For
$\mathbf z,\mathbf z'\in\mathbb Z^m$,
$
\Phi(\mathbf z)=\Phi(\mathbf z')
\quad\Longleftrightarrow\quad
\mathbf z-\mathbf z'\in\mathcal L,
$
because $\mathcal L=\ker\Phi$. Thus, if
$\Phi(\mathbf z_0)=h$, the conditional coefficient state is
proportional to
\[
\sum_{\mathbf z\in
(\mathbf z_0+\mathcal L)\cap
\{-D/2,\ldots,D/2-1\}^m}
\rho_R(\mathbf z)|\mathbf z\rangle .
\]
Hence the measurement produces a Gaussian-weighted coset state
truncated to the centered sampling grid.

\paragraph{Fourier measurement.}
We use the approximate Fourier transform of~\cite{Coppersmith94}, obtained by deleting sufficiently small two-qubit phase operations from the exact Fourier circuit. The approximation parameter is chosen so that the resulting transform differs from the exact Fourier transform by at most \(n^{-(c+1)}/m\) in operator norm. Hence the tensor product of the \(m\) transforms has operator-norm error at most \(n^{-(c+1)}\). Measuring the coefficient registers in the
computational basis produces
$
\mathbf w\in\{0,\ldots,D-1\}^m,
$
and we set
$
\widetilde{\mathbf v}:=\frac{\mathbf w}{D}\in[0,1)^m .
$ First consider exact preparation of $|\psi_0\rangle$ and the exact
Fourier transform over $\mathbb Z_D^m$. By
Lemma~\ref{lem:dual-sampling}, the resulting outcome can be coupled to
a class
$\mathbf v\in\mathcal L^*/\mathbb Z^m$ such that
\[
\operatorname{dist}_{\mathbb T^m}
(\widetilde{\mathbf v},\mathbf v)
\leq \frac{c_0\sqrt m}{R}
\]
except with probability $O(2^{-m})$, while the marginal distribution
of $\mathbf v$ is $O(2^{-m})$-close to uniform on
$\mathcal L^*/\mathbb Z^m$.

The state-preparation and Fourier-transform errors contribute at most \(2n^{-(c+1)}\) to the total variation distance of the output distribution, by the standard relation between state distance, trace distance, and measurement statistics. Combining this with the \(O(2^{-m})\) error in Lemma~\ref{lem:dual-sampling} gives
$
O\!\left(2^{-m}+n^{-(c+1)}\right)
=O\!\left(n^{-(c+1)}\right),
$
since \(m=\Theta(\sqrt n)\). Repeating the procedure independently with the same evaluation
elements gives
$
q=m+\lceil c\log_2 n\rceil
$
samples. By the product bound for total variation distance, the joint
distribution of these samples differs from the corresponding ideal
joint distribution by
$
O\!\left(
q\bigl(2^{-m}+n^{-(c+1)}\bigr)
\right)
=
O(n^{-c}).
$

We next bound the resources of the quantum sampling subroutine in the black-box model.

\begin{proposition}\label{prop:sampling-resources}
Let $G$ be a finite Abelian black-box group with $|G|\leq2^n$, unique encodings, and reversible group-operation cost $T_{\mathrm{op}}=\Omega(\sqrt n)$. Let $\Phi:\mathbb Z^m\to G$ be the evaluation homomorphism defined above. There is a quantum circuit for Gaussian Fourier sampling from $\Phi$ whose every execution outputs a sample $\mathbf w\in\mathbb Z_D^m$ using $O(n)$ qubits, $O(n)$ reversible group operations, and
$
\softO\!\left(nT_{\mathrm{op}}+n^{3/2}\right)
=\softO\!\left(nT_{\mathrm{op}}\right)
$
elementary gates.
\end{proposition}

\begin{proof}
The coefficient registers and the registers storing the coefficients
$\epsilon_{i,j}$ use
$
O(m\ell+mJ)=O(n)
$
qubits, since $J=\Theta(\ell)=O(\sqrt n)$. The temporary arithmetic
workspace uses $O(\ell)=O(\sqrt n)$ qubits. The group accumulators,
slice buffer, oracle output register, and group-operation workspace
use $O(n)$ further qubits. Thus the total space is $O(n)$.

By Lemma~\ref{lem:fibonacci-evaluation}, evaluating
$\Phi(\mathbf z)$ and uncomputing its work registers require
$O(n)$ reversible group operations and
$
\softO\!\left(nT_{\mathrm{op}}+n^{3/2}\right)
$
elementary gates.

At inverse-polynomial precision, the discrete-Gaussian preparation
of~\cite[Section~3]{Regev25} uses
$
\widetilde O\!\left(
m\bigl(\ell+\operatorname{poly}(\log m)\bigr)
\right)
=
\widetilde O(n)
$
gates. The approximate Fourier transforms can be implemented using
the truncation of~\cite{Coppersmith94}. At the precision chosen above,
their total gate count is
\[
O\!\left(
m\ell\log\!\frac{\ell m}{n^{-(c+1)}}
\right)
=
\widetilde O(n).
\]
These costs are no larger than the
$\widetilde O(n^{3/2})$ non-group-operation term already present in
the reversible evaluation. Therefore one execution uses
$
\softO\!\left(nT_{\mathrm{op}}+n^{3/2}\right)
=
\softO(nT_{\mathrm{op}})
$
elementary gates, where the last equality uses
$T_{\mathrm{op}}=\Omega(\sqrt n)$.
\end{proof}

\subsubsection{Lattice Recovery and Subgroup Decomposition}\label{sec:lattice-recovery}
Following the quantum subroutine, the classical post-processing recovers an integral basis of $\mathcal L$ from noisy dual samples and computes an invariant-factor decomposition of $H=\operatorname{im}\Phi$ by Smith normal form. The random-generation estimate used below is Corollary~\ref{cor:generation-probability}, and the distribution of the measured phases is given by Lemma~\ref{lem:dual-sampling}. The lattice-recovery statement is the following.

\begin{lemma}[Lattice Recovery from Noisy Dual Samples]\label{lem:lattice-recovery}
Let $\mathcal L\subset\mathbb{Z}^m$ be a full-rank lattice and $\mathbf v_1,\dots,\mathbf v_q$ independent uniform elements of $\mathcal L^*/\mathbb{Z}^m$, where $q\ge m+4$. Set 
\[
    \varepsilon:=\frac{1}{3}\left(2^{q-m-2}\det\mathcal L\right)^{-1/q}.
\]
 For some $\delta>0$, let $\widetilde{\mathbf v}_1,\dots,\widetilde{\mathbf v}_q\in[0,1)^m$ satisfy $\operatorname{dist}_{\mathbb T^m}(\widetilde{\mathbf v}_i,\mathbf v_i)<\delta$ for every $i$. With $\sigma=\delta^{-1}$, let $\Gamma\subset\mathbb{R}^{m+q}$ be the lattice generated by the columns of
\begin{equation}
    A=
    \begin{pmatrix}
        I_{m\times m} & 0_{m\times q}\\
        \sigma W & \sigma  I_{q\times q}
    \end{pmatrix},
\end{equation}
where the $i$-th row of $W\in\mathbb{R}^{q\times m}$ is $\widetilde{\mathbf v}_i^T$.

For every $\mathbf u\in\mathcal L$, there exists $\mathbf u'\in\Gamma$ whose first $m$ coordinates are $\mathbf u$ and such that $\|\mathbf u'\|_2\le\sqrt{q+1}\,\|\mathbf u\|_2$. Moreover, with probability greater than $1-6\cdot2^{-(q-m)}$ over the choice of $\mathbf v_1,\dots,\mathbf v_q$, every nonzero $\mathbf u'\in\Gamma$ satisfying
\[
    \|\mathbf u'\|_2<\frac{\delta^{-1}\varepsilon}{2}
    =\frac{\delta^{-1}}{6}\left(2^{q-m-2}\det\mathcal L\right)^{-1/q}
\]
has first $m$ coordinates equal to a nonzero vector of $\mathcal L$.
\end{lemma}

\begin{proof}
We follow the proofs of~\cite[Lemmas~4.3 and~4.4]{Regev25}. The finite group $\mathcal L^*/\mathbb{Z}^m$ can be generated by at most $m$ elements. Corollary~\ref{cor:generation-probability} therefore shows that $\mathbf v_1,\dots,\mathbf v_q$ fail to generate this group with probability less than $2^{-(q-m-1)}$.

Fix a nonzero class $\mathbf u+\mathcal L\in\mathbb{Z}^m/\mathcal L$. Since $(\mathcal L^*)^*=\mathcal L$, the natural pairing between $\mathbb{Z}^m/\mathcal L$ and $\mathcal L^*/\mathbb{Z}^m$ is nondegenerate. Hence this class defines a nontrivial homomorphism \(\mathbf v+\mathbb{Z}^m\longmapsto \langle\mathbf u,\mathbf v\rangle\pmod 1.\) Its image is a cyclic subgroup of $\mathbb T$ of some order $t\ge2$, and a uniform element of $\mathcal L^*/\mathbb{Z}^m$ is mapped uniformly onto this image. If $t<1/\varepsilon$ and the sampled cosets generate $\mathcal L^*/\mathbb{Z}^m$, then the homomorphism is nonzero on at least one $\mathbf v_i$. For that index, $\operatorname{dist}(\langle\mathbf u,\mathbf v_i\rangle,\mathbb{Z})\ge1/t>\varepsilon$.

Suppose instead that $t\ge1/\varepsilon$. Since $q-m\ge4$, we have $0<\varepsilon<1/3$, and for a uniform dual coset $\mathbf v$,
\[
    \Pr\!\left[
        \operatorname{dist}(\langle\mathbf u,\mathbf v\rangle,\mathbb{Z})
        \le\varepsilon
    \right]
    =\frac{1+2\lfloor t\varepsilon\rfloor}{t}
    \le\frac1t+2\varepsilon
    \le3\varepsilon.
\]
Thus the probability that none of the $q$ samples separates this fixed class is at most $(3\varepsilon)^q$. Since $|\mathbb{Z}^m/\mathcal L|=\det\mathcal L$, a union bound over all nonzero classes gives
\[
    (\det\mathcal L-1)(3\varepsilon)^q
    <\det\mathcal L\left(2^{q-m-2}\det\mathcal L\right)^{-1}
    =2^{-(q-m-2)}.
\]
Combining this estimate with the generation failure probability shows that, except with probability less than $6\cdot2^{-(q-m)}$, every nonzero class $\mathbf u+\mathcal L$ is separated by at least one sample by more than $\varepsilon$.

Choose representatives $\mathbf v_i\in\mathcal L^*$ such that $\|\widetilde{\mathbf v}_i-\mathbf v_i\|_2<\delta$. For $\mathbf u\in\mathcal L$, set $a_i=-\langle\mathbf u,\mathbf v_i\rangle\in\mathbb{Z}$ and $\mathbf a=(a_1,\dots,a_q)^T$. Then
\[
    \mathbf u'=
    \begin{pmatrix}
        \mathbf u\\
        \sigma (W\mathbf u+\mathbf a)
    \end{pmatrix}
    \in\Gamma
\]
has first block $\mathbf u$, and
\[
    \|\mathbf u'\|_2^2
    \le \|\mathbf u\|_2^2+q\sigma ^2\delta^2\|\mathbf u\|_2^2
    =(q+1)\|\mathbf u\|_2^2.
\]
This proves the lifting assertion.

Assume now that the simultaneous separation conclusion holds, and write a nonzero vector of $\Gamma$ as $\mathbf u'=(\mathbf u,\sigma (W\mathbf u+\mathbf a))^T$. Suppose that $\|\mathbf u'\|_2<\delta^{-1}\varepsilon/2$. If $\mathbf u=0$, then $\mathbf a\ne0$ and $\|\mathbf u'\|_2\ge \sigma=\delta^{-1}$, a contradiction. Thus $\mathbf u\ne0$. If $\mathbf u\notin\mathcal L$ and $\|\mathbf u\|_2\ge\varepsilon/(2\delta)$, then $\|\mathbf u'\|_2\ge\delta^{-1}\varepsilon/2$, again a contradiction. In the remaining case, $\mathbf u\notin\mathcal L$ and $\|\mathbf u\|_2<\varepsilon/(2\delta)$. There is then an index $i$ for which $\operatorname{dist}(\langle\mathbf u,\mathbf v_i\rangle,\mathbb{Z})>\varepsilon$, while $|\langle\mathbf u,\widetilde{\mathbf v}_i-\mathbf v_i\rangle|<\varepsilon/2$. Hence $\operatorname{dist}(\langle\mathbf u,\widetilde{\mathbf v}_i\rangle,\mathbb{Z})>\varepsilon/2$, so the corresponding coordinate of $\sigma (W\mathbf u+\mathbf a)$ has absolute value greater than $\sigma \varepsilon/2=\delta^{-1}\varepsilon/2$. This is impossible. Therefore $\mathbf u\in\mathcal L\setminus\{0\}$.
\end{proof}

\begin{lemma}[Claim 5.1 in~\cite{Regev25}]\label{lem:sublattice-extraction}
Let $\Lambda \subset \mathbb{R}^a$ be a full-rank lattice, and let $T > 0$ be a real-valued norm bound. There exists a classical polynomial-time algorithm that, given an arbitrary basis of $\Lambda$, outputs a sequence of $l \le a$ linearly independent vectors $\mathbf u_1, \dots, \mathbf u_l \in \Lambda$ satisfying the Euclidean norm bound:
\[
    \|\mathbf u_i\| \le \sqrt{a}2^{a/2}T, \quad \text{for all } 1 \le i \le l.
\]
Furthermore, any vector in $\Lambda$ with Euclidean norm at most $T$ is guaranteed to be an integer linear combination of $\{\mathbf u_1, \dots, \mathbf u_l\}$. Consequently, the sublattice generated by these vectors contains all vectors in $\Lambda$ of norm at most $T$.
\end{lemma}

To recover $\mathcal L$, we apply
Lemma~\ref{lem:sublattice-extraction} to the augmented lattice in
Lemma~\ref{lem:lattice-recovery}. Under the short-basis assumption,
if the extracted vectors satisfy the norm bound in
Lemma~\ref{lem:lattice-recovery}, their projections belong to
$\mathcal L$ and generate $\mathcal L$. The following theorem uses the short-basis bound $T=e^{33m}$ from Corollary~\ref{cor:uniform-short-basis} and includes the Smith normal form step for the image subgroup.

\begin{theorem}\label{thm:subgroup-decomposition}
Let $G$ be a finite Abelian black-box group with $|G|\leq2^n$, unique encodings, and reversible group-operation cost $T_{\mathrm{op}}=\Omega(\sqrt n)$. Let $\Phi:\mathbb Z^m\to G$ be the evaluation homomorphism defined above, with $H=\operatorname{im}\Phi$ and $\mathcal L=\ker\Phi$. Suppose that $\mathcal L$ has an integral basis $\mathbf b_1,\ldots,\mathbf b_m$ satisfying
\[
\max_{1\leq j\leq m}\|\mathbf b_j\|_2\leq e^{33m}.
\]
For all sufficiently large $n$, there is a quantum algorithm with classical post-processing that recovers an integral basis of $\mathcal L$ and computes an invariant-factor decomposition of $H$ together with corresponding cyclic generators with probability at least $1-O(n^{-c})$. The algorithm uses $O(n)$ qubits and $\softO(n^{3/2}T_{\mathrm{op}})$ quantum time complexity in total, while the classical computation uses polynomially many bit operations and group-operation queries.
\end{theorem}

\begin{proof}
Put $a=m+q=2m+\lceil c\log_2n\rceil=O(\sqrt n)$. By Lemma~\ref{lem:dual-sampling} and its total-variation estimate, the measured vectors $\widetilde{\mathbf v}_1,\ldots,\widetilde{\mathbf v}_q$ can be coupled to independent uniform classes $\mathbf v_1,\ldots,\mathbf v_q\in\mathcal L^*/\mathbb Z^m$ so that, except with probability $
O\!\left(q\bigl(2^{-m}+n^{-(c+1)}\bigr)\right)=O(n^{-c}),
$
\[
\operatorname{dist}_{\mathbb T^m}
(\widetilde{\mathbf v}_i,\mathbf v_i)<\delta,
\]
for every $1\leq i\leq q$, where $\delta$ is the dyadic bound in Section~\ref{sec:parameters-and-evaluation}. In particular, $\delta\leq2c_0\sqrt m/R$ and $\sigma=\delta^{-1}$ is rational. Since $q-m=\lceil c\log_2n\rceil\geq4$ for sufficiently large $n$, Lemma~\ref{lem:lattice-recovery} gives its separation conclusion except with probability less than $6\cdot2^{-(q-m)}$. We work on the event that both conclusions hold.

Construct the augmented lattice $\Gamma\subseteq\mathbb R^{a}$ of Lemma~\ref{lem:lattice-recovery}, with column basis
\[
\begin{pmatrix}
I_{m\times m}&0_{m\times q}\\
\sigma W&\sigma  I_{q\times q}
\end{pmatrix},
\]
where the $i$-th row of $W$ is $\widetilde{\mathbf v}_i^T$. Each basis vector $\mathbf b_j$ has a lift $\mathbf b_j'\in\Gamma$ of norm at most $\sqrt{q+1}\,e^{33m}$. Lemma~\ref{lem:sublattice-extraction}, applied with this norm bound, returns vectors $\mathbf z_1',\ldots,\mathbf z_l'$ such that
\[
\|\mathbf z_i'\|_2
\leq\sqrt{a}\,2^{a/2}\sqrt{q+1}\,e^{33m},
\]
and every $\mathbf b_j'$ is an integer linear combination of these vectors. If
\begin{equation}\label{eq:recovery-condition}
\sqrt{a}\,2^{a/2}\sqrt{q+1}\,e^{33m}
<\frac{\delta^{-1}}{6}
  \left(2^{q-m-2}\det\mathcal L\right)^{-1/q},
\end{equation}
then Lemma~\ref{lem:lattice-recovery} implies that the projections $\mathbf z_i$ onto the first $m$ coordinates belong to $\mathcal L$. Projecting the integer combinations expressing the lifted basis gives
\[
\mathcal L
=\langle\mathbf b_1,\ldots,\mathbf b_m\rangle_{\mathbb Z}
\subseteq\langle\mathbf z_1,\ldots,\mathbf z_l\rangle_{\mathbb Z}
\subseteq\mathcal L.
\]
Thus the projected vectors generate $\mathcal L$, and Hermite normal form yields an integral basis.

To verify Equation~\eqref{eq:recovery-condition}, use $\det\mathcal L\leq2^n$, $n\leq m^2$, and $q\geq m$ to obtain
\[
\left(2^{q-m-2}\det\mathcal L\right)^{-1/q}
\geq2^{-(n+q-m-2)/q}
\geq2^{-(m+1)}.
\]
Together with $\delta^{-1}\geq R/(2c_0\sqrt m)$, this shows that it suffices to have
\begin{equation}\label{eq:gaussian-width-condition}
e^{35m}>24c_0\sqrt{ma(q+1)}\,2^{a/2+m}e^{33m}.
\end{equation}
Taking logarithms and dividing by $m$, Equation~\eqref{eq:gaussian-width-condition} is equivalent to
\[
35>
33+\left(\frac{a}{2m}+1\right)\ln2
+\frac{\ln\!\left(24c_0\sqrt{ma(q+1)}\right)}{m}.
\]
Since $q=m+\lceil c\log_2n\rceil$ and $m\geq\lceil\sqrt n\rceil$, the right-hand side is \(33+2\ln2 +O\!\left(\frac{\log n}{\sqrt n}\right).\) The inequality therefore holds for sufficiently large $n$, since $35>33+2\ln2$.

The total failure probability is $O(n^{-c})$. Indeed, $2^{-(q-m)}\leq n^{-c}$, and $q2^{-m}=o(n^{-c})$ for fixed $c$, since $q=m+\lceil c\log_2n\rceil$ and $m\geq\lceil\sqrt n\rceil$.

Let $A\in\mathbb Z^{m\times m}$ be the recovered column basis matrix of $\mathcal L$. Compute unimodular matrices $U,V$ such that
\[
UAV=\operatorname{diag}(s_1,\ldots,s_m),
\]
where $1\leq s_1\mid s_2\mid\cdots\mid s_m$.
The map $\Phi$ induces an isomorphism $\mathbb Z^m/\mathcal L\cong H$. Since $U$ maps $\mathcal L$ to $\operatorname{diag}(s_1,\ldots,s_m)\mathbb Z^m$, Smith normal form gives \(H\cong\prod_{j:s_j>1}\mathbb Z_{s_j}.\) For each $j$ with $s_j>1$, let $\mathbf e_j$ be the $j$-th standard basis vector of $\mathbb Z^m$. The corresponding cyclic generator is
\[
\Phi(U^{-1}\mathbf e_j)
=\sum_{i=1}^m(U^{-1})_{i,j}e_i.
\]
It has order $s_j$, and these generators realize the displayed decomposition. Evaluate these integer combinations using the classical group operations. Trivial factors are omitted; if every $s_j=1$, then $H=\{0_G\}$.

The augmented basis has rational entries of $O(\sqrt n)$ bits, so lattice reduction and Hermite normal form run in polynomial time. Smith normal form and its unimodular transformations also have polynomial bit complexity~\cite{KannanBachem79}, and the resulting generator combinations require polynomially many classical group-operation queries. These deterministic steps do not add a failure probability. By Proposition~\ref{prop:sampling-resources}, the sampling circuit, which has \(\softO(nT_{\mathrm{op}})\) gates, is executed \(q=O(\sqrt n)\) times. With the same \(O(n)\)-qubit workspace reused across executions, the quantum time complexity is \(\softO(n^{3/2}T_{\mathrm{op}})\).
\end{proof}

\subsection{Short Bases for Random Relation Lattices}
\label{sec:short-bases}

Let $G$ be a finite Abelian group of order $N\leq2^n$, and fix a subgroup $K\leq G$. Let $d\geq1$ and $0\leq b\leq d$ be integers, and put $m=d+b$. Fix $w_1,\ldots,w_b\in G$ before sampling $g_1,\ldots,g_d$. Define the evaluation homomorphism $\Phi:\mathbb Z^d\times\mathbb Z^b\longrightarrow G$ by
\[
  \Phi(\mathbf z,\mathbf u)
  =\sum_{i=1}^d z_i g_i+\sum_{j=1}^b u_jw_j.
\]
The lattice considered below is
\begin{equation}\label{eq:preimage-lattice}
  \mathcal L=\Phi^{-1}(K)
  =\left\{(\mathbf z,\mathbf u)\in\mathbb Z^d\times\mathbb Z^b:
       \sum_{i=1}^d z_i g_i+\sum_{j=1}^b u_jw_j\in K\right\}.
\end{equation}
Writing $H=\operatorname{im}\Phi$, the induced map to $H/(H\cap K)$ gives
$
  \mathbb Z^m/\mathcal L\cong H/(H\cap K).
$
Consequently,
$  \det\mathcal L=|H/(H\cap K)|\leq2^n.
$
For group decomposition, $K=\{0_G\}$ and $\mathcal L=\ker\Phi$. The short-basis bounds do not require the random elements to generate $G$ or the subgroup from which they are sampled.

For the random-subset experiment, take $n=\lceil\log_2N\rceil$ and fix $c>0$. Set the sampling length to
$d=\lceil\sqrt n\rceil+\lceil c\log_2n\rceil.
$
The subset size is
$
  X=\lceil n^{500\sqrt n}\rceil.
$
Take $n$ sufficiently large that $X<N/2$. Choose $\mathcal X$ uniformly among the subsets of $G$ of size $X$, and then choose $g_1,\ldots,g_d$ independently and uniformly from $\mathcal X$. This experiment is used in the short-basis proof, which follows Pilatte's lattice-point counting argument~\cite{Pilatte26} and uses Hayes' Fourier-coefficient estimate~\cite{Hayes2001,BabaiNotes}. The next lemma compares this sampling law with independent uniform sampling from $G$.

\begin{lemma}\label{lem:sampling-comparison}
Let $G$ be a finite Abelian group of order $N$, and let $n=\lceil\log_2N\rceil$. Choose $\mathcal X$ uniformly among the subsets of $G$ of size $X=\lceil n^{500\sqrt n}\rceil$. Given $\mathcal X$, form a sequence $S_1$ by sampling $d$ elements independently and uniformly from $\mathcal X$. Let $S_2$ be a sequence of $d$ independent uniform samples from $G$, and denote the distributions of $S_1$ and $S_2$ by $\mu_1$ and $\mu_2$, respectively. Then their total variation distance satisfies $\operatorname{TV}(\mu_1,\mu_2)=O(d^2/X)$.
\end{lemma}

\begin{proof}
Let $\Omega=G^d$, let $\Omega_1\subseteq\Omega$ consist of the sequences with distinct entries, and put $\Omega_2=\Omega\setminus\Omega_1$. By definition,
\begin{equation}\label{eq:sampling-distance}
    \operatorname{TV}(\mu_1,\mu_2)
    =\frac12\sum_{A\in\Omega}\left|\mu_1(A)-\mu_2(A)\right|.
\end{equation}
If $d>X$, the claimed bound follows from $\operatorname{TV}(\mu_1,\mu_2)\le1$; hence assume $d\le X$.

For $A=(a_1,\ldots,a_d)\in\Omega_1$, we have $\mu_2(A)=N^{-d}$. The probability that $\mathcal X$ contains all entries of $A$ is $\binom{N-d}{X-d}/\binom NX$. Conditional on this event, sampling $A$ from $\mathcal X$ has probability $X^{-d}$. Therefore,
\[
    \mu_1(A)
    =\frac{\binom{N-d}{X-d}}{\binom NX}\frac1{X^d}
    =\frac1{N^d}\prod_{i=0}^{d-1}\frac{1-i/X}{1-i/N}.
\]
Write $P=\prod_{i=0}^{d-1}(1-i/X)/(1-i/N)$. Since $X\le N$, we have $P\le1$, while
\[
    P\ge\prod_{i=0}^{d-1}\left(1-\frac iX\right)
    \ge1-\sum_{i=0}^{d-1}\frac iX
    =1-\frac{d(d-1)}{2X}.
\]
Thus, for every $A\in\Omega_1$,
\begin{equation}\label{eq:distinct-sample-bound}
    \left|\mu_1(A)-\mu_2(A)\right|
    =\frac{1-P}{N^d}
    \le\frac{d(d-1)}{2XN^d}
    \le\frac{d^2}{2XN^d}.
\end{equation}
For independent uniform samples from $G$, a union bound over pairs of entries gives
\begin{equation}\label{eq:uniform-collision-bound}
    \mu_2(\Omega_2)\le\binom d2\frac1N\le\frac{d^2}{2N}.
\end{equation}
For samples from the random subset, the same argument gives
\begin{equation}\label{eq:subset-collision-bound}
    \mu_1(\Omega_2)\le\binom d2\frac1X\le\frac{d^2}{2X}.
\end{equation}
Splitting the sum in Equation~\eqref{eq:sampling-distance} over $\Omega_1$ and $\Omega_2$ and using $|\Omega_1|\le N^d$ yields
\[
    \operatorname{TV}(\mu_1,\mu_2)
    \le\frac{3d^2}{4X}
    =O\left(\frac{d^2}{X}\right).
\]
For $d=O(\sqrt n)$, this bound is at most $2^{-\Omega(\sqrt n\log n)}$.
\end{proof}

When a short-basis failure bound is transferred between these sampling laws, the total variation term must be added. For the random-subset experiment, the bound is as follows.

\begin{theorem}[Existence of a Short Lattice Basis]\label{thm:short-basis}
Let $G$ be a finite Abelian group of order $N$, and let $K\leq G$. Put $n=\lceil\log_2N\rceil$. Fix $c>0$ and set \(d=\lceil\sqrt n\rceil+\lceil c\log_2n\rceil.\) Fix an integer $0\leq b\leq d$ and elements $w_1,\ldots,w_b\in G$ before sampling. Take $n$ sufficiently large depending only on $c$. Choose $\mathcal X\subseteq G$ uniformly among subsets of size $X=\lceil n^{500\sqrt n}\rceil$. Choose $g_1,\ldots,g_d$ independently and uniformly from $\mathcal X$. Then the lattice $\mathcal L=\Phi^{-1}(K)$ in Equation~\eqref{eq:preimage-lattice} has an integral basis $\mathbf b_1,\ldots,\mathbf b_{d+b}$ satisfying
$
\max_{1\leq j\leq d+b}\|\mathbf b_j\|_2<e^{33(d+b)}
$
with probability at least $1-O(e^{-4d^2})$. The probability is over $\mathcal X$ and the $g_i$. The bound is uniform for $0\leq b\leq d$ and all choices of the fixed elements.
\end{theorem}

For direct uniform sampling, every nontrivial character of the sampling subgroup has mean zero. Using this identity in the lattice-point argument gives the following bound.

\begin{corollary}
\label{cor:uniform-short-basis}
Let $G$ be a finite Abelian group with $|G|\leq2^n$,
where $n\geq0$ is an integer.
Let $K\leq G$ be fixed.
Fix a subgroup $S\leq G$ from which the random elements
are sampled.
Let $d\geq\lceil\sqrt n\rceil$ and $0\leq b\leq d$,
and put $m=d+b\geq10$.
Fix $w_1,\ldots,w_b\in S$ before sampling.
Choose $g_1,\ldots,g_d$ independently and uniformly from $S$. Then
the lattice $\mathcal L=\Phi^{-1}(K)$ has an integral basis
$\mathbf b_1,\ldots,\mathbf b_m$ satisfying
$\max_{1\leq j\leq m}\|\mathbf b_j\|_2<e^{33m}$
with probability at least $1-2e^{-2m^2}$.
The probability is over the choice of $g_1,\ldots,g_d$.
\end{corollary}

The group-decomposition algorithm uses the bound for
the relation lattice $\ker\Phi$.
Initialization uses uniform sampling from $G$ with $b=0$.
In a subgroup extension, the subgroup used for the
uniform-sampling comparison is
$S=\langle H^{+},w_1,\ldots,w_b\rangle$.
The elements $w_1,\ldots,w_b$ are fixed during this sampling.

We now prove Theorem~\ref{thm:short-basis} for $d$ random evaluation elements and $b$ fixed elements. The result is restated in Theorem~\ref{thm:short-basis} below. We first identify the quotient whose order is the lattice determinant. Character orthogonality then expresses the number of lattice points in a box, and separate estimates control the low- and high-order character contributions. Counts at two scales give linearly independent short vectors, from which a basis is obtained. The character identities are stated for the general family $\mathcal L_M$; the final short-basis results use the fixed choice $M=1$.

Pilatte's analysis of Regev-type period lattices proceeds by expressing bounded lattice-point counts through character sums, separating the low- and high-order character contributions, and using point counts at two scales to obtain a full set of linearly independent short vectors~\cite[Section~3]{Pilatte26}. In the arithmetic setting of~\cite{Pilatte26}, the required cancellation is obtained from number-theoretic estimates for the small-prime sampling set. For a general finite Abelian group, this arithmetic input is not available. We instead use harmonic analysis on finite Abelian groups~\cite{Terras99} together with Hayes' Fourier-coefficient estimate for random subsets~\cite{BabaiNotes,Hayes2001} to control the high-order character contribution. The final conversion from lattice-point estimates to a short basis uses the corresponding geometry-of-numbers argument~\cite{Pilatte26,Cassels59}. The fixed elements $w_1,\ldots,w_b$ are included in the relation lattice from the outset. Thus its total dimension is $m=d+b$, whereas probabilistic cancellation is supplied only by the $d$ random coordinates.

\subsubsection{Characterization of the Lattice Determinant}
\label{sec:lattice-determinant}

We first identify the relation lattice associated with the hidden-subgroup oracle and compute its determinant. The determinant gives the normalization factor for the lattice-point estimates used later.

Let $G$ be a finite Abelian group and let $K\leq G$ be the hidden subgroup of an HSP instance $f:G\to\mathcal Y$. Let $d\geq1$ and $b\geq0$, let $g_1,\ldots,g_d,w_1,\ldots,w_b\in G$, and put $m=d+b$. Define the evaluation homomorphism $\Phi:\mathbb Z^d\times\mathbb Z^b\longrightarrow G$ by
\[
\Phi(\mathbf z,\mathbf u)
=\sum_{i=1}^{d}z_i g_i+\sum_{j=1}^{b}u_j w_j.
\]
By the HSP promise, \(f(\Phi(\mathbf z,\mathbf u))=f(0) \quad\Longleftrightarrow\quad \Phi(\mathbf z,\mathbf u)\in K.\) Hence the period lattice of the composed map $f\circ\Phi$ is
\[
\mathcal L
=\Phi^{-1}(K)
=\left\{(\mathbf z,\mathbf u)\in\mathbb Z^d\times\mathbb Z^b:
\sum_{i=1}^{d}z_i g_i+\sum_{j=1}^{b}u_j w_j\in K\right\}.
\]
Let \(H:=\langle g_1,\ldots,g_d,w_1,\ldots,w_b\rangle.\) The following lemma gives the determinant of this lattice.

\begin{lemma}\label{lem:lattice-determinant}
Let $G$ be a finite Abelian group, let $K\leq G$, let $d\geq1$ and $b\geq0$, and let $g_1,\ldots,g_d,w_1,\ldots,w_b\in G$. With $m$, $\Phi$, $\mathcal L$, and $H$ defined above, $\mathcal L$ is a full-rank sublattice of $\mathbb Z^m$ and
$
\mathbb Z^m/\mathcal L\cong H/(H\cap K).
$
Consequently,
$
\det\mathcal L=|H/(H\cap K)|\leq |G|.
$
\end{lemma}

\begin{proof}
Consider the homomorphism $\overline\Phi:\mathbb Z^d\times\mathbb Z^b\longrightarrow H/(H\cap K)$ defined by
\[
(\mathbf z,\mathbf u)\longmapsto
\Phi(\mathbf z,\mathbf u)+(H\cap K).
\]
Since $\Phi(\mathbb Z^d\times\mathbb Z^b)=H$, the map $\overline\Phi$ is surjective. Moreover, $\Phi(\mathbf z,\mathbf u)$ always lies in $H$, and therefore
\[
\begin{aligned}
\ker\overline\Phi
&=\{(\mathbf z,\mathbf u):\Phi(\mathbf z,\mathbf u)\in H\cap K\}\\
&=\{(\mathbf z,\mathbf u):\Phi(\mathbf z,\mathbf u)\in K\}
=\mathcal L.
\end{aligned}
\]
The First Isomorphism Theorem gives \((\mathbb Z^d\times\mathbb Z^b)/\mathcal L \cong H/(H\cap K).\) Identifying $\mathbb Z^d\times\mathbb Z^b$ with $\mathbb Z^m$ yields the stated quotient isomorphism. Since $H/(H\cap K)$ is finite, $\mathcal L$ has finite index in $\mathbb Z^m$ and hence rank $m$. For a full-rank sublattice of $\mathbb Z^m$, the determinant equals the index, so \(\det\mathcal L =[\mathbb Z^m:\mathcal L] =|H/(H\cap K)|.\) Finally, the canonical isomorphism \(H/(H\cap K)\cong(H+K)/K\) identifies this quotient with a subgroup of $G/K$. Hence \(\det\mathcal L =|H/(H\cap K)| \leq |G/K| \leq |G|.\) \end{proof}

\begin{remark}\label{rem:determinant-bound}
The determinant bound $\det\mathcal L\leq |G|\leq2^n$alone does not imply the \(e^{O(\sqrt n)}\) bound on the largest successive minimum required here. The lattice-point estimates at two scales established below provide the stronger information needed to obtain a full short basis.
\end{remark}

\subsubsection{Fourier Analytic Estimation of Lattice Points}
\label{sec:lattice-point-counting}

The determinant identity above gives the normalization for the lattice-point count. We now express that count by character orthogonality on a finite Abelian quotient, following the harmonic-analytic form of Pilatte's argument~\cite{Pilatte26,Terras99}. We keep an arbitrary integer multiplier $M\geq1$ in these identities, so they apply both to the target lattice and to the scaled evaluation map. The trivial-character contribution gives the density term; the low-order contribution is bounded deterministically, while the high-order contribution is estimated in Section~\ref{sec:high-order-characters} using the random coordinates.

We begin with the character-sum product that arises after the coordinate sums are separated.

\begin{definition}[Character Sum Product]\label{def:character-product}
Let $G$ be a finite Abelian group, let $d\geq1$ and $b\geq0$, and let
$g_1,\ldots,g_d,w_1,\ldots,w_b\in G$. Let $\chi\in\widehat G$ and let $h\geq1$ be an integer. Define
\[
F_\chi(h;g_1,\ldots,g_d;w_1,\ldots,w_b)
:=\prod_{i=1}^{d}\left(\sum_{v=-h}^{h}\chi(g_i)^v\right)
  \prod_{j=1}^{b}\left(\sum_{v=-h}^{h}\chi(w_j)^v\right).
\]
When $b=0$, the second product is understood to be $1$. If the evaluation elements are fixed, we write simply $F_\chi(h)$.
\end{definition}

To apply character orthogonality after multiplying the evaluation elements by an integer $M$, we identify characters of the subgroup $MG$ with the corresponding $M$-th powers of characters of $G$. We use the following result of Pilatte~\cite{Pilatte26}.

\begin{lemma}[Dual of the Subgroup of Multiples, Pilatte~\cite{Pilatte26}]\label{lem:dual-of-multiples}
Let $G$ be a finite Abelian group and let $M\geq1$ be an integer. Write $MG:=\{Mg:g\in G\}$ for the subgroup of multiples of $M$. For the dual group, set $\widehat G^{\,M}:=\{\chi^M:\chi\in\widehat G\}$. The map $\iota:\widehat{MG}\longrightarrow\widehat G$ defined by
$
\iota(\psi)(g):=\psi(Mg)
$
is a group isomorphism from $\widehat{MG}$ onto $\widehat G^{\,M}$.
\end{lemma}

We now set up the scaled relation lattice whose points will be counted. Let $K\leq G$, let $M,h\geq1$ be integers, and put $m=d+b$. For fixed evaluation elements $g_1,\ldots,g_d,w_1,\ldots,w_b\in G$, define
\[
\mathcal L_M
:=\left\{(\mathbf z,\mathbf u)\in\mathbb Z^d\times\mathbb Z^b:
\sum_{i=1}^{d}z_i(Mg_i)+\sum_{j=1}^{b}u_j(Mw_j)\in K\right\}.
\] Every sum in the definition of $\mathcal L_M$ lies in $MG$, so membership in $K$ is equivalent to membership in \(K':=MG\cap K.\) Define the annihilator
\[
(K')^\perp
:=\{\psi\in\widehat{MG}:\psi(x)=1\text{ for every }x\in K'\}.
\]
Its image under $\iota$ is denoted by
$
\mathcal A:=\iota\bigl((K')^\perp\bigr)\subseteq\widehat G^{\,M}.
$
By Lemma~\ref{lem:dual-of-multiples}, \(|\mathcal A|=|(K')^\perp|=|MG/K'|.\) For $y\in MG$, define the indicator of $K'$ by
\[
\mathbf1_{K'}(y)
:=
\begin{cases}
1,& y\in K',\\
0,& y\notin K'.
\end{cases}
\]
Character orthogonality on the finite quotient $MG/K'$ gives
\begin{equation}\label{eq:subgroup-indicator}
\mathbf1_{K'}(y)
=\frac{1}{|MG/K'|}\sum_{\psi\in(K')^\perp}\psi(y).
\end{equation}
The next lemma applies this identity to the lattice-point count.

\begin{lemma}[Lattice Point Count via Character Sums]\label{lem:lattice-point-count}
Let $G$ be a finite Abelian group, let $K\leq G$, let $d\geq1$ and $b\geq0$, and let $M,h\geq1$ be integers. Let $g_1,\ldots,g_d,w_1,\ldots,w_b\in G$, and let $m$, $\mathcal L_M$, $K'$, and $\mathcal A$ be defined as above. Then
\begin{equation}\label{eq:character-count-formula}
|\mathcal L_M\cap[-h,h]^m|
=\frac{1}{|\mathcal A|}\sum_{\chi\in\mathcal A}
F_\chi(h;g_1,\ldots,g_d;w_1,\ldots,w_b).
\end{equation}
\end{lemma}

\begin{proof}
By the definition of $\mathcal L_M$ and the observation preceding Equation~\eqref{eq:subgroup-indicator},
\[
|\mathcal L_M\cap[-h,h]^m|
=
\sum_{\substack{\mathbf z\in[-h,h]^d\cap\mathbb Z^d\\
                  \mathbf u\in[-h,h]^b\cap\mathbb Z^b}}
\mathbf1_{K'}\!\left(
\sum_{i=1}^{d}z_i(Mg_i)+\sum_{j=1}^{b}u_j(Mw_j)
\right).
\]
Substituting Equation~\eqref{eq:subgroup-indicator} gives
\begin{align*}
|\mathcal L_M\cap[-h,h]^m|
&=\frac{1}{|MG/K'|}
\sum_{\psi\in(K')^\perp}
\sum_{\mathbf z,\mathbf u}
\psi\!\left(
\sum_{i=1}^{d}z_i(Mg_i)+\sum_{j=1}^{b}u_j(Mw_j)
\right),
\end{align*}
where the inner sum is over the same integer box. Since $\psi$ is a character, the inner summand factors as \(\prod_{i=1}^{d}\psi(Mg_i)^{z_i} \prod_{j=1}^{b}\psi(Mw_j)^{u_j}.\) For $\chi=\iota(\psi)$, Lemma~\ref{lem:dual-of-multiples} gives $\psi(Mg_i)=\chi(g_i)$ for the elements $g_i$ and $\psi(Mw_j)=\chi(w_j)$ for the fixed elements $w_j$. The coordinate sums therefore separate:
\begin{align*}
&\sum_{\mathbf z,\mathbf u}
\prod_{i=1}^{d}\chi(g_i)^{z_i}
\prod_{j=1}^{b}\chi(w_j)^{u_j}\\
&\qquad=
\prod_{i=1}^{d}\left(\sum_{z_i=-h}^{h}\chi(g_i)^{z_i}\right)
\prod_{j=1}^{b}\left(\sum_{u_j=-h}^{h}\chi(w_j)^{u_j}\right)
=F_\chi(h).
\end{align*}
Finally, $\iota$ maps $(K')^\perp$ bijectively onto $\mathcal A$, and
$|MG/K'|=|\mathcal A|$. Substituting these identities yields Equation~\eqref{eq:character-count-formula}.
\end{proof}

For the remainder of Section~\ref{sec:lattice-point-counting}, set $b=0$. We first establish the deterministic main-term estimate in dimension $d$, which can later be replaced by the total dimension $m=d+b$. The high-order estimate in Section~\ref{sec:high-order-characters} uses only the $d$ random coordinates. In Theorem~\ref{thm:short-basis}, the fixed coordinates are included by bounding each additional character sum in absolute value by $2h+1$; they provide no probabilistic cancellation. This separates the deterministic dimension dependence from the number of independent samples, as in the treatment of fixed coordinates in Pilatte's argument~\cite{Pilatte26}.

Accordingly, in what follows
\[
\mathcal L_M
:=\left\{\mathbf z\in\mathbb Z^d:
\sum_{i=1}^{d}z_i(Mg_i)\in K\right\}.
\]
Set $K':=MG\cap K$ and write $\mathcal A:=\iota\bigl((K')^\perp\bigr)$ for the image of its annihilator under $\iota$. Define the subgroup generated by the scaled random evaluation elements by \(S_M:=\langle Mg_1,\ldots,Mg_d\rangle\leq MG.\) The corresponding density term is
\begin{equation}\label{eq:lattice-density}
V(h)
:=\frac{|K'|(2h+1)^d}{|S_M+K'|}
=\frac{(2h+1)^d}{\det\mathcal L_M},
\end{equation}
where the second identity follows from Lemma~\ref{lem:lattice-determinant} applied to the scaled evaluation elements.

We partition $\mathcal A$ according to the action of a character on $g_1,\ldots,g_d$ and according to its order. The characters that are trivial on every evaluation element form
\[
\mathcal A_0
:=\{\chi\in\mathcal A:\chi(g_i)=1\text{ for every }1\leq i\leq d\}.
\]
Among the remaining characters, those of order at most $e^{10d}$ form
\[
\mathcal A_-
:=\{\chi\in\mathcal A\setminus\mathcal A_0:
\operatorname{ord}(\chi)\leq e^{10d}\},
\]
while those of larger order form
\[
\mathcal A_+
:=\{\chi\in\mathcal A\setminus\mathcal A_0:
\operatorname{ord}(\chi)>e^{10d}\}.
\]
These sets form a disjoint partition of $\mathcal A$. We also define
\begin{equation}\label{eq:high-order-term}
E(h)
:=\frac{1}{|\mathcal A|}
\sum_{\chi\in\mathcal A_+}F_\chi(h;g_1,\ldots,g_d),
\end{equation}
and the relative contribution of the low-order characters by
\begin{equation}\label{eq:low-order-error}
\varepsilon(h)
:=\frac{1}{V(h)|\mathcal A|}
\sum_{\chi\in\mathcal A_-}
F_\chi(h;g_1,\ldots,g_d).
\end{equation}
The next lemma evaluates the trivial contribution and bounds the low-order contribution.

\begin{lemma}[Main Term Estimation]\label{lem:main-term}
Let $G$ be a finite Abelian group, let $K\leq G$, let $d\geq1$, and let $M,h\geq1$ be integers. Let $g_1,\ldots,g_d\in G$, and let $\mathcal L_M$, $K'$, $\mathcal A$, $S_M$, $V(h)$, and the character partition be defined as above. Then
\begin{equation}\label{eq:count-decomposition}
|\mathcal L_M\cap[-h,h]^d|
=V(h)\bigl(1+\varepsilon(h)\bigr)
+E(h),
\end{equation}
where
\begin{equation}\label{eq:low-order-bound}
|\varepsilon(h)|
\leq
\left(1+\frac{e^{30d}}{2h+1}\right)^d-1.
\end{equation}
In particular, if $d\geq2$ and $h\geq e^{31d}$, then
$
|\varepsilon(h)|
\leq0.145\frac{e^{31d}}{h}.
$
\end{lemma}

\begin{proof}
We first evaluate the contribution of $\mathcal A_0$. Under the isomorphism $\iota$, the characters in $\mathcal A_0$ correspond exactly to the characters of $MG$ that are trivial on both $K'$ and $S_M$, hence on $S_M+K'$. Therefore \(|\mathcal A_0| =\frac{|MG|}{|S_M+K'|}.\) For every $\chi\in\mathcal A_0$, each factor in Definition~\ref{def:character-product} equals $2h+1$, so \(F_\chi(h;g_1,\ldots,g_d)=(2h+1)^d.\) Since \(|\mathcal A|=|MG/K'|=\frac{|MG|}{|K'|},\) the normalized contribution of $\mathcal A_0$ is
\[
\frac{1}{|\mathcal A|}
\sum_{\chi\in\mathcal A_0}F_\chi(h;g_1,\ldots,g_d)
=\frac{|K'|(2h+1)^d}{|S_M+K'|}
=V(h).
\]

It remains to control the low-order characters. For $\chi\in\mathcal A_-$, define \(I_\chi:=\{i\in\{1,\ldots,d\}:\chi(g_i)\neq1\}.\) This set is nonempty because $\chi\notin\mathcal A_0$. If $i\in I_\chi$, then $\chi(g_i)$ is a nontrivial root of unity of order at most $e^{10d}$. Let this order be $r$. The sum \(\sum_{v=-h}^{h}\chi(g_i)^v\) can be divided into complete blocks of length $r$, each of which sums to zero, together with a remaining block of fewer than $r$ terms. Hence \(\left|\sum_{v=-h}^{h}\chi(g_i)^v\right| \leq r\leq e^{10d}.\) If $i\notin I_\chi$, then the same sum equals $2h+1$. Therefore
\begin{equation}\label{eq:low-order-product-bound}
|F_\chi(h;g_1,\ldots,g_d)|
\leq e^{10d|I_\chi|}(2h+1)^{d-|I_\chi|}.
\end{equation}

Consider the evaluation homomorphism $\rho:\mathcal A\longrightarrow(\mathbb C^\times)^d$ defined by
\[
\rho(\chi)=(\chi(g_1),\ldots,\chi(g_d)).
\]
Its kernel is $\mathcal A_0$. Consequently, every nonempty fiber of $\rho$ has cardinality $|\mathcal A_0|$. A character of order at most $e^{10d}$ can take on each $g_i$ only a root of unity of order at most $e^{10d}$, and the number of such roots is bounded by \(\sum_{r=1}^{\lfloor e^{10d}\rfloor}\varphi(r) \leq e^{20d}.\) Thus, for any fixed nonempty subset $I\subseteq\{1,\ldots,d\}$, the number of characters $\chi\in\mathcal A_-$ with $I_\chi=I$ is at most \(|\mathcal A_0|e^{20d|I|}.\) Combining this bound with Equation~\eqref{eq:low-order-product-bound} and summing over all nonempty $I\subseteq\{1,\ldots,d\}$ gives
\begin{align*}
\frac{1}{|\mathcal A|}
\sum_{\chi\in\mathcal A_-}
|F_\chi(h;g_1,\ldots,g_d)|
&\leq
\frac{|\mathcal A_0|}{|\mathcal A|}(2h+1)^d
\sum_{\emptyset\neq I\subseteq\{1,\ldots,d\}}
\left(\frac{e^{30d}}{2h+1}\right)^{|I|}\\
&=V(h)
\left[
\left(1+\frac{e^{30d}}{2h+1}\right)^d-1
\right].
\end{align*}
By the definition in Equation~\eqref{eq:low-order-error}, the preceding inequality proves Equation~\eqref{eq:low-order-bound}. Combining the contributions of $\mathcal A_0$, $\mathcal A_-$, and $\mathcal A_+$ in Equation~\eqref{eq:character-count-formula} gives Equation~\eqref{eq:count-decomposition}.

Assume now that $d\geq2$ and $h\geq e^{31d}$, and put \(y:=\frac{e^{31d}}{h}\leq1.\) Then \(\frac{e^{30d}}{2h+1} \leq\frac{e^{30d}}{2h} =\frac{y}{2e^d}.\) Hence
\[
\left(1+\frac{e^{30d}}{2h+1}\right)^d-1
\leq
\exp\!\left(\frac{dy}{2e^d}\right)-1.
\]
For $d\geq2$, one has $d/(2e^d)\leq e^{-2}$, and therefore \(\exp\!\left(\frac{dy}{2e^d}\right)-1 \leq e^{e^{-2}y}-1.\) The function $y\mapsto e^{e^{-2}y}-1$ is convex on $[0,1]$ and vanishes at $0$. Its secant line on this interval therefore gives
\[
e^{e^{-2}y}-1
\leq\bigl(e^{e^{-2}}-1\bigr)y
<0.145y.
\]
Substituting $y=e^{31d}/h$ gives \(|\varepsilon(h)| \leq0.145\frac{e^{31d}}{h},\) as required.
\end{proof}

\subsubsection{Control of High-Order Character Contributions}
\label{sec:high-order-characters}

Section~\ref{sec:lattice-point-counting} separates the lattice-point count into a main term, a controlled low-order contribution, and the high-order term $E(h)$. We now bound the last term. The argument uses the randomness of $g_1,\ldots,g_d$: Hayes' theorem controls the nontrivial Fourier coefficients of the sampling distribution, and this control is converted into a second-moment estimate for the character-sum product. As in Lemma~\ref{lem:main-term}, we continue with $b=0$; fixed evaluation elements are treated separately in the final short-basis theorem.

We first record a consequence of Hayes' Fourier-coefficient estimate~\cite{Hayes2001}, which supplies the required cancellation; see also~\cite{BabaiNotes}.

\begin{theorem}[Hayes~\cite{BabaiNotes,Hayes2001}]\label{thm:hayes}
Let $G$ be a finite Abelian group of order $N$, let $\eta>0$, and let $t$ be an integer satisfying $1\leq t\leq N/2$. Among the subsets $\mathcal X\subseteq G$ of cardinality $t$, all but an $O(N^{-\eta})$ fraction satisfy
\[
\max_{\chi\in\widehat G\setminus\{\chi_0\}}
\left|\sum_{u\in\mathcal X}\chi(u)\right|
<4\sqrt{(1+\eta)t\ln N},
\]
where $\chi_0$ denotes the trivial character of $G$.
\end{theorem}

We now apply Hayes' estimate to the random-subset experiment used in the short-basis theorem. Let $N=|G|$ and $n=\lceil\log_2N\rceil$. For a fixed constant $c>0$, set $d=\lceil\sqrt n\rceil+\lceil c\log n\rceil$ and $X=\lceil n^{500\sqrt n}\rceil$. For sufficiently large $n$, choose $\mathcal X\subseteq G$ uniformly among the subsets of cardinality $X$, and, conditional on $\mathcal X$, choose $g_1,\ldots,g_d$ independently and uniformly from $\mathcal X$.

Write $\mathbf g=(g_1,\ldots,g_d)$. For a function $Y:G^d\to\mathbb C$, define
\begin{equation}\label{eq:conditional-expectation}
\mathbb E_{\mathbf g\mid\mathcal X}[Y(g_1,\ldots,g_d)]
:=\frac{1}{X^d}
\sum_{(x_1,\ldots,x_d)\in\mathcal X^d}
Y(x_1,\ldots,x_d).
\end{equation}
For a function $Z:G\to\mathbb C$, similarly set \(\mathbb E_{g\mid\mathcal X}[Z(g)] :=\frac1X\sum_{x\in\mathcal X}Z(x).\) The probability in the following proposition concerns the choice of $\mathcal X$; the conditional expectations concern the samples drawn from that fixed subset. For \(K\leq G\) and an integer \(M\geq1\) fixed before sampling, use \(K'=MG\cap K\) and \(\mathcal A=\iota((K')^\perp)\subseteq\widehat G^{\,M}\) as defined in Section~\ref{sec:lattice-point-counting}.

\begin{proposition}[High-Order Character Bound]\label{prop:high-order-moment}
Let $G$ be a finite Abelian group of order $N$, let $K\leq G$, and let $M\geq1$ be an integer fixed before sampling. Use the random-subset experiment and conditional expectations defined above, with $n=\lceil\log_2N\rceil$ and $d=\lceil\sqrt n\rceil+\lceil c\log n\rceil$ for fixed $c>0$. Then, for all sufficiently large $n$ depending only on $c$, with probability at least $1-O(N^{-15})$ over $\mathcal X$, one has
\begin{equation}\label{eq:high-order-moment}
\sum_{\substack{\chi\in\mathcal A\\
\operatorname{ord}(\chi)\geq e^{10d}}}
\left(
\mathbb E_{\mathbf g\mid\mathcal X}
\left|F_\chi(h;g_1,\ldots,g_d)\right|^2
\right)^{1/2}
\ll h^de^{-4d^2}
\end{equation}
simultaneously for all integers $h\geq e^{31d}$, with an absolute implied constant.
\end{proposition}

\begin{proof}
We first verify the size condition required to apply Theorem~\ref{thm:hayes}. Since $X\leq2n^{500\sqrt n}$ and $N\geq2^{n-1}$, the inequality \(500\sqrt n\ln n+\ln2\leq(n-2)\ln2\) holds for all sufficiently large $n$. Hence $X\leq N/2$ throughout the asymptotic range considered here, and Theorem~\ref{thm:hayes} is applicable to subsets of cardinality $X$.

Fix an integer $h\geq e^{31d}$ and a character $\chi\in\mathcal A$ with $\operatorname{ord}(\chi)\geq e^{10d}$. For a fixed subset $\mathcal X$, expanding the squared modulus gives
\[
\mathbb E_{g\mid\mathcal X}
\left|\sum_{v=-h}^{h}\chi(g)^v\right|^2
=
\sum_{u,v=-h}^{h}
\mathbb E_{g\mid\mathcal X}[\chi^{u-v}(g)],
\]
because $\overline{\chi(g)^v}=\chi(g)^{-v}$. We estimate this sum according to whether $\chi^{u-v}$ is trivial. Since $\chi^{u-v}=\chi_0$ exactly when $\operatorname{ord}(\chi)\mid u-v$, the triangle inequality yields
\begin{align*}
\mathbb E_{g\mid\mathcal X}
\left|\sum_{v=-h}^{h}\chi(g)^v\right|^2
&\leq
\sum_{\substack{-h\leq u,v\leq h\\
\operatorname{ord}(\chi)\mid u-v}}1 +
\sum_{\substack{-h\leq u,v\leq h\\
\operatorname{ord}(\chi)\nmid u-v}}
\left|\mathbb E_{g\mid\mathcal X}[\chi^{u-v}(g)]\right|.
\end{align*}
All indices in these sums are integers.

Consider first the pairs satisfying $\operatorname{ord}(\chi)\mid u-v$. Their contribution is deterministic, since the corresponding character is $\chi_0$ and therefore has expectation $1$. For each fixed $u$, the admissible values of $v$ form a single residue class modulo $\operatorname{ord}(\chi)$; hence there are at most $1+\lfloor2h/\operatorname{ord}(\chi)\rfloor$ such values in $[-h,h]$. It follows that
\begin{align*}
\sum_{\substack{-h\leq u,v\leq h\\
\operatorname{ord}(\chi)\mid u-v}}1
&\leq(2h+1)\left(1+\frac{2h}{\operatorname{ord}(\chi)}\right)\\
&\leq3h+6h^2e^{-10d}
\leq9h^2e^{-10d}.
\end{align*}
Here $2h+1\leq3h$, while $\operatorname{ord}(\chi)\geq e^{10d}$ and $h\geq e^{31d}$ imply $h\geq e^{10d}$.

We next treat the complementary pairs, for which $\chi^{u-v}$ is nontrivial. Their conditional averages are controlled uniformly by Hayes' theorem. Applying Theorem~\ref{thm:hayes}~\cite{BabaiNotes,Hayes2001} with $t=X$ and $\eta=15$, we obtain that, outside an $O(N^{-15})$ fraction of the subsets $\mathcal X$, every nontrivial character $\psi\in\widehat G$ satisfies
\begin{equation}\label{eq:subset-character-bias}
\left|\mathbb E_{g\mid\mathcal X}[\psi(g)]\right|
=\frac1X\left|\sum_{x\in\mathcal X}\psi(x)\right|
<\frac{4\sqrt{(1+15)X\ln N}}{X}
=16\sqrt{\frac{\ln N}{X}}.
\end{equation}
We henceforth condition on a subset $\mathcal X$ for which this simultaneous estimate holds. For fixed $c$, the prescribed choices of $d$ and $X$ give the stronger bound \(16\sqrt{\frac{\ln N}{X}}\leq e^{-10d}\) throughout the same asymptotic range. Indeed, using $d\leq\sqrt n+c\log_2n+2$, $\ln N\leq n\ln2$, and $\ln X\geq500\sqrt n\ln n$, we have
\begin{align*}
\ln\left(16e^{10d}\sqrt{\frac{\ln N}{X}}\right)
&\leq\ln16+10(\sqrt n+c\log_2n+2)\\
&\quad+\frac12\ln(n\ln2)-250\sqrt n\ln n
\longrightarrow-\infty.
\end{align*}
Consequently, every nontrivial character $\chi^{u-v}$ occurring in the second sum satisfies
\begin{equation}\label{eq:exponential-bias-bound}
\left|\mathbb E_{g\mid\mathcal X}[\chi^{u-v}(g)]\right|
\leq e^{-10d}.
\end{equation}
Since there are at most $(2h+1)^2$ such pairs,
\[
\sum_{\substack{-h\leq u,v\leq h\\
\operatorname{ord}(\chi)\nmid u-v}}
\left|\mathbb E_{g\mid\mathcal X}[\chi^{u-v}(g)]\right|
\leq(2h+1)^2e^{-10d}
\leq9h^2e^{-10d}.
\]
The estimate from Theorem~\ref{thm:hayes} holds simultaneously for all nontrivial characters, so this summation introduces no additional exceptional event.

Combining the two contributions gives
\begin{equation}\label{eq:single-coordinate-moment}
\mathbb E_{g\mid\mathcal X}
\left|\sum_{v=-h}^{h}\chi(g)^v\right|^2
\leq18h^2e^{-10d}.
\end{equation}
Conditional on the chosen subset $\mathcal X$, the variables $g_1,\ldots,g_d$ remain independent and uniformly distributed on $\mathcal X$. Hence Definition~\ref{def:character-product} with $b=0$ gives
\begin{align*}
\mathbb E_{\mathbf g\mid\mathcal X}
\left|F_\chi(h;g_1,\ldots,g_d)\right|^2
&=\prod_{i=1}^{d}
\mathbb E_{g_i\mid\mathcal X}
\left|\sum_{v=-h}^{h}\chi(g_i)^v\right|^2\\
&\leq18^dh^{2d}e^{-10d^2}.
\end{align*}
Taking square roots therefore yields
\begin{equation}\label{eq:character-product-moment}
\left(
\mathbb E_{\mathbf g\mid\mathcal X}
\left|F_\chi(h;g_1,\ldots,g_d)\right|^2
\right)^{1/2}
\leq18^{d/2}h^de^{-5d^2}.
\end{equation}

It remains to sum this estimate over the relevant characters. Since $\mathcal A\subseteq\widehat G$, one has $|\mathcal A|\leq|\widehat G|=N$. Moreover, $d\geq\sqrt n$ and $n=\lceil\log_2N\rceil$ imply $N\leq2^n\leq e^{(\ln2)d^2}$. Thus Equation~\eqref{eq:character-product-moment} gives
\begin{align*}
&\sum_{\substack{\chi\in\mathcal A\\
\operatorname{ord}(\chi)\geq e^{10d}}}
\left(
\mathbb E_{\mathbf g\mid\mathcal X}
\left|F_\chi(h;g_1,\ldots,g_d)\right|^2
\right)^{1/2}\\
&\qquad\leq N18^{d/2}h^de^{-5d^2}
\leq18^{d/2}h^de^{-(5-\ln2)d^2}
\leq h^de^{-4d^2}.
\end{align*}
For $d\geq5$, the final inequality follows from $\frac12\ln18\leq(1-\ln2)d$. This proves Equation~\eqref{eq:high-order-moment}. The only exceptional subsets are those excluded by Theorem~\ref{thm:hayes}; in particular, the exceptional event does not depend on $h$ or on $\chi$. Therefore the asserted estimate holds with probability at least $1-O(N^{-15})$, simultaneously for every integer $h\geq e^{31d}$.
\end{proof}

The high-order estimate uses the multiplier only through the inclusion $\mathcal A\subseteq\widehat G$. No bound on the torsion of $G$ or on $M$ enters its proof. In particular, the estimate applies to the fixed choice $M=1$. We now combine it with the deterministic main-term estimate at two box sizes.

\subsubsection{Geometry of Numbers and Existence of Short Vectors}
\label{sec:geometry-of-numbers}

We now convert the preceding lattice-point estimates into a short basis for the relation lattice. Following Pilatte~\cite[Sections~3.4--3.5]{Pilatte26}, we compare the counts in two concentric boxes to obtain a full set of linearly independent vectors, and then pass from these vectors to a lattice basis. The three geometric lemmas below supply these implications. Their dimension parameter is denoted by $m$, so that they apply both to the $d$ random coordinates and to the full set of $d+b$ evaluation coordinates.

The first lemma bounds the number of grid cubes that can meet a hyperplane. It is used in the next lemma to rule out the possibility that all lattice points in the larger box lie in a proper linear subspace.

\begin{lemma}[{\cite[Lemma~3.15]{Pilatte26}}]\label{lem:cube-intersection}
Let $m\geq1$ and $L\geq1$ be integers. Cover $[-L,L]^m$ by the $(2L)^m$ axis-parallel unit cubes with integer vertices. If $V\subseteq\mathbb R^m$ is a hyperplane through the origin, then at most $(m+1)(2L)^{m-1}$ of these cubes intersect $V$.
\end{lemma}

The resulting criterion compares the lattice-point densities at two scales. When the ratio of the scales is sufficiently large relative to the two densities, the points in the larger box span the ambient space.

\begin{lemma}[{\cite[Lemma~3.16]{Pilatte26}}]\label{lem:independent-vectors}
Let $m\geq1$ be an integer and let $\Lambda\subseteq\mathbb R^m$ be a full-rank lattice. Let $2\leq h_0<h_1$ be real numbers with $h_1/h_0\in\mathbb Z$. Suppose that $\theta_0,\theta_1>0$ satisfy, for $i=0,1$,
\begin{equation}\label{eq:density-assumption}
|\Lambda\cap[-h_i,h_i]^m|
=\theta_i\frac{(2h_i+1)^m}{\det\Lambda}.
\end{equation}
If
$
\frac{h_1}{h_0}>\frac{\theta_0}{\theta_1}m\left(\frac52\right)^m,
$ then $\Lambda\cap[-h_1,h_1]^m$ contains $m$ linearly independent vectors.
\end{lemma}

These independent vectors need not generate the entire lattice. The next lemma uses the standard basis bound in terms of successive minima to obtain a basis of the target lattice after scaling the independent vectors~\cite{Pilatte26,Cassels59}.

\begin{lemma}[{\cite[Lemma~3.17]{Pilatte26}}]\label{lem:basis-transfer}
Let $m\geq1$ be an integer, let $\Lambda_1,\Lambda_2\subseteq\mathbb R^m$ be full-rank lattices, and let $M\geq1$ be an integer such that $M\Lambda_1\subseteq\Lambda_2$. If $\Lambda_1\cap[-h,h]^m$ contains $m$ linearly independent vectors for some real $h>0$, then $\Lambda_2$ has a basis whose vectors have Euclidean norm at most $m^{3/2}Mh$.
\end{lemma}

We return to the random-subset experiment of Proposition~\ref{prop:high-order-moment}. Let $G$ be a finite Abelian group of order $N$ and let $K\leq G$. Let $n=\lceil\log_2N\rceil$, $d=\lceil\sqrt n\rceil+\lceil c\log n\rceil$, and $X=\lceil n^{500\sqrt n}\rceil$, where $c>0$ is fixed. The subset $\mathcal X$ and the samples $g_1,\ldots,g_d$ are chosen as in that proposition. Fix an integer $b$ with $0\leq b\leq d$ and elements $w_1,\ldots,w_b\in G$ before sampling, and put $m=d+b$. Thus $d$ is the number of random elements, whereas $m$ is the dimension of the relation lattice. Define the homomorphism $\Phi:\mathbb Z^d\times\mathbb Z^b\longrightarrow G$ by
\[
\Phi(\mathbf z,\mathbf u)=\sum_{i=1}^d z_i g_i+\sum_{j=1}^b u_jw_j,
\]
and fix an arbitrary integer $M\geq1$ before sampling. The target lattice is
\[
\mathcal L
:=\left\{(\mathbf z,\mathbf u)\in\mathbb Z^d\times\mathbb Z^b:
          \sum_{i=1}^d z_i g_i+\sum_{j=1}^b u_jw_j\in K\right\}.
\]
Multiplying the evaluation elements by $M$ gives the auxiliary lattice
\[
\mathcal L_M
:=\left\{(\mathbf z,\mathbf u)\in\mathbb Z^d\times\mathbb Z^b:
          \sum_{i=1}^d z_i(Mg_i)+\sum_{j=1}^b u_j(Mw_j)\in K\right\}.
\]
In particular, $\mathcal L=\Phi^{-1}(K)$ is the period lattice of $f\circ\Phi$ when $f$ hides $K$. Lemma~\ref{lem:lattice-determinant} shows that both lattices have full rank. For the auxiliary lattice, put $K'=MG\cap K$ and denote the image of its annihilator under $\iota$ by $\mathcal A=\iota((K')^\perp)$. The scaled evaluation elements generate
\[
S_M=\langle Mg_1,\ldots,Mg_d,Mw_1,\ldots,Mw_b\rangle.
\]
For an integer $h\geq1$, the corresponding density term is
\[
V(h):=\frac{(2h+1)^m}{\det\mathcal L_M}
    =\frac{|K'|(2h+1)^m}{|S_M+K'|}.
\]
The notation $\mathcal L_M$, $S_M$, and $V(h)$ now includes the fixed coordinates; it agrees with the notation of Section~\ref{sec:lattice-point-counting} when $b=0$.

Lemma~\ref{lem:main-term} is deterministic and therefore applies to the full evaluation sequence with its dimension parameter replaced by $m$. For this application, the characters trivial on all random and fixed evaluation elements form
\[
\mathcal A_0
:=\{\chi\in\mathcal A:\chi(g_i)=1\ \text{for all }i
       \text{ and }\chi(w_j)=1\ \text{for all }j\}.
\]
Among the remaining characters, define the low-order part by
\[
\mathcal A_-
:=\{\chi\in\mathcal A\setminus\mathcal A_0:
                         \operatorname{ord}(\chi)\leq e^{10m}\}.
\]
The high-order part is
\[
\mathcal A_+
:=\{\chi\in\mathcal A\setminus\mathcal A_0:
                         \operatorname{ord}(\chi)>e^{10m}\}.
\]
Using the character product of Definition~\ref{def:character-product}, define the relative low-order contribution by
\[
\varepsilon(h)
:=\frac{1}{V(h)|\mathcal A|}
  \sum_{\chi\in\mathcal A_-}
  F_\chi(h;g_1,\ldots,g_d;w_1,\ldots,w_b).
\]
The high-order contribution is
\[
E(h)
:=\frac{1}{|\mathcal A|}
  \sum_{\chi\in\mathcal A_+}
  F_\chi(h;g_1,\ldots,g_d;w_1,\ldots,w_b).
\]
The cutoff $e^{10m}$ here belongs only to this deterministic application of Lemma~\ref{lem:main-term}. Proposition~\ref{prop:high-order-moment} will still be used with the original $d$ random elements and its cutoff $e^{10d}$. Since $m\geq d$, its sum includes every character in $\mathcal A_+$. As in Pilatte's proof~\cite[Theorem~3.18]{Pilatte26}, the additional character factors will be bounded deterministically before applying the random-coordinate estimate.

\begin{proof}[Proof of Theorem~\ref{thm:short-basis}]
Take $n$ sufficiently large for Proposition~\ref{prop:high-order-moment} to apply and for $d\geq10$. We first keep $M\geq1$ arbitrary and fixed before sampling; at the basis-length estimate we specialize to the deterministic choice $M=1$. By Proposition~\ref{prop:high-order-moment}, with probability at least $1-O(N^{-15})$ over the choice of $\mathcal X$, the conditional second-moment estimate holds simultaneously for all integers $h\geq e^{31d}$. Fix a subset for which this estimate holds. Until we average over $\mathcal X$ at the end of the proof, all expectations and probabilities are conditional on this fixed subset.

Lemma~\ref{lem:main-term} requires no distributional assumptions on the evaluation elements. Applying it in dimension $m$ to the sequence $(g_1,\ldots,g_d,w_1,\ldots,w_b)$ therefore gives, for every integer $h\geq e^{31m}$,
\[
|\mathcal L_M\cap[-h,h]^m|
=V(h)\bigl(1+\varepsilon(h)\bigr)
 +E(h),
\]
where the low-order error satisfies
\[
|\varepsilon(h)|\leq0.145\frac{e^{31m}}{h}.
\]
Since $S_M+K'\leq MG$ and $|\mathcal A|=|MG|/|K'|$, the determinant identity in Lemma~\ref{lem:lattice-determinant} also yields
\begin{equation}\label{eq:density-lower-bound}
V(h)=\frac{|K'|(2h+1)^m}{|S_M+K'|}
\geq\frac{(2h+1)^m}{|\mathcal A|}
\geq\frac{(2h+1)^m}{N}.
\end{equation}
In particular, the first lower bound is independent of the sampled elements $g_i$.

To estimate the high-order contribution, we separate the fixed coordinates in the character product. Definition~\ref{def:character-product} gives
\[
F_\chi(h;g_1,\ldots,g_d;w_1,\ldots,w_b)
=F_\chi(h;g_1,\ldots,g_d)
 \prod_{j=1}^b\left(\sum_{u=-h}^{h}\chi(w_j)^u\right).
\]
Since every character value has modulus one, each sum over a fixed coordinate has absolute value at most $2h+1$. Hence
\[
\left|F_\chi(h;g_1,\ldots,g_d;w_1,\ldots,w_b)\right|
\leq(2h+1)^b\left|F_\chi(h;g_1,\ldots,g_d)\right|.
\]
The set $\mathcal A_+$ depends on the evaluation sequence, but is contained in the deterministic set $\{\chi\in\mathcal A:\operatorname{ord}(\chi)>e^{10m}\}$. By the triangle inequality, we may enlarge the sum to this set before taking expectations. Applying Cauchy--Schwarz to each summand then yields
\begin{align}\label{eq:subset-error-expectation}
\mathbb E_{\mathbf g\mid\mathcal X}|E(h)|
&\leq\frac{(2h+1)^b}{|\mathcal A|}
 \sum_{\substack{\chi\in\mathcal A\\\operatorname{ord}(\chi)>e^{10m}}}
 \mathbb E_{\mathbf g\mid\mathcal X}|F_\chi(h;g_1,\ldots,g_d)|\notag\\
&\leq\frac{(2h+1)^b}{|\mathcal A|}
 \sum_{\substack{\chi\in\mathcal A\\\operatorname{ord}(\chi)\geq e^{10d}}}
 \left(\mathbb E_{\mathbf g\mid\mathcal X}
           |F_\chi(h;g_1,\ldots,g_d)|^2\right)^{1/2}\notag\\
&\ll\frac{(2h+1)^bh^d}{|\mathcal A|}e^{-4d^2}.
\end{align}
Here $m\geq d$ permits the larger summation range in the second line, and Proposition~\ref{prop:high-order-moment} bounds the resulting sum over the $d$ random coordinates. The fixed elements contribute only the factor $(2h+1)^b$.

We now compare the high-order error with the density term. Since $V(h)$ depends on the sampled elements, we first replace the threshold $V(h)/m$ by the deterministic lower bound supplied by Equation~\eqref{eq:density-lower-bound}. Markov's inequality and the preceding expectation estimate then give, for every fixed integer $h\geq e^{31m}$,
\begin{align}\label{eq:subset-error-probability}
\mathbb P_{\mathbf g\mid\mathcal X}
 \left(|E(h)|>\frac{V(h)}{m}\right)
&\leq\mathbb P_{\mathbf g\mid\mathcal X}
 \left(|E(h)|>
            \frac{(2h+1)^m}{m|\mathcal A|}\right)\notag\\
&\leq\frac{m|\mathcal A|}{(2h+1)^m}
       \mathbb E_{\mathbf g\mid\mathcal X}|E(h)|\notag\\
&\ll m\left(\frac{h}{2h+1}\right)^d e^{-4d^2}
\ll e^{-4d^2}.
\end{align}
The last bound follows from $m\leq2d$ and $m2^{-d}\leq2d2^{-d}\leq1$. The implied constant is absolute and does not depend on $b$ or on the fixed elements.

We apply this estimate at two integer scales. Set
$h_0=m\lceil e^{31m}\rceil$ and
$h_1=\left\lceil m^2(5/2)^m\right\rceil h_0$.
Then $2\leq h_0<h_1$ and $h_1/h_0\in\mathbb Z$. By Equation~\eqref{eq:subset-error-probability} and a union bound, with conditional probability $1-O(e^{-4d^2})$ the high-order error has absolute value at most $V(h_i)/m$ at both scales. On this event, define the density factors for $i=0,1$ by
\[
\theta_i:=\frac{|\mathcal L_M\cap[-h_i,h_i]^m|}{V(h_i)}
=1+\varepsilon(h_i)
 +\frac{E(h_i)}{V(h_i)}.
\]
The density factors $\theta_i$ are real. Combining the low- and high-order error bounds gives
\[
|\theta_i-1|
\leq0.145\frac{e^{31m}}{h_i}+\frac1m
\leq\frac{1.145}{m}.
\]
For $m\geq10$, both density factors are therefore positive, with \(\frac{\theta_0}{\theta_1} \leq\frac{m+1.145}{m-1.145}<m.\) The strict inequality follows from $m^2-2.145m-1.145>0$. Consequently,
\[
\frac{h_1}{h_0}
\geq m^2\left(\frac52\right)^m
>\frac{\theta_0}{\theta_1}m\left(\frac52\right)^m.
\]
Lemma~\ref{lem:independent-vectors} now applies to $\mathcal L_M$ in dimension $m$ and shows that $\mathcal L_M\cap[-h_1,h_1]^m$ contains $m$ linearly independent vectors.

It remains to obtain a basis of the target lattice from these independent vectors. For any $(\mathbf z,\mathbf u)\in\mathcal L_M$, we have \(\Phi(M\mathbf z,M\mathbf u) =\sum_{i=1}^d z_i(Mg_i)+\sum_{j=1}^b u_j(Mw_j)\in K,\) and hence $(M\mathbf z,M\mathbf u)\in\mathcal L$. Thus $M\mathcal L_M\subseteq\mathcal L$, and Lemma~\ref{lem:basis-transfer}, applied with $\Lambda_1=\mathcal L_M$ and $\Lambda_2=\mathcal L$, gives an integral basis $\mathbf b_1,\ldots,\mathbf b_m$ of $\mathcal L$ satisfying
$
\|\mathbf b_j\|_2\leq m^{3/2}Mh_1
$
for every $1\leq j\leq m$.
For $m\geq10$, the choice of $h_0$ gives
$
h_0<1.01m e^{31m}.
$
Substituting this bound in the estimate for the larger scale gives
\[
h_1<1.01m^2\left(\frac52\right)^m h_0
     <1.03m^3\left(\frac52\right)^m e^{31m}.
\]
All the preceding estimates hold for each integer $M\geq1$ fixed before sampling. Apply them to the fixed choice $M=1$. Then $\mathcal L_M=\mathcal L$, so the basis-transfer bound has no multiplier loss and gives
\begin{equation}\label{eq:basis-length-bound}
\|\mathbf b_j\|_2
\leq m^{3/2}h_1
<1.03m^{9/2}\left(\frac52\right)^m e^{31m}
<e^{33m}.
\end{equation}
For the last inequality, divide the logarithm of the middle expression by $m$. This gives \(31+\ln(5/2)+\frac{\ln1.03+(9/2)\ln m}{m}.\) The last fraction is decreasing for $m\geq10$ and is less than $1.05$ at $m=10$, while $\ln(5/2)<0.92$. The displayed quantity is therefore less than $32.97<33$. The basis-transfer lemma remains necessary even when $M=1$, since the independent vectors found by the density comparison need not generate the entire lattice.

Finally, average over $\mathcal X$. The conditional failure estimate is uniform over the subsets supplied by Proposition~\ref{prop:high-order-moment}. For fixed $c$, one has $d^2=n+o(n)$, while $\ln N\geq(n-1)\ln2$. Since $15\ln2>4$, it follows that $15\ln N\geq4d^2$ for all sufficiently large $n$ depending only on $c$. Thus $N^{-15}\leq e^{-4d^2}$, and the total failure probability is $O(e^{-4d^2})$, uniformly for $0\leq b\leq d$ and in the fixed elements.
\end{proof}

This proves Theorem~\ref{thm:short-basis}, including the case $b=0$. If the random-subset experiment is transferred to the implemented sampler using Lemma~\ref{lem:sampling-comparison}, its total-variation error must be added. For direct uniform sampling from a group or subgroup, the next corollary instead uses character orthogonality and does not require this comparison.

\subsubsection{Uniform Sampling from a Subgroup}
\label{sec:uniform-subgroup-sampling}

Let $G$ be a finite Abelian group with $|G|\leq2^n$,
where $n\geq0$ is an integer.
Let $K\leq G$ be the fixed subgroup in
Equation~\eqref{eq:preimage-lattice}.
Fix a subgroup $S\leq G$ from which the random elements
will be sampled, and put $J:=K\cap S$.
In the HSP setting, a function $f:G\to\mathcal Y$
hides $K$, and its restriction to $S$ hides $J$.

Let $d\geq1$ and $b\geq0$ be integers, and put $m=d+b$.
Fix $w_1,\ldots,w_b\in S$ before sampling.
Choose $g_1,\ldots,g_d$ independently and uniformly from $S$.
Define the evaluation homomorphism
$\Phi:\mathbb Z^d\times\mathbb Z^b\longrightarrow S$ by
\[
\Phi(\mathbf z,\mathbf u)
=\sum_{i=1}^{d}z_i g_i+\sum_{j=1}^{b}u_jw_j.
\]
The lattice considered here is
\[
\mathcal L:=\Phi^{-1}(J)
=\left\{(\mathbf z,\mathbf u)\in\mathbb Z^d\times\mathbb Z^b:
\sum_{i=1}^{d}z_i g_i+\sum_{j=1}^{b}u_jw_j\in J\right\}.
\]
Since $\Phi$ takes values in $S$, this lattice also equals
$\Phi^{-1}(K)$.
If $f$ hides $K$, then $\mathcal L$ is the period lattice
of $f\circ\Phi$.

Every nontrivial character of $S$ has mean zero under
uniform sampling from $S$.
We use this orthogonality relation for the second-moment
estimate below.
The comparison of lattice-point counts at two scales follows
the proof of Theorem~\ref{thm:short-basis}.

\begin{proof}[Proof of Corollary~\ref{cor:uniform-short-basis}]
Use the definitions preceding
Theorem~\ref{thm:short-basis}
for the ambient group $S$ and its subgroup $J=K\cap S$ and use the character-order cutoff $e^{10m}$. Keep an arbitrary integer $M\geq1$ fixed before sampling until the final specialization $M=1$. Denote the resulting auxiliary lattice, character group, density, and error terms by $\mathcal L_M$, $\mathcal A$, $V(h)$, $\varepsilon(h)$, and $E(h)$. In particular, $\mathcal A\subseteq\widehat{S}$ is fixed before sampling. Lemmas~\ref{lem:lattice-determinant}, \ref{lem:lattice-point-count}, and~\ref{lem:main-term} give, for every integer $h\geq e^{31m}$,
\begin{equation}\label{eq:subgroup-point-count}
|\mathcal L_M\cap[-h,h]^m|
=V(h)\bigl(1+\varepsilon(h)\bigr)
 +E(h),
\end{equation}
where the density term satisfies
\[
V(h)=\frac{(2h+1)^m}{\det\mathcal L_M}
\geq\frac{(2h+1)^m}{|S|}.
\]
The low-order error is bounded by
\[
|\varepsilon(h)|\leq0.145\frac{e^{31m}}h.
\]

It remains to estimate the high-order contribution. Fix $\chi\in\mathcal A$ with $\operatorname{ord}(\chi)>e^{10m}$, and let $g$ be uniform on $S$. Character orthogonality~\cite{Terras99} gives $\mathbb E_g[\chi^{u-v}(g)]=1$ if $\operatorname{ord}(\chi)\mid u-v$ and $0$ otherwise. Consequently,
\begin{align*}
\mathbb E_g\left|\sum_{v=-h}^h\chi(g)^v\right|^2
&=\left|\{(u,v)\in\mathbb Z^2:-h\leq u,v\leq h,
                       \ \operatorname{ord}(\chi)\mid u-v\}\right|\\
&\leq(2h+1)\left(1+\frac{2h}{\operatorname{ord}(\chi)}\right)\\
&\leq2(2h+1)^2e^{-10m}.
\end{align*}
For the first inequality, each fixed $u$ permits at most $1+\lfloor2h/\operatorname{ord}(\chi)\rfloor$ values of $v$; the last inequality uses $h\geq e^{31m}$. The fixed-coordinate sums have absolute value at most $2h+1$. Independence of the $d$ random elements therefore gives, with $F_\chi(h)=F_\chi(h;g_1,\ldots,g_d;w_1,\ldots,w_b)$,
\[
\bigl(\mathbb E_{\mathbf g}|F_\chi(h)|^2\bigr)^{1/2}
\leq\bigl(2(2h+1)^2e^{-10m}\bigr)^{d/2}(2h+1)^b
=(2h+1)^m2^{d/2}e^{-5dm}.
\]
Although the high-order part of the character partition depends on the sampled evaluations, it is contained in the fixed set $\{\chi\in\mathcal A:\operatorname{ord}(\chi)>e^{10m}\}$. Enlarging the sum to this set and applying Cauchy--Schwarz yields
\begin{align}\label{eq:subgroup-error-expectation}
\mathbb E_{\mathbf g}|E(h)|
&\leq\frac1{|\mathcal A|}
 \sum_{\substack{\chi\in\mathcal A\\
                 \operatorname{ord}(\chi)>e^{10m}}}
       \bigl(\mathbb E_{\mathbf g}|F_\chi(h)|^2\bigr)^{1/2}\notag\\
&\leq(2h+1)^m2^{d/2}e^{-5dm}.
\end{align}
Using the deterministic lower bound for $V(h)$ before applying Markov's inequality, we obtain
\begin{equation}\label{eq:subgroup-error-probability}
\mathbb P_{\mathbf g}\left(
 |E(h)|>\frac{V(h)}m\right)
\leq m|S|2^{d/2}e^{-5dm}
\leq e^{-2m^2}.
\end{equation}
Indeed, $|S|\leq2^n$, $n\leq d^2$, and $b\leq d$ imply $m/2\leq d\leq m$, so
\begin{align*}
\ln\bigl(m|S|2^{d/2}e^{-5dm}\bigr)
&\leq\ln m+d^2\ln2+\frac d2\ln2-5dm\\
&\leq\ln m+\frac m2\ln2-(5-\ln2)dm\\
&\leq m-\frac{5-\ln2}{2}m^2\leq-2m^2.
\end{align*}
The last line follows from $m\geq10$, which gives $\ln m\leq m/2$ and $m\leq(1-\ln2)m^2/2$.

Take $h_0$ and $h_1$ as in the proof of
Theorem~\ref{thm:short-basis}. By Equation~\eqref{eq:subgroup-error-probability} and a union bound, with probability at least $1-2e^{-2m^2}$ the high-order errors at both scales have absolute value at most $V(h_i)/m$. Equation~\eqref{eq:subgroup-point-count} then gives density factors $\theta_i=|\mathcal L_M\cap[-h_i,h_i]^m|/V(h_i)$ satisfying $|\theta_i-1|\leq1.145/m$. The deterministic comparison in the proof of Theorem~\ref{thm:short-basis} applies with the same dimension $m$ and scales $h_0,h_1$: the bound $|\theta_i-1|\leq1.145/m$ gives the required density ratio, so Lemma~\ref{lem:independent-vectors} supplies $m$ independent vectors in $\mathcal L_M\cap[-h_1,h_1]^m$. For each fixed $M$, the inclusion $M\mathcal L_M\subseteq\mathcal L$ then gives a basis of $\mathcal L$ of norm at most $m^{3/2}Mh_1$. Specialize these estimates to the fixed choice $M=1$. Since $\mathcal L_1=\mathcal L$, Equation~\eqref{eq:basis-length-bound} yields
\[
\max_j\|\mathbf b_j\|_2\leq m^{3/2}h_1<e^{33m}.
\]
The two-scale failure probability remains $2e^{-2m^2}$, proving the assertion.
\end{proof}

\begin{remark}\label{rem:separate-generation-event}
Neither Theorem~\ref{thm:short-basis} nor Corollary~\ref{cor:uniform-short-basis} assumes that the random evaluation elements generate the sampling group. If an application also requires $\langle g_1,\ldots,g_d\rangle=S$, this is a separate event. Corollary~\ref{cor:generation-probability} gives the required generation estimate when $d$ exceeds the rank of $S$ by the prescribed margin. The generation and short-basis failure probabilities may be added by a union bound; independence of the events is not required.
\end{remark}

\section{Proof of the Group-Decomposition Theorem}
\label{sec:main-theorem-proof}

We now prove Theorem~\ref{thm:group-decomposition}, using the lattice-recovery and short-basis results of Sections~\ref{sec:lattice-recovery} and~\ref{sec:short-bases}. We first construct the initial subgroup and the auxiliary elements for each extension. We then analyse the joint sampling circuit, recover exact character values from its outcomes, and compute the enlarged subgroup by Smith normal form. The final proof combines the success probabilities and resource bounds over all extensions. Group evaluations use the black-box operations of Section~\ref{sec:computational-model}; only the current subgroup $H_k$ has a known decomposition.

\subsection{Initialization and Subgroup Extension Setup}
\label{sec:extension-setup}

Let $G$ satisfy the assumptions of
Theorem~\ref{thm:group-decomposition}.
We use the following estimate for random integer combinations.

\begin{lemma}\label{lem:random-combinations}
Let $G$ be a finite Abelian group of order $N$, let $t\geq0$ be an integer, and let $\phi:\mathbb Z^t\to G$ be a surjective homomorphism. Let $M\geq1$ be an integer. If $\mathbf x$ is uniform on $\{0,\ldots,M-1\}^t$ and $\mu$ denotes the distribution of $\phi(\mathbf x)$, then
$
\operatorname{TV}(\mu,U(G))\leq\frac{tN}{4M},
$
where $U(G)$ is the uniform distribution on $G$.
\end{lemma}

\begin{proof}
If $t=0$, surjectivity implies $G=\{0_G\}$, and the bound holds with equality. Suppose $t\geq1$. Since $Ng=0_G$ for every $g\in G$, we have $N\mathbb Z^t\subseteq\ker\phi$. Thus $\phi$ induces a surjective homomorphism $\bar\phi:A\to G$, where $A=(\mathbb Z/N\mathbb Z)^t$. Its fibers have equal cardinality, so $\bar\phi$ sends the uniform distribution $U(A)$ to $U(G)$.

Write $M=qN+s$, where $q\geq0$ and $0\leq s<N$ are integers. Each coordinate of $\mathbf x\bmod N$ has distribution $\nu$ on $\mathbb Z/N\mathbb Z$, with probability mass function
\[
\nu(a)=
\begin{cases}
(q+1)/M,&0\leq a<s,\\
q/M,&s\leq a<N.
\end{cases}
\]
Consequently,
\[
\operatorname{TV}(\nu,U(\mathbb Z/N\mathbb Z))
=\frac12\sum_{a=0}^{N-1}\left|\nu(a)-\frac1N\right|
=\frac{s(N-s)}{NM}
\leq\frac{N}{4M}.
\]

The coordinates of $\mathbf x$ are independent, so $\mathbf x\bmod N$ has distribution $\nu^{\otimes t}$. Contraction under $\bar\phi$ and the product bound for total variation distance give
\[
\operatorname{TV}(\mu,U(G))
\leq\operatorname{TV}(\nu^{\otimes t},U(A))
\leq t\operatorname{TV}(\nu,U(\mathbb Z/N\mathbb Z))
\leq\frac{tN}{4M}.
\]
\end{proof}

\paragraph{Initialization.}
We first compute an invariant-factor decomposition of
$H_0=\langle g_1,\ldots,g_d\rangle$, as defined in
Section~\ref{sec:subgroup-extensions}.
The list $\mathcal S=(g_1,\ldots,g_L)$ consists of independent
uniform samples from $G$.
Since $r(G)\leq n$, Corollary~\ref{cor:generation-probability}
implies
$\Pr[\langle\mathcal S\rangle\neq G]=O(n^{-c})$. The homomorphism $\Phi_0$ defined in
Section~\ref{sec:subgroup-extensions} induces an isomorphism
$\mathbb Z^d/\ker\Phi_0\cong H_0$.
Hence $\det(\ker\Phi_0)=|H_0|\leq2^n$.
Apply Corollary~\ref{cor:uniform-short-basis}
to the relation lattice $\ker\Phi_0$,
with independent uniform sampling from $G$ and $b=0$. Then
the lattice, with probability at least $1-2e^{-2d^2}$,
$\ker\Phi_0$ has an integral basis whose vectors have
Euclidean norm less than $e^{33d}$. Apply the sampling procedure of
Section~\ref{sec:gaussian-sampling} to $\Phi_0$ in dimension $d$.
Set $R=e^{35d}$ and choose
$D=2^{\lceil\log_2(2\sqrt d\,R)\rceil}$ and run the procedure independently $d+\lceil c\log_2 n\rceil$
times with the same evaluation elements $g_1,\ldots,g_d$.
For each fixed sequence satisfying the basis bound,
the classical post-processing in
Theorem~\ref{thm:subgroup-decomposition} recovers an integral
basis of $\ker\Phi_0$ with failure probability $O(n^{-c})$.
The total failure probability, including failure of the
basis bound, is $O(n^{-c})$.

On successful recovery, apply the Smith normal form procedure
in Theorem~\ref{thm:subgroup-decomposition} to the recovered
basis matrix.
This computes the invariant factors of $H_0$ and integer
expressions for the corresponding cyclic generators in terms
of $g_1,\ldots,g_d$.
Evaluate these expressions using the classical group operations.
The orders of these generators are the corresponding
invariant factors.

\paragraph{Cyclic-factor partition and auxiliary elements.}
Process the remaining elements of $\mathcal S$ in batches of
size $b=\lceil\sqrt n\rceil$.
Pad the last batch with $0_G$ if necessary.
Fix an extension $k$ for which an invariant-factor decomposition
of $H_k$ and its cyclic generators are known.
Let $w_1,\ldots,w_b$ be the next batch, and
$H_{k+1}:=\langle H_k,w_1,\ldots,w_b\rangle$. Partition the cyclic factors of $H_k$ at the threshold
$B=2^{\sqrt n}$.
Let $y_1,\ldots,y_r$ be the generators of the factors with
orders $2\leq N_i\leq B$.
Let $y'_1,\ldots,y'_s$ be the generators of the factors with
orders $M_i>B$.
Set $H^{-}:=\langle y_1,\ldots,y_r\rangle$ and
$H^{+}:=\langle y'_1,\ldots,y'_s\rangle$.
Then $H_k=H^{-}\oplus H^{+}$.
The corresponding cyclic decompositions are
$H^{-}\cong\prod_{i=1}^{r}\mathbb Z_{N_i}$ and
$H^{+}\cong\prod_{i=1}^{s}\mathbb Z_{M_i}$
where $r=r(H^{-})$ and $s=r(H^{+})$.
Since $\prod_i M_i\leq|H_k|\leq2^n$ and each
$M_i>2^{\sqrt n}$, we have $s\leq\sqrt n$.

Fix the current decomposition and target batch before sampling
the coefficients.
Let $S:=\langle H^{+},w_1,\ldots,w_b\rangle$, $M=2^{2n}$ and choose all coefficients
$\alpha_{j,i}$ and $\beta_{j,l}$ independently and uniformly
from $\{0,\ldots,M-1\}$.
For $1\leq j\leq d$, define
$a_j:=\sum_{i=1}^{s}\alpha_{j,i}y'_i$ and
$h_j:=a_j+\sum_{l=1}^{b}\beta_{j,l}w_l$.
Record the coefficients for the character-extension congruences
in Section~\ref{sec:exact-character-recovery}.

Let $m=d+b$ and 
$\mathbf t=(\mathbf z,\mathbf u)\in
\mathbb Z^d\times\mathbb Z^b$.
Define the evaluation homomorphism
$\Phi:\mathbb Z^m\longrightarrow G$ by
\begin{equation}\label{eq:local-evaluation}
\Phi(\mathbf t)
=\sum_{j=1}^{d}z_jh_j+\sum_{l=1}^{b}u_lw_l.
\end{equation}
The relations among $h_1,\ldots,h_d,w_1,\ldots,w_b$
form the lattice
\begin{equation}\label{eq:local-relations}
\mathcal L:=\ker\Phi
=\left\{(\mathbf z,\mathbf u)\in\mathbb Z^m:
\sum_{j=1}^{d}z_jh_j+\sum_{l=1}^{b}u_lw_l=0_G\right\}.
\end{equation}
The first isomorphism theorem implies
$\mathbb Z^m/\mathcal L\cong\operatorname{im}\Phi$.
Hence
\begin{equation}\label{eq:local-determinant}
\det\mathcal L
=\bigl|\langle h_1,\ldots,h_d,w_1,\ldots,w_b\rangle\bigr|
\leq2^n.
\end{equation}
Sincee $y'_1,\ldots,y'_s,w_1,\ldots,w_b$ generate $S$,
then apply Lemma~\ref{lem:random-combinations} to the surjective
homomorphism $\mathbb Z^{s+b}\to S$ defined by these generators.
The distribution of each $h_j$ has total variation distance
at most $(s+b)|S|/(4M)$ from $U(S)$.
The coefficient vectors for distinct $j$ are independent.
By the product bound for total variation distance, the joint
distribution of $h_1,\ldots,h_d$ has distance at most
$d(s+b)|S|/(4M)=O(n2^{-n})$ from $U(S)^{\otimes d}$.
Here we use $d,s+b=O(\sqrt n)$, $|S|\leq2^n$ and $M=2^{2n}$.
Consider independent uniform samples $h_1,\ldots,h_d$
from $S$.
Apply Corollary~\ref{cor:uniform-short-basis}
to the relation lattice $\mathcal L=\ker\Phi$
in Equation~\eqref{eq:local-relations},
with fixed elements $w_1,\ldots,w_b$.
With probability at least $1-2e^{-2m^2}$,
this lattice has an integral basis whose vectors have
Euclidean norm less than $e^{33m}$.
By the preceding total variation estimate, the failure
probability under the sampled coefficients is $O(n2^{-n})$.

For the elements $a_j$, apply
Lemma~\ref{lem:random-combinations} to the surjective homomorphism
$\mathbb Z^s\to H^{+}$ defined by $y'_1,\ldots,y'_s$.
The distribution of each $a_j$ has total variation distance
at most $s|H^{+}|/(4M)$ from $U(H^{+})$.
Independence across $j$ bounds the distance of their joint
distribution from $U(H^{+})^{\otimes d}$ by
$ds|H^{+}|/(4M)=O(n2^{-n})$.
Since $d-s\geq\lceil c\log_2 n\rceil$,
Corollary~\ref{cor:generation-probability} bounds the generation
failure probability under $U(H^{+})^{\otimes d}$ by $O(n^{-c})$.
The total variation estimate implies the same asymptotic bound
for the sampled elements 
$
\Pr[\langle a_1,\ldots,a_d\rangle\neq H^{+}]
=O(n^{-c}).
$
These probability bounds hold for every fixed current
decomposition and target batch. The definitions of $a_j$ and $h_j$ imply
$
\langle h_1,\ldots,h_d,w_1,\ldots,w_b\rangle
=
\langle a_1,\ldots,a_d,w_1,\ldots,w_b\rangle.
$
On the event $\langle a_1,\ldots,a_d\rangle=H^{+}$,
this subgroup is $S$.

\subsection{Joint Fourier Sampling and Quantum Resources}
\label{sec:joint-sampling}

Fix the current subgroup decomposition, the target batch,
and the auxiliary elements from
Section~\ref{sec:extension-setup}.
Identify $H^{-}$ with $\prod_{i=1}^{r}\mathbb Z_{N_i}$
through the isomorphism
$\mathbf x\mapsto\sum_{i=1}^{r}x_i y_i$.
Let $\mathcal G=H^{-}\times\mathbb Z^m$ and define
$\Psi:\mathcal G\to G$ by
\begin{equation}\label{eq:joint-evaluation}
\Psi(\mathbf x,\mathbf t)
=\sum_{i=1}^{r}x_i y_i+\Phi(\mathbf t).
\end{equation}
where 
$
P:=\ker\Psi.
$
The local relation lattice satisfies
$
\mathcal L
=\{\mathbf t\in\mathbb Z^m:(0,\mathbf t)\in P\}.
$

For the integer coordinates, take the Gaussian width
$R=e^{35m}$ and the grid modulus $D=2^\ell$ as in
Section~\ref{sec:parameters-and-evaluation}.
Here $\ell=\lceil\log_2(2\sqrt m\,R)\rceil$ is the number
of qubits in each integer-coordinate register. Let $q=m+\lceil c\log_2n\rceil$ be the number of samples
used for lattice recovery.
For each subgroup extension, collect $Q=r+q$ independent
joint samples and use the lattice components of the first
$q$ samples to recover $\mathcal L$.

\paragraph{State preparation.}
The initial state has uniform amplitudes on the
cyclic coordinates of $H^{-}$ and Gaussian amplitudes
on the integer coordinates:
\[
|\Psi_0\rangle\propto
\sum_{\mathbf x\in H^{-}}
\sum_{\mathbf t\in\mathbb Z_D^m}
\rho_R(\mathbf t)
|\mathbf x,\mathbf t\rangle|0\rangle.
\]
Here $\mathbf x$ is represented in the known cyclic
coordinates of $H^{-}$, and $\mathbf t$ uses the centered
representatives $\{-D/2,\ldots,D/2-1\}^m$.
The final register is initialized to the all-zero bit string. Since $H^{-}\cong\prod_{i=1}^{r}\mathbb Z_{N_i}$,
its uniform superposition is the tensor product of the states
$N_i^{-1/2}\sum_{x_i=0}^{N_i-1}|x_i\rangle$.
Each factor is prepared from $|0\rangle$ by the exact
Fourier transform over $\mathbb Z_{N_i}$,
using the construction of Mosca and Zalka~\cite{mosca2004exact}.
The integer-coordinate register is prepared by the Gaussian
procedure of Section~\ref{sec:gaussian-sampling}.

\paragraph{Oracle evaluation.}
The evaluation of $\Psi(\mathbf x,\mathbf t)$ consists
of two parts.
We compute $\Phi(\mathbf t)$ using the Fibonacci evaluation
of Lemma~\ref{lem:fibonacci-evaluation} and compute $\sum_{i=1}^{r}x_i y_i$ from the binary expansions
of the $x_i$ and the precomputed multiples $2^j y_i$.
Their sum is written into the output register.
The oracle acts as
\begin{equation}\label{eq:clean-joint-evaluation}
U_{\Psi}:\quad
|\mathbf x,\mathbf t\rangle|0\rangle|0\rangle
\longmapsto
|\mathbf x,\mathbf t\rangle|0\rangle
|\Psi(\mathbf x,\mathbf t)\rangle.
\end{equation}
For the initial state, the joint state after
evaluation is proportional to
\[
\sum_{\mathbf x\in H^{-}}
\sum_{\mathbf t\in\mathbb Z_D^m}
\rho_R(\mathbf t)
|\mathbf x,\mathbf t\rangle
|\Psi(\mathbf x,\mathbf t)\rangle.
\]
Measure the output register in the computational basis.
The measured bit string uniquely identifies a group element
$h\in\operatorname{im}\Psi$.
Conditioned on this outcome, the input state is supported
on the pairs satisfying $\Psi(\mathbf x,\mathbf t)=h$.
Let $P=\ker\Psi$.
For any $(\mathbf x_h,\mathbf t_h)$ satisfying
$\Psi(\mathbf x_h,\mathbf t_h)=h$, the fiber over $h$
is the coset $(\mathbf x_h,\mathbf t_h)+P$.
For the Gaussian preparation, the conditional
input state is therefore
\begin{equation}\label{eq:joint-coset-state}
|\Psi_h\rangle\propto
\sum_{\substack{
(\mathbf x,\mathbf t)\in(\mathbf x_h,\mathbf t_h)+P\\
\mathbf x\in H^{-},\ \mathbf t\in\mathbb Z_D^m}}
\rho_R(\mathbf t)|\mathbf x,\mathbf t\rangle.
\end{equation}
The measurement outcome is discarded.

The quotient $\mathcal G/P$ is isomorphic to
$\operatorname{im}\Psi$.
On the event $\langle a_1,\ldots,a_d\rangle=H^{+}$,
we have $\operatorname{im}\Psi=H_{k+1}$ and hence
$\mathcal G/P\cong H_{k+1}$. For $\eta\in\widehat{H^{-}}$, let
$\boldsymbol\eta\in\prod_{i=1}^{r}\mathbb Z_{N_i}$
be its cyclic label.
We write its additive phase as
$
\eta(\mathbf x)
=\sum_{i=1}^{r}\frac{\eta_i x_i}{N_i}\pmod1.
$
A character of $\mathcal G$ is represented by
$(\eta,\mathbf v)\in\widehat{H^{-}}\times\mathbb T^m$.
Write $\mathbf v=(\mathbf p,\mathbf q)$, where
$\mathbf p\in\mathbb T^d$ and $\mathbf q\in\mathbb T^b$.
The annihilator $P^\perp$ consists of the pairs
$(\eta,\mathbf v)$ satisfying
$
\eta(\mathbf x)+\langle\mathbf v,\mathbf t\rangle
\equiv0\pmod1
$ for every $(\mathbf x,\mathbf t)\in P$.

\paragraph{Fourier transformation and measurement.}
Apply the exact Fourier transform ~\cite{mosca2004exact} over $H^{-}$ to the
cyclic coordinate registers.
Apply the approximate Fourier transform~\cite{Coppersmith94} over $\mathbb Z_D^m$
to the integer-coordinate registers,
with the precision specified in
Section~\ref{sec:gaussian-sampling}.
We use the positive-exponent convention.
For the Gaussian state and the exact grid transform,
Equation~\eqref{eq:joint-coset-state} becomes
\begin{equation}\label{eq:joint-fourier-state}
\begin{aligned}
&(\mathrm{QFT}_{H^{-}}\otimes
  \mathrm{QFT}_{\mathbb Z_D^{m}})|\Psi_h\rangle\\
&\quad\propto
\sum_{\eta\in\widehat{H^{-}}}
\sum_{\mathbf w\in\{0,\ldots,D-1\}^{m}}
\left(
\sum_{\substack{
(\mathbf x,\mathbf t)\in(\mathbf x_h,\mathbf t_h)+P\\
\mathbf x\in H^{-},\ \mathbf t\in\mathbb Z_D^{m}}}
\rho_R(\mathbf t)
e^{2\pi i(
\eta(\mathbf x)+\langle\mathbf t,\mathbf w\rangle/D)}
\right)|\eta,\mathbf w\rangle.
\end{aligned}
\end{equation}
Measure the cyclic and integer-coordinate registers in
the computational basis.
In execution $t$, denote the measured cyclic label by
$\widehat{\boldsymbol\eta}_t$ and the measured grid vector by
$\mathbf w_t\in\{0,\ldots,D-1\}^{m}$.
Split $\mathbf w_t/D$ into its first $d$ and last $b$
coordinates:
\[
\widetilde{\mathbf v}_t
=\left(\widetilde{\mathbf p}_t,
       \widetilde{\mathbf q}_t\right)
=\frac{\mathbf w_t}{D}\in[0,1)^m.
\]
The recorded sample is
$
\left(\widehat{\boldsymbol\eta}_t,
      \widetilde{\mathbf p}_t,
      \widetilde{\mathbf q}_t\right).
$

\begin{proposition}
\label{prop:joint-sampling-resources}
Use the black-box model and gate conventions of
Section~\ref{sec:computational-model}.
Let $|G|\leq 2^n$ and
$T_{\mathrm{op}}=\Omega(\sqrt n)$.
Fix a known decomposition of $H_k$, a target batch,
and the coefficients defining $h_1,\ldots,h_d$ as in
Section~\ref{sec:extension-setup}.

One execution of the joint sampling circuit uses
$O(n)$ qubits, $O(n)$ reversible group operations, and
$\softO(nT_{\mathrm{op}}+n^{3/2})
=\softO(nT_{\mathrm{op}})$ elementary gates.
\end{proposition}

\begin{proof}
Proposition~\ref{prop:sampling-resources} bounds the
resources for Gaussian preparation, Fibonacci evaluation
of $\Phi$, and Fourier transforms on the integer-coordinate
registers.
These steps use $O(n)$ qubits, $O(n)$ reversible group
operations, and
$\softO(nT_{\mathrm{op}})$ gates.

Let $\ell_i=\lceil\log_2N_i\rceil$ for $1\leq i\leq r$.
Since $N_i\geq2$ and
$\prod_{i=1}^{r}N_i=|H^{-}|\leq2^n$,
the cyclic-coordinate registers use
\[
\sum_{i=1}^{r}\ell_i
\leq \log_2|H^{-}|+r
\leq2n
\]
qubits.
The threshold $N_i\leq B=2^{\sqrt n}$ also implies
$\ell_i\leq\sqrt n+1$. In the Mosca--Zalka construction,
Fourier-state preparation and eigenvalue estimation
can each be implemented with $O(\ell_i^2)$ gates
and $O(\ell_i)$ auxiliary qubits.
These bounds follow by counting the controlled phase
gates and standard reversible arithmetic operations.
The phase inversion on the good subspace has the same
resource bounds.
Exact amplitude amplification and uncomputation require
only a constant number of these operations
\cite[Secs.~2, 4, and~5.1]{mosca2004exact}.
Thus the exact Fourier transform over $H^{-}$ uses
$O(\sum_i\ell_i^2)$ gates, where
\begin{equation}\label{eq:cyclic-qft-cost}
\sum_{i=1}^{r}\ell_i^2
\leq(\sqrt n+1)\sum_{i=1}^{r}\ell_i
\leq2n(\sqrt n+1)
=O(n^{3/2}).
\end{equation}
The uniform state on $H^{-}$ is prepared by applying
the same Fourier transform to $|0\rangle$.
By Equation~\eqref{eq:cyclic-qft-cost}, this requires
$O(n^{3/2})$ gates.
The auxiliary registers are reused across cyclic factors
and use $O(\sqrt n)$ qubits. Since the cost of evaluating $\Phi$ is included above,
it remains to count the additions for
$\sum_{i=1}^{r}x_i y_i$ in the evaluation of $\Psi$.
Using the binary expansions of the $x_i$ and the
precomputed multiples $2^j y_i$, the computation and
uncomputation require at most
$2\sum_{i=1}^{r}\ell_i\leq4n$ controlled group additions.
Their gate cost is therefore $O(nT_{\mathrm{op}})$.

Thus one execution of the joint sampling circuit uses
$O(n)$ qubits and $O(n)$ reversible group operations.
Its total gate count is
$\softO(nT_{\mathrm{op}}+n^{3/2})
=\softO(nT_{\mathrm{op}})$,
where the equality follows from
$T_{\mathrm{op}}=\Omega(\sqrt n)$.
\end{proof}

\subsection{Joint Sample Distributions and Character Extension}
\label{sec:joint-sample-distribution}

The joint kernel \(P\) relates the cyclic coordinates of \(H^{-}\) to the local relation lattice \(\mathcal L\). The following lemma shows that \(\mathcal L\) has finite index in \(\mathbb Z^{m}\) and that the joint quotient \(\mathcal G/P\) is finite. We then analyze the joint Fourier outcomes through the annihilator \(P^\perp\), identify the distribution of their lattice components, and derive the congruences used to recover character values on the known cyclic coordinates of \(H^{+}\).

\begin{lemma}
\label{lem:finite-joint-quotient}
Fix the setup of Section~\ref{sec:extension-setup}. Let
\(\mathcal G=H^{-}\times\mathbb Z^{m}\), \(P=\ker\Psi\), and let
$
\mathcal L
=
\{\mathbf t\in\mathbb Z^{m}:(0,\mathbf t)\in P\}
=
\ker\Phi.
$
Then \(\mathcal L\) is a full-rank sublattice of \(\mathbb Z^{m}\), with
\[
[\mathbb Z^{m}:\mathcal L]
=
|\operatorname{im}\Phi|
=
\bigl|\langle h_1,\ldots,h_d,w_1,\ldots,w_b\rangle\bigr|
\leq |G|
\leq 2^n.
\]
Moreover,
$
\mathcal G/P\cong\operatorname{im}\Psi\leq G,
$
so \(\mathcal G/P\) is a finite Abelian group.
\end{lemma}

\begin{proof}
By the definitions of \(P\) and \(\Psi\), a vector
\(\mathbf t\in\mathbb Z^{m}\) satisfies \((0,\mathbf t)\in P\) if and only if
\(\Phi(\mathbf t)=0_G\). Hence \(\mathcal L=\ker\Phi\).
The First Isomorphism Theorem gives
\[
\mathbb Z^{m}/\mathcal L
\cong
\operatorname{im}\Phi
=
\langle h_1,\ldots,h_d,w_1,\ldots,w_b\rangle.
\]
The image is a finite subgroup of \(G\), which proves the index bound stated in the lemma. In particular, \(\mathcal L\) has finite index in \(\mathbb Z^{m}\) and therefore has rank \(m\).

Applying the First Isomorphism Theorem  to
\(\Psi:\mathcal G\to G\), gives
$
\mathcal G/P
\cong
\operatorname{im}\Psi.
$
Since \(\operatorname{im}\Psi\leq G\), the quotient \(\mathcal G/P\) is finite.
\end{proof}

As recalled in Section~\ref{sec:groups-characters-lattices},
\(\widehat{\mathbb Z^{m}}\cong\mathbb T^{m}\). Hence
$
\widehat{\mathcal G}
\cong
\widehat{H^{-}}\times\mathbb T^{m}.
$
Let
\[
P^\perp
=
\left\{
(\eta,\mathbf v)\in
\widehat{H^{-}}\times\mathbb T^{m}:
\eta(\mathbf x)+\langle\mathbf v,\mathbf t\rangle
\equiv0\pmod 1
\text{ for every }(\mathbf x,\mathbf t)\in P
\right\}
\]
be the annihilator of \(P\). Pullback along the quotient map
\(\mathcal G\to\mathcal G/P\) identifies
\(\widehat{\mathcal G/P}\) with \(P^\perp\). Since
\(\mathcal G/P\) is finite, \(P^\perp\) is finite and
$
|P^\perp|=|\mathcal G/P|.
$
We use the uniform distribution on \(P^\perp\) below.

\begin{lemma}\label{lem:joint-dual-sampling}
Let $\mathcal G=H^{-}\times\mathbb Z^{m}$, and let
$P=\ker\Psi$ be as in Section~\ref{sec:joint-sampling}.
Take $R$ and $D$ as in
Section~\ref{sec:parameters-and-evaluation}, with lattice
dimension $m$. For $\mathbf v\in\mathbb T^{m}$ and
$\mathbf w\in\{0,\ldots,D-1\}^{m}$, define
\[
p_{\mathbf v}(\mathbf w/D)
:=
\frac{
\rho_{1/(\sqrt2R)}
(\mathbf v-\mathbf w/D+\mathbb Z^{m})
}{
\rho_{1/(\sqrt2R)}
(\mathbf v-D^{-1}\mathbb Z^{m})
}.
\]
The denominator is the normalizing factor, since
\[
\sum_{\mathbf w\in\{0,\ldots,D-1\}^{m}}
\rho_{1/(\sqrt2R)}
(\mathbf v-\mathbf w/D+\mathbb Z^{m})
=
\rho_{1/(\sqrt2R)}
(\mathbf v-D^{-1}\mathbb Z^{m}).
\]
This is the grid Gaussian distribution of
\cite[Eq.~(3)]{Regev25}, with dimension $m$.

For the ideal circuit with exact Gaussian preparation and exact
Fourier transforms, the measured pair
$(\widehat\eta,\widetilde{\mathbf v})$ is
$O(2^{-m})$-close in total variation distance to the following
experiment: choose $(\eta,\mathbf v)$ uniformly from $P^\perp$
and, conditional on this choice, sample
$\widetilde{\mathbf v}$ according to $p_{\mathbf v}$.

With the Gaussian-state preparation and approximate lattice
Fourier transform prescribed in
Section~\ref{sec:gaussian-sampling}, the total variation bound is
$
O\!\left(n^{-(c+1)}\right).
$
Consequently, the implemented outcome can be coupled to a uniform
$(\eta,\mathbf v)\in P^\perp$ so that
$\widehat\eta=\eta$ and
\[
\operatorname{dist}_{\mathbb T^{m}}
(\widetilde{\mathbf v},\mathbf v)
\leq \frac{\sqrt m}{\sqrt2R}
\]
except with probability
$O(n^{-(c+1)})$.
\end{lemma}

\begin{proof}
We first consider the ideal circuit. Sums over
$\mathbb Z_D^{m}$ use the centered representatives
$\{-D/2,\ldots,D/2-1\}^{m}$. For
$\mathbf v\in\mathbb T^{m}$, put
\[
|\psi_{\mathbf v}\rangle
=
Z^{-1/2}
\sum_{\mathbf t\in\mathbb Z_D^{m}}
\rho_R(\mathbf t)
e^{-2\pi i\langle\mathbf v,\mathbf t\rangle}
|\mathbf t\rangle,
\]
where $
Z=
\sum_{\mathbf t\in\mathbb Z_D^{m}}
\rho_R(\mathbf t)^2.
$
After tracing out the oracle-output register, only matrix elements \(|\mathbf x,\mathbf t\rangle\langle\mathbf x',\mathbf t'|\) satisfying \((\mathbf x-\mathbf x',\mathbf t-\mathbf t')\in P\) remain.
Since $\mathcal G/P$ is finite,
character orthogonality gives
\[
\mathbf1_P(\mathbf x,\mathbf t)
=
\frac1{|P^\perp|}
\sum_{(\eta,\mathbf v)\in P^\perp}
e^{-2\pi i(
\eta(\mathbf x)+\langle\mathbf v,\mathbf t\rangle)}.
\]
Define
$|\lambda_\eta\rangle
=
|H^{-}|^{-1/2}
\sum_{\mathbf x\in H^{-}}
e^{-2\pi i\eta(\mathbf x)}|\mathbf x\rangle.
$
The reduced input state is therefore
\[
\frac1{|P^\perp|}
\sum_{(\eta,\mathbf v)\in P^\perp}
|\lambda_\eta\rangle\langle\lambda_\eta|
\otimes
|\psi_{\mathbf v}\rangle\langle\psi_{\mathbf v}|.
\]
With the positive-exponent convention, the Fourier transform over
$H^{-}$ maps $|\lambda_\eta\rangle$ to $|\eta\rangle$.
It remains to determine the Fourier measurement of
$|\psi_{\mathbf v}\rangle$, uniformly in $\mathbf v$.

Define the unnormalized periodized state
\[
|h_{\mathbf v}\rangle
=
\sum_{\mathbf t\in\mathbb Z_D^{m}}
\left(
\sum_{\mathbf r\in\mathbb Z^{m}}
\rho_R(\mathbf t+D\mathbf r)
e^{-2\pi i
\langle\mathbf v,\mathbf t+D\mathbf r\rangle}
\right)|\mathbf t\rangle.
\]
Taking absolute values removes the modulation, and the truncation
estimate used in the proof of
\cite[Claim~A.5]{Regev25}, specialized to the lattice
$\mathbb Z^{m}$, gives
\[
\left\|
\sqrt Z\,|\psi_{\mathbf v}\rangle
-|h_{\mathbf v}\rangle
\right\|_2
\leq
2^{-m}\|h_{\mathbf0}\|_2.
\]
The normalization estimates in
\cite[Claims~A.4--A.5]{Regev25} yield
$
Z=(R/\sqrt2)^m(1+O(2^{-m}))
$
and
$
\|h_{\mathbf0}\|_2^2
=(R/\sqrt2)^m(1+O(2^{-m})).
$
It follows uniformly in $\mathbf v$ that
$
\|h_{\mathbf v}\|_2^2
=(R/\sqrt2)^m(1+O(2^{-m}))
$
and hence
\[
\left\|
|\psi_{\mathbf v}\rangle
-\frac{|h_{\mathbf v}\rangle}{\|h_{\mathbf v}\|_2}
\right\|_2
=O(2^{-m}).
\]
Applying the same Fourier transform and measurement to these two
states changes the resulting distributions by
$O(2^{-m})$ in total variation distance.

For
$\mathbf w\in\{0,\ldots,D-1\}^{m}$, Poisson summation gives the
Fourier coefficient
\[
D^{-m/2}
\sum_{\mathbf t\in\mathbb Z^{m}}
\rho_R(\mathbf t)
e^{2\pi i
\langle\mathbf w/D-\mathbf v,\mathbf t\rangle}
=
R^mD^{-m/2}
\rho_{1/R}
(\mathbf v-\mathbf w/D+\mathbb Z^{m}).
\]
Let $p'_{\mathbf v}$ denote the measurement distribution obtained
from the normalized state
$|h_{\mathbf v}\rangle/\|h_{\mathbf v}\|_2$.
Expanding the square of the periodized Gaussian, the diagonal terms give
\[
p'_{\mathbf v}(\mathbf w/D)
\geq
\frac{R^{2m}D^{-m}}{\|h_{\mathbf v}\|_2^2}
\rho_{1/(\sqrt2R)}
(\mathbf v-\mathbf w/D+\mathbb Z^{m}).
\]
Thus
$
p'_{\mathbf v}\geq
\alpha_{\mathbf v}p_{\mathbf v},
$
where
\[
\alpha_{\mathbf v}
=
\frac{R^{2m}D^{-m}}{\|h_{\mathbf v}\|_2^2}
\rho_{1/(\sqrt2R)}
(\mathbf v-D^{-1}\mathbb Z^{m}).
\]

A second application of Poisson summation gives
\[
\rho_{1/(\sqrt2R)}
(\mathbf v-D^{-1}\mathbb Z^{m})
=
\left(\frac{D}{\sqrt2R}\right)^m
\sum_{\mathbf r\in D\mathbb Z^{m}}
\rho_{\sqrt2R}(\mathbf r)
e^{2\pi i\langle\mathbf r,\mathbf v\rangle}.
\]
Since $D\geq2\sqrt m\,R$,
\cite[Corollary~A.2]{Regev25} gives
$
\rho_{\sqrt2R}
(D\mathbb Z^{m}\setminus\{\mathbf0\})
=O(2^{-m}),
$
uniformly in $\mathbf v$. Hence
$
\alpha_{\mathbf v}=1+O(2^{-m}).
$
The pointwise lower bound and the fact that
$p'_{\mathbf v}$ is a probability distribution imply
$\alpha_{\mathbf v}\leq1$. Therefore
$
p'_{\mathbf v}
=
\alpha_{\mathbf v}p_{\mathbf v}
+r_{\mathbf v},
(
r_{\mathbf v}\geq0),
$
with
\[
\sum_{\mathbf w}r_{\mathbf v}(\mathbf w/D)
=
1-\alpha_{\mathbf v}
=
O(2^{-m}).
\]
It follows that
$
\operatorname{TV}
(p'_{\mathbf v},p_{\mathbf v})
=
O(2^{-m})
$
uniformly in $\mathbf v$.

Averaging over the uniform mixture on $P^\perp$ proves the
$O(2^{-m})$ bound for the ideal circuit. The Gaussian-state
preparation and approximate lattice Fourier transform used in the
implemented circuit each contribute
$O(n^{-(c+1)})$ to the output distance, by the bounds in
Section~\ref{sec:gaussian-sampling} and contractivity of trace
distance. This gives
$
O\!\left(2^{-m}+n^{-(c+1)}\right)=O(n^{-(c+1)})
$
for the implemented sampler.

Finally, \cite[Claim~A.7]{Regev25} gives, uniformly in
$\mathbf v$,
\[
\Pr_{\widetilde{\mathbf v}\sim p_{\mathbf v}}
\left[
\operatorname{dist}_{\mathbb T^m}
(\widetilde{\mathbf v},\mathbf v)
>
\frac{\sqrt m}{\sqrt2R}
\right]
=
O(2^{-m}).
\]
Combining this tail estimate with a maximal coupling and using \(2^{-m}=o(n^{-(c+1)})\) gives the stated \(O(n^{-(c+1)})\) failure bound.
\end{proof}

We identify \(\mathcal L^*/\mathbb Z^{m}\) with the subgroup of
\(\mathbb T^{m}\) consisting of the classes that are trivial on
\(\mathcal L\):
\[
\mathcal L^*/\mathbb Z^{m}
=
\left\{
\mathbf v\in\mathbb T^{m}:
\langle\mathbf v,\mathbf t\rangle
\equiv0\pmod1
\text{ for every }\mathbf t\in\mathcal L
\right\}.
\]

\begin{lemma}
\label{lem:uniform-projection}
Let \(P\leq\mathcal G=H^{-}\times\mathbb Z^{m}\) and
$
\mathcal L
=
\{\mathbf t\in\mathbb Z^{m}:(0,\mathbf t)\in P\}
$
be as above. The projection onto the lattice component,
$
\pi_2:P^\perp\longrightarrow\mathcal L^*/\mathbb Z^{m}
$
by $
(\eta,\mathbf v)\longmapsto\mathbf v,
$
is a surjective homomorphism. Consequently, if
\((\eta,\mathbf v)\) is uniform on \(P^\perp\), then
\(\mathbf v\) is uniform on \(\mathcal L^*/\mathbb Z^{m}\).
\end{lemma}

\begin{proof}
For \((\eta,\mathbf v)\in P^\perp\), every
\((\mathbf x,\mathbf t)\in P\) satisfies
\begin{equation}\label{eq:joint-annihilator}
\eta(\mathbf x)+\langle\mathbf v,\mathbf t\rangle
\equiv0\pmod1.
\end{equation}
If \(\mathbf t\in\mathcal L\), then \((0,\mathbf t)\in P\), and
Equation~\eqref{eq:joint-annihilator} gives
$
\langle\mathbf v,\mathbf t\rangle\equiv0\pmod1.
$
Hence \(\mathbf v\in\mathcal L^*/\mathbb Z^{m}\), so the
second-coordinate projection restricts to a homomorphism
$
\pi_2:P^\perp\to\mathcal L^*/\mathbb Z^{m}.
$
It remains to prove surjectivity. Fix
\(\mathbf v\in\mathcal L^*/\mathbb Z^{m}\), and let
\(\pi_1(P)\leq H^{-}\) denote the projection of \(P\) onto the
first factor. For \(\mathbf x\in\pi_1(P)\), choose
\(\mathbf t\in\mathbb Z^{m}\) such that
\((\mathbf x,\mathbf t)\in P\), and set
$
\phi(\mathbf x)
:=
-\langle\mathbf v,\mathbf t\rangle
\pmod1.
$
If \((\mathbf x,\mathbf t')\in P\) as well, then
\((0,\mathbf t-\mathbf t')\in P\), so
\(\mathbf t-\mathbf t'\in\mathcal L\). Since
\(\mathbf v\) is trivial on \(\mathcal L\),
$
\langle\mathbf v,\mathbf t-\mathbf t'\rangle
\equiv0\pmod1,
$
and therefore \(\phi\) is well defined. Moreover, if
\((\mathbf x,\mathbf t),(\mathbf x',\mathbf t')\in P\), then
\[
\phi(\mathbf x+\mathbf x')
=
-\langle\mathbf v,\mathbf t+\mathbf t'\rangle
\equiv
\phi(\mathbf x)+\phi(\mathbf x')
\pmod1.
\]
Thus \(\phi\) is a character of \(\pi_1(P)\) in additive phase
notation.
The restriction map
$
\widehat{H^{-}}\longrightarrow\widehat{\pi_1(P)}
$
is surjective, so \(\phi\) extends to some
\(\eta\in\widehat{H^{-}}\). For every
\((\mathbf x,\mathbf t)\in P\),
\[
\eta(\mathbf x)+\langle\mathbf v,\mathbf t\rangle
=
\phi(\mathbf x)+\langle\mathbf v,\mathbf t\rangle
\equiv0\pmod1.
\]
Hence \((\eta,\mathbf v)\in P^\perp\), which proves that
\(\pi_2\) is surjective.

Since \(\pi_2\) is a surjective homomorphism of finite groups,
each fiber has cardinality \(|\ker\pi_2|\), and
$
|P^\perp|
=
|\ker\pi_2|\,
|\mathcal L^*/\mathbb Z^{m}|.
$
Therefore, for every
\(\mathbf v_0\in\mathcal L^*/\mathbb Z^{m}\),
\[
\Pr[\mathbf v=\mathbf v_0]
=
\frac{|\ker\pi_2|}{|P^\perp|}
=
\frac{1}{|\mathcal L^*/\mathbb Z^{m}|}.
\]
Thus the lattice component is uniform on
\(\mathcal L^*/\mathbb Z^{m}\).
\end{proof}

For an exact lattice phase
\(\mathbf v=(\mathbf p,\mathbf q)\in\mathcal L^*/\mathbb Z^{m}\),
the character
$
\mathbf t\longmapsto
e^{2\pi i\langle\mathbf v,\mathbf t\rangle}
$
of \(\mathbb Z^{m}\) is trivial on
\(\mathcal L=\ker\Phi\). It therefore descends to a character
\(\chi_{\mathbf v}\) of \(\operatorname{im}\Phi\).
Since \(h_j,w_l\in\operatorname{im}\Phi\) and
\[
a_j=h_j-\sum_{l=1}^{b}\beta_{j,l}w_l,
\]
we also have \(a_j\in\operatorname{im}\Phi\). Let
$
A:=\langle a_1,\ldots,a_d\rangle\leq H^{+} $ and $
\psi:=\chi_{\mathbf v}|_A.
$
Then
\[
\psi(a_j)
=
\exp\!\left(
2\pi i
\left(
p_j-\sum_{l=1}^{b}\beta_{j,l}q_l
\right)
\right),
\qquad 1\leq j\leq d.
\]
The next lemma expresses the extensions of \(\psi\) in the known
cyclic coordinates of \(H^{+}\).

\begin{lemma}
\label{lem:character-extension}
Let
$
H^{+}
=
\bigoplus_{i=1}^{s}\langle y'_i\rangle,$ where $
\operatorname{ord}(y'_i)=M_i
$
and let
$
a_j=\sum_{i=1}^{s}\alpha_{j,i}y'_i$,
for $1\leq j\leq d.
$
Set \(A=\langle a_1,\ldots,a_d\rangle\). For an exact lattice
phase
\(\mathbf v=(\mathbf p,\mathbf q)\in\mathcal L^*/\mathbb Z^{m}\),
let \(\psi\in\widehat A\) be the character defined above. The congruence system
\begin{equation}\label{eq:character-extension}
\sum_{i=1}^{s}
\frac{\alpha_{j,i}\xi_i}{M_i}
\equiv
p_j-\sum_{l=1}^{b}\beta_{j,l}q_l
\pmod1,
\qquad
1\leq j\leq d,
\end{equation}
has exactly
$
[H^{+}:A]
=
\frac{|H^{+}|}{|A|}
$
solutions
\(\boldsymbol\xi\in\prod_{i=1}^{s}\mathbb Z_{M_i}\).
Under the natural identification
$
\widehat{H^{+}}
\cong
\prod_{i=1}^{s}\mathbb Z_{M_i},
$
these solutions are precisely the characters of \(H^{+}\) whose
restriction to \(A\) is \(\psi\). Hence the system always has a
solution, and that solution is unique if and only if
\(A=H^{+}\).
\end{lemma}

\begin{proof}
The fixed cyclic decomposition identifies
\(\boldsymbol\xi=(\xi_1,\ldots,\xi_s)\in
\prod_i\mathbb Z_{M_i}\)
with the character
\(\chi_{\boldsymbol\xi}\in\widehat{H^{+}}\) defined by
$
\chi_{\boldsymbol\xi}(y'_i)
=
e^{2\pi i\xi_i/M_i}.
$
For each \(j\),
\[
\chi_{\boldsymbol\xi}(a_j)
=
\exp\!\left(
2\pi i
\sum_{i=1}^{s}
\frac{\alpha_{j,i}\xi_i}{M_i}
\right).
\]
Thus Equation~\eqref{eq:character-extension} holds for every
\(j\) if and only if
$
\chi_{\boldsymbol\xi}(a_j)=\psi(a_j),
(1\leq j\leq d).
$
Since the \(a_j\) generate \(A\), this is equivalent to
$
\chi_{\boldsymbol\xi}|_A=\psi.
$ Consider the restriction homomorphism
$
\operatorname{res}_A:
\widehat{H^{+}}\longrightarrow\widehat A$
with $
\chi\longmapsto\chi|_A.
$
It is surjective by the character-extension theorem, and
$
\ker(\operatorname{res}_A)=A^\perp.
$
Therefore the set of characters extending \(\psi\) is a coset of
\(A^\perp\), and hence has cardinality
\[
|A^\perp|
=
|\widehat{H^{+}/A}|
=
|H^{+}/A|
=
[H^{+}:A].
\]
This is also the number of solutions of
Equation~\eqref{eq:character-extension}. The solution is unique
exactly when \([H^{+}:A]=1\), equivalently when \(A=H^{+}\).
\end{proof}

\subsection{Exact Character Recovery and Subgroup Decomposition}
\label{sec:exact-character-recovery}

We first recover the local relation lattice from the lattice
components of the joint samples. Pairings with the short
generators of this lattice then determine the exact dual classes
associated with all measured outcomes. These joint characters are
then expressed in the finite coordinates used for the
subgroup update. Once they generate the corresponding annihilator,
their common kernel yields an integer presentation of
\(H_{k+1}\), from which the next invariant-factor decomposition is
obtained by Smith normal form.

\paragraph{Lattice recovery and exact phases.}
Fix the current subgroup presentation, target batch, and
coefficient sequence, and use the sampling parameters of
Section~\ref{sec:joint-sampling}. Run the joint sampling circuit
independently \(Q\) times with the same evaluation elements. Write
$(\widehat{\boldsymbol\eta}_t,
 \widetilde{\mathbf p}_t,
 \widetilde{\mathbf q}_t),$ $(1\leq t\leq Q),
$
for the measured outcomes, and set
$
\widetilde{\mathbf v}_t
:=
(\widetilde{\mathbf p}_t,\widetilde{\mathbf q}_t)
\in[0,1)^m.
$

Applying Lemma~\ref{lem:joint-dual-sampling} independently to the
\(Q\) executions gives a coupling with independent uniform
characters
$
(\boldsymbol\eta_t,\mathbf p_t,\mathbf q_t)\in P^\perp.
$
Set
$
\mathbf v_t:=(\mathbf p_t,\mathbf q_t)
\in\mathcal L^*/\mathbb Z^m.
$
By Lemma~\ref{lem:uniform-projection}, the classes
\(\mathbf v_1,\ldots,\mathbf v_Q\) are independent and uniform on
\(\mathcal L^*/\mathbb Z^m\). Except with probability
$
O\!\left(Qn^{-(c+1)}\right)=O(n^{-c}),
$
all measured cyclic labels agree with their ideal labels and
\[
\operatorname{dist}_{\mathbb T^m}
(\widetilde{\mathbf v}_t,\mathbf v_t)<\delta
\qquad
(1\leq t\leq Q),
\]
where \(\delta\) is the dyadic error bound from
Section~\ref{sec:parameters-and-evaluation}.

Use the first
$
q=m+\lceil c\log_2n\rceil
$
measured lattice components for lattice recovery. On the
short-basis event, we have
$
\det\mathcal L\leq2^n$ and
$
\max_j\|\mathbf b_j\|_2<e^{33m}.
$
Hence the parameter calculation in the proof of
Theorem~\ref{thm:subgroup-decomposition} verifies the recovery
condition for the present lattice. Lemmas~\ref{lem:lattice-recovery}
and~\ref{lem:sublattice-extraction} then recover vectors in the
augmented lattice whose projections onto the first $m$ coordinates
belong to \(\mathcal L\) and generate \(\mathcal L\), except on a
separation event of probability at most
$
6\cdot2^{-(q-m)}
=
6\cdot2^{-\lceil c\log_2n\rceil}
=
O(n^{-c}).
$
The augmented lattice has dimension
\(m+q=O(\sqrt n)\).

Write
\(\mathbf b_1,\ldots,\mathbf b_\ell\in\mathcal L\)
for  these projected vectors, and let
\(\mathbf b'_j\) denote the corresponding vectors of the augmented
lattice. Since
\(\|\mathbf b_j\|_2\leq\|\mathbf b'_j\|_2\), the recovery condition
gives
\[
\delta\|\mathbf b_j\|_2
<
\frac16
\left(
2^{q-m-2}\det\mathcal L
\right)^{-1/q}
\leq\frac16<\frac12.
\]
Fix \(1\leq t\leq Q\). On the sampling event, choose a
representative
\(\mathbf v_t^\sharp\in\mathcal L^*\) of the class
\(\mathbf v_t\) such that
$
\|\widetilde{\mathbf v}_t-\mathbf v_t^\sharp\|_2<\delta.
$
For each \(1\leq j\leq\ell\), since
\(\mathbf b_j\in\mathcal L\) and
\(\mathbf v_t^\sharp\in\mathcal L^*\), we have
$
\langle\mathbf v_t^\sharp,\mathbf b_j\rangle\in\mathbb Z.
$
Moreover,
\[
\left|
\langle
\widetilde{\mathbf v}_t-\mathbf v_t^\sharp,
\mathbf b_j
\rangle
\right|
\leq
\delta\|\mathbf b_j\|_2
<\frac12.
\]
Therefore rounding
\(\langle\widetilde{\mathbf v}_t,\mathbf b_j\rangle\)
to the nearest integer recovers the exact pairing
$
k_{t,j}
=
\langle\mathbf v_t^\sharp,\mathbf b_j\rangle.
$
Let
$
B=[\,\mathbf b_1\ \cdots\ \mathbf b_\ell\,]
$
be the matrix of projected vectors. 
Compute an HNF basis matrix \(H\) of the lattice generated by the projected vectors, and let \(C\) be an integer matrix such that
$
H=BC.
$
Since the exact pairings \(\mathbf k_t=(k_{t,1},\ldots,k_{t,\ell})^T\) are obtained by rounding against the projected vectors before this change of basis, the pairings with the HNF basis are
$
H^T\mathbf v_t^\sharp=C^T\mathbf k_t.
$
Solving
$
H^T\mathbf x_t=C^T\mathbf k_t
$
over the rationals and reducing \(\mathbf x_t\) modulo \(\mathbb Z^m\) recovers the exact class
$
\mathbf v_t=(\mathbf p_t,\mathbf q_t)\in\mathcal L^*/\mathbb Z^m.
$

As a consistency check, verify that the reconstructed phase reproduces the rounded pairings on all projected vectors
$
B^T\mathbf x_t=\mathbf k_t.
$
If the projected vectors do not have rank \(m\), or if this equality fails for any \(1\leq t\leq Q\), the reconstruction is regarded as unsuccessful. Otherwise, let
\(\mathcal L_{1}\) denote the recovered integral lattice and write
$
(\widehat{\mathbf p}_t,\widehat{\mathbf q}_t)
$
for the reconstructed exact phases.

For every fixed evaluation sequence satisfying the short-basis condition, the recovered lattice and all \(Q\) exact phases are correct except with probability \(O(n^{-c})\). Indeed, the sampling coupling contributes \(O(Qn^{-(c+1)})=O(n^{-c})\), while the separation event contributes \(O(2^{-(q-m)})=O(n^{-c})\).

\paragraph{Finite coordinates and character extension.}
Put
\(\Delta:=\det\mathcal L=|\operatorname{im}\Phi|.\)
Since \(w_l\in\operatorname{im}\Phi\), the order of \(w_l\) divides
\(\Delta\), and hence
\(\Delta w_l=0_G, (1\leq l\leq b).\)
Define
\[
F:=
\left(\prod_{i=1}^{r}\mathbb Z_{N_i}\right)
\times
\left(\prod_{i=1}^{s}\mathbb Z_{M_i}\right)
\times
\mathbb Z_{\Delta}^{\,b}.
\]
The homomorphism
\(\Theta:F\longrightarrow H_{k+1}\)
is given by
\(\Theta(\mathbf x,\mathbf z',\mathbf u) = \sum_{i=1}^{r}x_i y_i + \sum_{i=1}^{s}z'_i y'_i + \sum_{l=1}^{b}u_lw_l.\)
It is well defined because
\(N_i y_i=0_G\), \(M_i y'_i=0_G\), and
\(\Delta w_l=0_G\). Its image contains \(H^{-}\), \(H^{+}\), and
the target elements \(w_1,\ldots,w_b\), so
\(\operatorname{im}\Theta=H_{k+1}.\)
Let
\(K_\Theta:=\ker\Theta.\)
Then
\(F/K_\Theta\cong H_{k+1}.\)

For an exact joint character
\((\boldsymbol\eta_t,\mathbf p_t,\mathbf q_t)\in P^\perp,\)
consider the congruences
\begin{equation}\label{eq:character-pullback}
\sum_{i=1}^{s}
\frac{\alpha_{j,i}\xi_{t,i}}{M_i}
\equiv
p_{t,j}
-\sum_{l=1}^{b}\beta_{j,l}q_{t,l}
\pmod1,
\qquad
1\leq j\leq d,
\end{equation}
where
\(\boldsymbol\xi_t=(\xi_{t,1},\cdots,\xi_{t,s})\in\prod_{i=1}^{s}\mathbb Z_{M_i}.\)
Under the generation condition
\(\langle a_1,\ldots,a_d\rangle=H^{+}\), the next lemma shows that
this system has a unique solution and that
\(\Delta q_{t,l}\in\mathbb Z\) for every \(l\). In that case define
\begin{equation}\label{eq:character-row}
\mathbf a_t
:=
\left(
\boldsymbol\eta_t,
\boldsymbol\xi_t,
\Delta\mathbf q_t
\right),
\end{equation}
where the three blocks are reduced modulo \(N_i\), \(M_i\), and
\(\Delta\), respectively.

\begin{lemma}
\label{lem:annihilator-isomorphism}
Assume that
\(\langle a_1,\ldots,a_d\rangle=H^{+}.\)
For every
\((\boldsymbol\eta,\mathbf p,\mathbf q)\in P^\perp\),
Equation~\eqref{eq:character-pullback} has a unique solution
\(\boldsymbol\xi\). Moreover,
\(\Delta q_l\in\mathbb Z, (1\leq l\leq b),\)
and the vector in Equation~\eqref{eq:character-row} defines an
element of \(K_\Theta^\perp\). The resulting map
\(\tau:P^\perp\longrightarrow K_\Theta^\perp\)
is a group isomorphism. Hence a uniform character in \(P^\perp\)
is mapped to a uniform character in \(K_\Theta^\perp\).
\end{lemma}

\begin{proof}
The definitions of \(a_j\) and \(h_j\) give
\(\langle h_1,\ldots,h_d,w_1,\ldots,w_b\rangle = \langle a_1,\ldots,a_d,w_1,\ldots,w_b\rangle.\)
Under the generation hypothesis,
\(\operatorname{im}\Psi = \langle H^{-},H^{+},w_1,\ldots,w_b\rangle = H_{k+1}.\)
Thus both
\(\Psi:\mathcal G\to H_{k+1} \) and \(\Theta:F\to H_{k+1}\)
are surjective. Pullback of characters gives isomorphisms
\(\Psi^*:\widehat{H_{k+1}}\longrightarrow P^\perp\)
and
\(\Theta^*:\widehat{H_{k+1}}\longrightarrow K_\Theta^\perp.\)

Fix
\((\boldsymbol\eta,\mathbf p,\mathbf q)\in P^\perp\)
and let
\(\chi:=(\Psi^*)^{-1} (\boldsymbol\eta,\mathbf p,\mathbf q),\)
written in additive phase notation. Then, in additive phase notation,
\(\chi(y_i)\equiv \eta_i/N_i\pmod1\),
\(\chi(h_j)\equiv p_j\pmod1\) and
\(\chi(w_l)\equiv q_l\pmod1\).
Write
\(\chi(y'_i)\equiv\frac{\xi_i}{M_i}\pmod1.\)
Since
\(a_j=h_j-\sum_{l=1}^{b}\beta_{j,l}w_l,\)
we obtain
\[
\sum_{i=1}^{s}
\frac{\alpha_{j,i}\xi_i}{M_i}
\equiv
\chi(a_j)
\equiv
p_j-\sum_{l=1}^{b}\beta_{j,l}q_l
\pmod1.
\]
Because the
\(a_j\) generate \(H^{+}\), Lemma~\ref{lem:character-extension}
gives a unique value of \(\boldsymbol\xi\).

For each target element,
\(\Delta w_l=0_G.\)
Hence
\(0 \equiv \chi(\Delta w_l) \equiv \Delta q_l \pmod1,\)
so \(\Delta q_l\in\mathbb Z\). Its residue modulo \(\Delta\) is
the character coordinate associated with the \(l\)-th
\(\mathbb Z_\Delta\) factor of \(F\).

The three blocks in Equation~\eqref{eq:character-row} are exactly
the cyclic coordinates of the pullback character
\(\chi\circ\Theta\). Therefore
\(\mathbf a=\Theta^*(\chi),\)
and the coordinate construction realizes
$
\tau
=
\Theta^*\circ(\Psi^*)^{-1}.
$
Thus \(\tau\) is a group isomorphism. Since it is bijective, it
maps the uniform distribution on \(P^\perp\) to the uniform
distribution on \(K_\Theta^\perp\).
\end{proof}

For the recovered lattice \(\mathcal L_1\), put
\(\Delta_{1}:=\det\mathcal L_{1}.\)
First verify
\(\Delta_{1} w_l=0_G , (1\leq l\leq b).\)
Use the measured cyclic labels and the reconstructed exact phases
in Equation~\eqref{eq:character-pullback}. For each \(t\), check
that the system has a unique solution
\(\widehat{\boldsymbol\xi}_t\) and that
\(\Delta_{1}\,\widehat q_{t,l}\in\mathbb Z , (1\leq l\leq b).\)
If any of these checks fails, the reconstruction is regarded as
unsuccessful. Otherwise define
$
\widehat{\mathbf a}_t
=
\left(
\widehat{\boldsymbol\eta}_t,
\widehat{\boldsymbol\xi}_t,
\Delta_{1}\,\widehat{\mathbf q}_t
\right),
$
with each block reduced modulo its corresponding cyclic order.
All these steps use exact rational and integer arithmetic. The
quantities
\(\widehat p_{t,j} -\sum_{l=1}^{b}\beta_{j,l}\widehat q_{t,l}\)
are formed only after exact phase recovery; applying the
coefficients \(\beta_{j,l}\) directly to the noisy measured phases
would not preserve the preceding error bound.

Fix a coefficient sequence satisfying the generation and short-basis conditions. Let
$
(\boldsymbol\eta_t,\mathbf p_t,\mathbf q_t)\in P^\perp,
 (1\le t\le Q)
$
be the independent uniform ideal characters coupled to the \(Q\) quantum executions. By Lemma~\ref{lem:annihilator-isomorphism}, the corresponding character rows from Equation~\eqref{eq:character-row} satisfy
$
\mathbf a_t
=
\tau(\boldsymbol\eta_t,\mathbf p_t,\mathbf q_t),
$ and hence
$
(\mathbf a_1,\ldots,\mathbf a_Q)
\sim U(K_\Theta^\perp)^{\otimes Q}.
$
Let \(E_{\mathrm{row}}\) denote the event that the reconstruction is
successful,
\(\mathcal L_{1}=\mathcal L,\)
and
\(\widehat{\mathbf a}_t=\mathbf a_t ,  (1\leq t\leq Q).\)
The sampling and exact-phase recovery bounds established above give
\begin{equation}\label{eq:joint-row-error}
\Pr(E_{\mathrm{row}}^c)=O(n^{-c}).
\end{equation}
On the generation event, Equation~\eqref{eq:local-determinant} gives
$
\Delta=|S|.
$
Hence the groups \(F\), \(K_\Theta=\ker\Theta\), and \(K_\Theta^\perp\) are determined by the current subgroup and target batch, independently of the particular coefficient sequence satisfying the generation condition. Therefore the bound in Equation~\eqref{eq:joint-row-error} holds uniformly over all such coefficient sequences.

\paragraph{Kernel recovery and Smith normal form.}
Since \(K_\Theta^\perp\leq\widehat F\), its minimum number of generators
is at most
\(r+s+b\leq r+m.\)
The number of ideal character rows is
\(Q=r+m+\lceil c\log_2n\rceil.\)
Hence Corollary~\ref{cor:generation-probability} shows that
\[
\Pr\!\left[
\langle\mathbf a_1,\ldots,\mathbf a_Q\rangle
\neq K_\Theta^\perp
\right]
=
O(n^{-c}).
\]
Combining this bound with
Equation~\eqref{eq:joint-row-error} and the coefficient-generation
estimates, we obtain, except with probability \(O(n^{-c})\),
the correct coordinate group \(F\) and
\(\langle \widehat{\mathbf a}_1,\ldots,\widehat{\mathbf a}_Q \rangle = K_\Theta^\perp.\)
On this event,
\[
\begin{aligned}
K_\Theta
=
\bigcap_{\chi\in K_\Theta^\perp}\ker\chi
=
\left\{
\mathbf v\in F:
\widehat{\mathbf a}_t(\mathbf v)
\equiv0\pmod1
\text{ for every }1\leq t\leq Q
\right\}.
\end{aligned}
\]

Write the computed coordinate group as
$
F_{\mathrm{rec}}
=
\prod_{j=1}^{p}\mathbb Z_{\nu_j},
p=r+s+b,
$
where the periods are the \(N_i\), the \(M_i\), and \(b\) copies
of \(\Delta_{1}\). Put
\(\nu:=\operatorname{lcm}(\nu_1,\ldots,\nu_p)\)
and define \(A'\in\mathbb Z^{Q\times p}\) by
\(A'_{t,j}
:=
\frac{\nu}{\nu_j}\widehat a_{t,j}.\)
Then, for \(\mathbf z\in\mathbb Z^p\), the phase of the \(t\)-th
character is
\(\frac{(A'\mathbf z)_t}{\nu} \pmod1.\)
Thus the inverse image in \(\mathbb Z^p\) of the common kernel is
$
\Lambda
:=
\left\{
\mathbf z\in\mathbb Z^p:
A'\mathbf z
\equiv\mathbf0\pmod\nu
\right\}.
$

Consider the integer kernel
$
\ker_{\mathbb Z}
\bigl[
A'\mid-\nu I_Q
\bigr].
$
A pair \((\mathbf z,\mathbf y)\) belongs to this kernel precisely
when
\(A'\mathbf z=\nu\mathbf y.\)
Hence projection onto the first \(p\) coordinates is an
isomorphism from this integer kernel onto \(\Lambda\), since
\(\mathbf y=A'\mathbf z/\nu\) is uniquely determined
by \(\mathbf z\). Compute an integer basis of this kernel, project
it onto the first \(p\) coordinates, and use HNF to obtain a basis
matrix of \(\Lambda\).

For every \(j\),
\(\nu_j\mathbf e_j\in\Lambda,\)
because
\(A'(\nu_j\mathbf e_j) \equiv\mathbf0\pmod\nu.\)
Let
\(\rho:\mathbb Z^p\longrightarrow F_{\mathrm{rec}}\)
be coordinate reduction. On the successful reconstruction event,
\(F_{\mathrm{rec}}=F\) and
\(\Lambda=\rho^{-1}(K_\Theta).\)
The First Isomorphism Theorem then gives
$
\mathbb Z^p/\Lambda
\cong
F/K_\Theta
\cong
H_{k+1}.
$

Apply Smith normal form to a column basis matrix of \(\Lambda\).
Its nontrivial invariant factors are the invariant factors of
\(H_{k+1}\). If \(U\) denotes the left unimodular transformation,
the corresponding columns of \(U^{-1}\) give the cyclic generators
as integer combinations of
\(y_1,\ldots,y_r,\; y'_1,\ldots,y'_s,\; w_1,\ldots,w_b.\)
These combinations are evaluated using the classical black-box
group operations.

After this computation, check that the product of the nontrivial
invariant factors is at most \(2^n\). If this bound fails, the
reconstruction is regarded as unsuccessful. A correct reconstruction
always passes this check because
\(|H_{k+1}|\leq|G|\leq2^n.\)
The integer matrices above have polynomial dimension and polynomial
bit length, so the computation requires polynomially many classical
bit operations.

\subsection{Completion of the Proof}
\label{sec:decomposition-proof}

\begin{proof}[Proof of Theorem~\ref{thm:group-decomposition}]
Fix \(c\geq2\) and the parameters of
Section~\ref{sec:subgroup-extensions}. Initialization uses the first
\(d\) candidates. The remaining candidates are processed in
\(k_* = \left\lceil\frac{L-d}{b}\right\rceil = O(\sqrt n)\)
extensions, with zero padding in the final batch when necessary.

\paragraph{Correctness and success probability.}
The candidate list fails to generate \(G\) with probability
\(O(n^{-c})\), and the initialization step fails with probability
\(O(n^{-c})\). Consider an extension for which all preceding
subgroup presentations are correct. Conditional on this history and
on the next target batch, the coefficient estimates of
Section~\ref{sec:extension-setup} are unchanged. The short-basis
bound, the generation estimate for the auxiliary elements,
Equation~\eqref{eq:joint-row-error}, and
Corollary~\ref{cor:generation-probability} imply that this extension
fails with probability \(O(n^{-c})\).

Let \(E\) be the event that the candidate list generates \(G\), the
initialization is correct, and every subgroup extension is correct.
Summing the conditional failure probabilities over the
\(k_*=O(\sqrt n)\) extensions gives
$
\Pr(E^c)
=
O(n^{1/2-c}).
$
No independence between different extensions is required for this
estimate.

On event \(E\),
\(H_{k+1} = \langle H_k,w_1,\ldots,w_b\rangle\)
at every extension, and therefore the final subgroup is
\(H_{k_*} = \langle\mathcal S\rangle = G.\)
The Smith normal form computation at each stage gives the invariant
factors and corresponding cyclic generators. Since \(c\geq2\),
\(n^{1/2-c}=O(n^{-3/2}),\)
which proves the stated success probability.

\paragraph{Space and gates per execution.}
After each Smith normal form computation, the algorithm checks that
the product of the nontrivial invariant factors is at most \(2^n\).
Thus, on every execution that proceeds to the next extension, the
current cyclic orders \(N_i\geq2\) satisfy
$
\prod_i N_i\leq2^n.
$
Consequently,
\[
\sum_i\lceil\log_2N_i\rceil
\leq
2\sum_i\log_2N_i
=
2\log_2\!\left(\prod_iN_i\right)
\leq2n.
\]
Hence the cyclic-coordinate registers use \(O(n)\) qubits, even if
the current decomposition resulted from an earlier incorrect
reconstruction that passed the size check.

The number of cyclic factors is at most \(n\). At most \(\sqrt n\)
of them can have order greater than
\(B=2^{\sqrt n}\). Each cyclic coordinate of \(H^{-}\) uses at most
\(\sqrt n+1\) qubits. The lattice dimension is
\(m=d+b=O(\sqrt n)\), and
\(\log D=O(\sqrt n)\), so the integer-coordinate registers and the
registers storing the Fibonacci coefficients $\epsilon_{i,j}$ each use
\(O(n)\) qubits. The group
accumulators, slice buffer, oracle-output register, and work
registers for the group operations require another \(O(n)\) qubits.
The auxiliary space for recoding and the exact cyclic Fourier
transforms is \(O(\sqrt n)\) and is reused across cyclic factors.
Thus one circuit execution uses \(O(n)\) qubits.

For the cyclic part of \(\Psi\), binary expansion of the
coordinates requires at most
\(2\sum_{i=1}^{r}\lceil\log_2N_i\rceil \leq4n\)
controlled group additions when computation and uncomputation are
both counted. Their gate cost is \(O(nT_{\mathrm{op}})\).

For the integer-coordinate part, Lemma~\ref{lem:fibonacci-evaluation}
uses
\(O((m+1)J)=O(n)\)
group operations, since
\(m,J=O(\sqrt n)\). The reversible recoding and its inverse require
\(\widetilde O(mJ\log D) = \widetilde O(n^{3/2})\)
elementary gates, and register exchanges require
\(O(nJ)=O(n^{3/2})\) gates. The exact Fourier operations on the
cyclic coordinates and the Gaussian preparation and lattice Fourier
transform contribute \(\widetilde O(n^{3/2})\) further gates.
Therefore one execution uses
$
\widetilde O
\left(
nT_{\mathrm{op}}+n^{3/2}
\right)
=
\widetilde O(nT_{\mathrm{op}}),
$
where the last equality uses
\(T_{\mathrm{op}}=\Omega(\sqrt n)\).
The initialization circuit satisfies the same bound.

\paragraph{Circuit count and quantum time.}
There is one initialization circuit and one circuit for each subgroup
extension, hence \(O(\sqrt n)\) circuits in total. Since each circuit
has
\(\widetilde O(nT_{\mathrm{op}})\) gates, their gate counts sum to
$
\widetilde O(n^{3/2}T_{\mathrm{op}})
$
when each distinct circuit is counted once.

Initialization uses \(O(\sqrt n)\) executions. Each extension uses
\(Q=O(r+\sqrt n)=O(n)\)
executions. Since there are \(O(\sqrt n)\) extensions, the total
quantum time is
$
\widetilde O(n^{5/2}T_{\mathrm{op}}).
$

\paragraph{Classical computation.}
The sampled coefficients have \(2n\) bits. Constructing the
auxiliary elements, precomputing fixed multiples, and evaluating the
final generator combinations require polynomially many classical
group-operation queries. The augmented lattice has
\(O(\sqrt n)\) dimensions and entries of polynomial bit length.
The numerical checks above bound every cyclic period used in a
subsequent quantum circuit by \(2^n\).

The character-extension systems and the integer matrices used for
kernel computation therefore have polynomial dimension and
polynomial bit length. Polynomial-time LLL, Hermite normal form,
and Smith normal form algorithms compute the required bases and
transformation matrices in polynomially many bit operations.
The same bounds apply to initialization. Since the algorithm performs
\(O(\sqrt n)\) subgroup extensions, the total classical computation
uses polynomially many bit operations and group-operation queries.
\end{proof}

\paragraph{AI statement:}
The original theorems and lemmas in this paper, including
their proofs, were derived by the authors.
ChatGPT (OpenAI) was used to help check the correctness of
the short-basis proofs in Section~\ref{sec:short-bases}.
It was also used to check the distributional arguments in
Section~\ref{sec:main-theorem-proof}, including the proofs of
Lemmas~\ref{lem:joint-dual-sampling},
\ref{lem:uniform-projection},
and~\ref{lem:annihilator-isomorphism}.
It also assisted in checking the gate counts and qubit
requirements of individual quantum circuits, the number
of circuit executions, and the resulting total quantum time.
The initial manuscript was written by the authors.
ChatGPT was also used to help correct minor errors in
intermediate drafts, resolve inconsistencies in notation
and references, and improve the language.
The authors reviewed the revisions and take responsibility
for the final manuscript.

\clearpage
\bibliographystyle{alpha} 
\bibliography{references}
\end{document}